\documentclass[10pt]{article}
\usepackage{amssymb}
\usepackage{amsfonts}
\usepackage{amsthm}
\usepackage{amsmath}
\usepackage{bm}	
\usepackage{bbm}
\usepackage{upgreek}
\usepackage{hyperref}
\usepackage{longtable}
\usepackage{tabu}

\title{Single Scale Cluster Expansions with Applications to Many Boson and Unbounded Spin Systems}
\author{Martin Lohmann \\ \textit{{\small  Department Mathematik, ETH Z\"urich }}}
\date{}
\allowdisplaybreaks

\newcommand{\Text}[1]{\textnormal{#1}}
\newcommand{\ud}{\,\mathrm{d}}

\newcommand{\erem}{ \begin{flushright} \(\diamond\)  \end{flushright}}
\newcommand{\bd}[1]{\mathbf{#1}}
\newcommand{\mc}[1]{\mathcal{#1}}
\newcommand{\mf}[1]{\mathfrak{#1}}
\newcommand{\half}[0]{\frac{1}{2}}
\newcommand{\id}[0]{\mathbbm{1}}
\newcommand{\pderi}[1]{\frac{\partial}{\partial #1}}
\newcommand{\ol}[1]{\overline{#1}}
\newcommand{\ul}[1]{\underline{#1}}
\newcommand{\mb}[1]{\mathbb{#1}}
\newcommand{\qqquad}{\qquad\qquad\qquad\qquad\qquad}
\newcommand{\const}{\Text{ const }}

\newcommand{\ob}[1]{^{(#1)}}
\newcommand{\supp}{\textnormal{supp }}
\newcommand{\N}{\ul{\sf N}}

\theoremstyle{plain}
\newtheorem{lem}{Lemma}
\newtheorem{thm}{Theorem}
\newtheorem{prop}{Proposition}

\theoremstyle{definition}
\newtheorem{defin}{Definition}

\theoremstyle{remark}
\newtheorem{rem}{Remark}

\begin{document}

\maketitle

\begin{abstract}
We develop a cluster expansion to show exponential decay of correlations for quite general single scale spin systems, as they arise in lattice quantum field theory and discretized functional integral representations for observables of quantum statistical mechanics. We apply our results to: the small field approximation to the coherent state correlation functions of the grand canonical Bose gas at negative chemical potential, constructed by Balaban, Feldman, Kn\"orrer and Trubowitz in \cite{BFKT4}; and to $\sf N $ component unbounded spin systems with repulsive two body interaction and massive, possibly complex, covariance. Our cluster expansion is derived by a single application of the BKAR interpolation formula.
\end{abstract}

\section{Introduction}

Cluster expansions are a well-known tool to prove critical behavior and decay properties of truncated correlation functions in spin systems from statistical mechanics and quantum field theory. For multi scale systems, they are part of many rigorous implementations of the renormalization group, where they give detailed control over the exponents of the power law decay of the correlation functions. For single scale systems, they can be used to prove exponential decay of correlations. \\
In field theoretic language, most systems that qualify as ``single scale'' are defined by an action such that either: the single site (``diagonal'') contributions to the action are uniformly bounded below, and the multi site (``off-diagonal'') contributions are small as compared to the diagonal ones; or: the off-diagonal contributions which are not small are Gaussian (quadratic) and given by a uniformly bounded and exponentially decaying covariance, and all other contributions to the action (including (nonquadratic) diagonal ones) are small as compared to the Gaussian contribution. \\
Exponential decay of the two point function for the first type of systems has been investigated intensively using Witten Laplacian techniques, see for example \cite{Sj,H,BM}. In somewhat less generality (but in more detail), decay of all truncated correlation functions has also been studied with cluster expansions \cite{PS}. These works were mainly motivated by the link between exponential decay of correlations to logarithmic Sobolev inequalities (see, e.g. \cite{L,Y,GZ}). \\
Even though the methods of this thesis imply strong results in these directions, we have optimized their implementation for application to systems of the second type, which have only been accessible to cluster expansions. The structure of the expansion in this situation is more complicated than for the first type, and relies on the well-known integration by parts identity for Gaussian measures. Most classical versions of such cluster expansions are designed to treat single steps of a renormalization group iteration for multi scale systems, see e.g. \cite{MP,BY,Ri,Br,AR2}. An example of a single scale system whose decay properties determine the time evolution of quantum fluids \cite{LS} was solved in \cite{APS} and is a special case of the models treated in this thesis. \\
An interesting single scale system more complicated than the ones in the above applications is the Bose gas with strictly negative chemical potential $\mu<0 $. Its lattice version is defined by the Hamiltonian $$ H = \sum_{\bd x,\bd y\in\mb T} \psi^\dagger (\bd x) \big[-\Delta-\mu\big](\bd x,\bd y) \psi(\bd y) + \half\sum_{\bd x,\bd y\in\mb T} \psi^\dagger(\bd x)\psi^\dagger(\bd y) v(\bd x,\bd y)\psi(\bd x)\psi(\bd y), $$ where $\mb T = \mb Z^D/L\mb Z^D,L\in\mb N, $ is the discrete torus, $\Delta $ is the lattice Laplacian, and $ v$ is a small, repulsive, and exponentially decaying two body potential. $\psi(\bd x),\psi^\dagger(\bd x) $ are canonical creation and annihilation operators on Fock space. Standard arguments \cite{BR} (applicable also in the continuum model) using a Feynman Kac functional integral representation and based on the repulsiveness of the interaction and positivity of the heat kernel can be used to show that the untruncated correlation functions are pointwise dominated by their value at $v=0$. Since the free system is uncondensed at $\mu<0 $, this, in particular, gives exponential decay of the two point function. See \cite{US} for a recent application of this to the critical temperature of Bose Einstein condensation.\\
The arguments above have no direct implications on higher order truncated correlation functions, and, more importantly, the low temperature critical behavior of the model as the chemical potential becomes positive is obscure in the Feynman Kac type representation for the correlations. A more intuitive representation that explains the Bogoliubov picture of Bose Einstein condensation by a simple mean field argument and has long been used by physicists to study this critical behavior is the one by coherent state functional integrals (see, e.g. \cite{NO}). This representation has been made rigorous in \cite{BFKT1}, and leads (at fixed infrared and ultraviolet space cut-off) to a superrenormalizable (as a high temperature cut-off is removed) theory, which was subsequently constructed in \cite{BFKT4}. This construction is done at large finite temperature, small chemical potential of either sign, small interaction strength, and uniformly in the infrared cut-off.\\
Proving exponential decay of correlations in the resulting effective model at $\mu<0 $ should not be expected as simple as in the Feynman Kac representation. On the most obvious technical level, one has to deal with a complicated ``large field small field expansion'' that is the output of \cite{BFKT4} (and whose precise form is tailored towards the more exciting $\mu>0 $ phase). But even the pure small field term in this expansion cannot be treated by methods that only rely on the repulsiveness of the interaction because of the oscillatory nature of coherent state functional integrals.\\
On the other hand, this pure small field term is a natural candidate for an analysis by a cluster expansion. Unfortunately, due to their quite technical nature, cluster expansions have never been formulated on the level of generality necessary for an application to this situation. In particular, complex covariance matrices, non-local repulsive interactions, together with non-polynomial interactions (in an appropriate small field domain) have not yet been analyzed by cluster expansions. It is the motivation of this thesis to achieve this.\\ 
We have tried to find a formalism in which the logic of the argument is not buried too deeply under its technicalities. We have also tried to assure that this formalism can naturally be extended to a multi scale setting. The cluster expansion described in this thesis converges for chemical potentials of the order $\mu\lesssim -\mf v^{\frac1{4+D}} $, where $\mf v\sim \Vert v\Vert_{1,\infty} $ is roughly the size of the interaction. For technical reasons related to the oscillatory nature of coherent states, this is not where intuition suggests the appearance of new scales, namely at $ \mu\sim -g^{\frac 12}$. To get to this point and beyond, multi scale expansions should be useful. In particular, they should prove exponential decay of correlations at $0\leq \mu\ll g $ and positive temperature, and critical (algebraic) decay at $\mu=0 $ and zero temperature. Further work on this is in progress.\\
In the remainder of this introductory section, we describe in more detail the models for which our cluster expansion was developed. In section 2, the algebra of the expansion is described. In section 3, we formulate the framework for the norms that we will use in section 4 to control the expansion. Section 5 defines explicit norms to which this framework applies. In section 6, we give details about how the bounds of section 4 apply to our models. For the convenience of the reader, we have collected tables of symbols and notation in an appendix.

\subsection{Small field many Boson systems}

We recall some of the results from \cite{BFKT4} that are relevant to the present work. They concern the small field part of their construction only and are also the content, in a slightly different form, of the more accessible paper \cite{BFKT3}.\\
Let $(X,d) $ be a finite metric space, and $\text{h} $ be a real symmetric operator on $L^2(X) $; in \cite{BFKT4}, $X=\mb T $ is the $D$ dimensional discrete torus, and $\text{h} $ is a slight generalization of (minus) the lattice Laplacian. Let the two body potential $v(\bd x,\bd y) $ be a real symmetric, exponentially decaying (with fixed positive mass $5{\sf m}$) and repulsive operator on $L^2(X) $; in \cite{BFKT4}, it is also assumed translation invariant. Let $\mu\in\mb R $ be the chemical potential. Depending on $\text{h}$, $\mu$ and ${\sf m}$, let $\theta>0 $ be small enough\footnote{See Remark \ref{remallowedmu} below for an important comment on this smallness condition}. Let $\beta \in \theta\mb N $ be an inverse temperature. Set $\mb L =  \theta\mb Z/\beta\mb Z \times X $. Denote by $\psi^{(\dagger)}(\bd x)$ canonical (creation and) annihilation operators acting on the Fock space $\mc F= \oplus_{n\geq 0} L^2(X)^{\otimes_S n} $, and, for $ x = (\tau,\bd x)\in \mb L $, set $ \psi^{(\dagger)}(x) = e^{\tau H_\mu}\psi^{(\dagger)}(\bd x) e^{-\tau H_\mu} $ with
\begin{align*}
H_\mu &= H-\mu N = \sum_{\bd x,\bd y\in X} \psi^\dagger(\bd x) \big[h(\bd x,\bd y)-\mu\delta_{\bd x,\bd y}\big]\psi(\bd y) + \half\sum_{\bd x,\bd y\in X} \psi^\dagger(\bd x)\psi^\dagger(\bd y) v(\bd x,\bd y)\psi(\bd x)\psi(\bd y).
\end{align*}
For a kernel $A(\bd x,\bd y)$ on $X$, denote 
\begin{align}\label{notnorms}
\Vert A\Vert_{\sf m} &= \sup_{\bd x}\sum_{\bd y} e^{{\sf m}d(\bd x,\bd y)}\vert A(\bd x,\bd y)\vert\\
\Vert A\Vert_{{\sf m},\infty} &= \sup_{\bd x,\bd y} e^{{\sf m}d(\bd x,\bd y)}\vert A(\bd x,\bd y)\vert\nonumber
\end{align}
Assume that $0<\Vert v\Vert_{5{\sf m}}=:\mf v $ is small enough (but cf. Remark \ref{remallowedmu}), and that the smallest eigenvalue $\lambda_{\text{min}}(v)\geq c_v\mf v $, $c_v>0 $. The result that was motivated in \cite{BFKT3} and proven in \cite{BFKT4} is that then, for any $x_1,\ldots,x_n\in \mb L $, the unnormalized, untruncated correlation functions are approximately given by ($\mb T $ denotes the usual ``time ordering'')
\begin{align}\label{fnintrep}
\text{Tr }e^{-\beta H_\mu} \mb T \prod_{m=1}^n \psi^{(\dagger)}(x_m) &\approx \int_{\mb C^{\mb L}} \prod_{x\in \mb L}\ud\phi(x) \chi\big(\phi\big)e^{\mc A(\phi)} \prod_{m=1}^n \phi(x_m)^{(*)}
\end{align}
``Approximately'' means ``modulo large field contributions'', see \cite{BFKT4} for their complicated description. In the above integral, $\ud\phi(x) = \frac1\pi \ud\Re\phi(x)\ud\Im\phi(x) $ is the standard Lebesgue measure. The domain of integration is determined by the characteristic function $\chi(\phi) $. This function is equal to $1$ iff the following three conditions are satisfied at every $x\in\mb L $:
\begin{align}\nonumber
\vert \phi(x)\vert &\leq \big(\theta\mf v\big)^{-\frac 14-\epsilon} &&\\
\vert \partial_0\phi(x)\vert &\leq   \big(\theta\mf v\big)^{-\epsilon} && \partial_0\phi(\tau,\bd x) = \phi(\tau+\theta,\bd x)-\phi(\tau,\bd x) \label{originalchi}\\
\vert \nabla\phi(x)\vert &\leq  \theta^{-\half}\big(\theta\mf v\big)^{-\epsilon} && \nabla_i\phi(\tau,\bd x) = \phi(\tau,\bd x+\bd e_i)-\phi(\tau,\bd x).\nonumber
\end{align}
The last condition is absent in \cite{BFKT3} and was added in \cite{BFKT4} in view of a potential infrared analysis in the $\mu>0 $ regime (see Remark \ref{remallowedchi} below). The action $\mc A $ defining the integrand is
\begin{align*}
\mc A(\phi) &= -\big\langle \phi^*,{\sf C}^{-1}\phi\big\rangle + V(\phi) + \mc D(\phi).
\end{align*}
The explicit quadratic contribution is defined through $\big\langle \phi^*,A\phi\big\rangle = \sum_{x,y\in\mb L}\phi(x)^*A(x,y)\phi(y) $ by 
\begin{align}\label{mbbfktp}
{\sf C}^{-1}\big((\tau,\bd x),(\tau',\bd x')\big) &= -\delta_{\tau+\theta,\tau'}j(\theta)(\bd x,\bd y) + \delta_{(\tau,\bd x),(\tau',\bd x')}. 
\end{align}
Here, for any $\tau\geq 0 $, $j(\tau)=e^{-\tau(\text{h}-\mu)} $. Note that ${\sf C}^{-1} $ is not symmetric. If $X=\mb T $ is the torus, $\text{h}$ is translation invariant with positive, analytic Fourier transform, and $\mu <0 $, it is standard to check that ${\sf C}^{-1}$ is normal, has spectrum $\sigma({\sf C}^{-1})\subset \{\Re z\geq 1-e^\mu\} $, and that both ${\sf C} $ and ${\sf C}^{-1} $ are exponentially decaying in the space and ``time'' direction; in particular, with constants $c_t,c_x>0 $ depending on $\text{h} $ and $\theta $, $$ {\sf C}\big((\tau,\bd x),(\tau',\bd x')\big)\sim e^{-c_t\mu \vert \tau-\tau' \vert - c_x \sqrt{-\mu} \vert\bd x-\bd x'\vert_2}. $$ We assume that ${\sf m}$ above was chosen bigger than ($\frac16$ of) the corresponding decay rate.\\
The explicit quartic part of the interaction is
\begin{align*}
V(\phi) &= -\sum_{x_1,\ldots,x_4\in\mb L} w(x_1,\ldots,x_4)\phi(x_1)^*\phi(x_2)\phi(x_3)^*\phi(x_4)\\
w\big((\tau_1,\bd x_1),\ldots,(\tau_4,\bd x_4)\big) &= \delta_{\tau_1,\tau_3}\delta_{\tau_1+\theta,\tau_2}\delta_{\tau_3+\theta,\tau_4} \sum_{\bd x,\bd y }v(\bd x,\bd y)\\&\qquad\qquad\times \int_0^\theta\ud t\; j(t)(\bd x,\bd x_1)j(t)(\bd x,\bd x_2)j(t)(\bd y,\bd x_3)j(t)(\bd y,\bd x_4)
\end{align*}
Its kernel is exponentially decaying in the space direction and ``nearest neighbor'' in the ``time'' direction. \\
Finally, the non-explicit part $\mc D $ of the action is an even power series (starting with quartic term) that is again ``nearest neighbor'' in the ``time'' direction:
\begin{align*}
\mc D(\phi) &= \sum_{\tau\in \theta\mb Z/\beta\mb Z}\sum_{n\geq 2} \sum_{\substack{\bd x_1,\ldots,\bd x_n \\ \bd y_1,\ldots,\bd y_n}} \mc D\big(\bd x_1,\ldots,\bd x_n;\bd y_1,\ldots,\bd y_n\big)\prod_{m=1}^n \phi\big((\tau-\theta,\bd x_m)\big)^*\phi\big((\tau,\bd y_m)\big)
\end{align*}
This power series is convergent if the small field condition $\vert \phi(x)\vert\leq \big(\theta\mf v\big)^{-\frac 14-\epsilon}=: R $ is satisfied everywhere because, as always for $\mf v $ small enough, we have the bound $$ \Vert \mc D\Vert_{2R,2{\sf m}}\leq \const \big(\theta\mf v\big)^{\frac12-8\epsilon}   $$ with the $1,\infty \;+ $ tree decay norm
\begin{align}\label{originalnorm}
\Vert \mc D\Vert_{R,{\sf m}} &= \sum_{n\geq 2} \max_{\bd x\in X}\max_{1\leq m\leq 2n} \sum_{\substack{\bd x_1,\ldots,\bd x_{2n} \\ \bd x_m = \bd x }} R^{2n}e^{{\sf m}d_{\text{t}}(\bd x_1,\ldots,\bd x_{2n}) } \big\vert d\big(\bd x_1,\ldots,\bd x_n;\bd x_{n+1},\ldots,\bd x_{2n}\big)\big\vert.
\end{align}
Here, for any set or sequence $S $ of points of $X$, $d_{\text{t}}(S) $ is the minimal size of a tree on $S$. Note already here that the corresponding norm of the explicit quartic part $\Vert w\Vert_{R,{\sf m}}\leq \text{const } \theta\big(\theta\mf v\big)^{-4\epsilon} $ is much larger than $\Vert \mc D\Vert_{R,{\sf m}} , $ but for $\mu<0 $ still much smaller than the kinetic (quadratic) contribution, whose real part at the boundary of the small field region is $\sim -(1-e^\mu)R^2 $.\\
The normalized and truncated correlation functions can be obtained from the unnormalized, untruncated ones by the usual formalism. Since the representation (\ref{fnintrep}) is through an integral over a compact domain, by dominated convergence, the resulting truncated correlation functions have a generating function given by
\begin{align}\label{genfctmb}
\log \mc Z(J^*,J) &= \log  \int_{\mb C^{\mb L}} \prod_{x\in \mb L}\ud\phi(x) \chi\big(\phi\big)e^{\mc A(\phi) + \langle J^*,\phi\rangle + \langle J,\phi^*\rangle },
\end{align}
as long as the logarithm exists for $\Vert J\Vert_\infty $ small enough.

\begin{rem}\label{remallowedmu}
The smallness assumption on $\theta $ and validity of the whole construction depends, among many other things, on $\Vert j(\tau) -\id \Vert_{{\sf m}} $, with $\tau\geq 0 $. In \cite{BFKT3}, this allows chemical potentials with $\vert\mu\vert\leq \mc O(1) $. In \cite{BFKT4}, due to reasons related to large fields, a $v$ dependent smallness condition $\vert\mu\vert\leq\const \mf v^{\frac12+\epsilon} $ is added. It is mild enough to include the symmetry breaking region for $\mu>0 $, but rather inconvenient for $\mu<0 $. In particular, it excludes the ``single scale'' region, and does not fall into the range of application of the present work. We believe that, for $\mu <0 $, the restrictions of both papers could be dropped upon changing the single site measure used in these works.

\erem
\end{rem}

\subsection{Unbounded spin systems}

A more simple minded approximation to the (generating functional of the) truncated correlation functions of the Bose gas is
\begin{align}\label{fnintrep2}
\log \mc Z(J^*,J) &= \log  \int_{\mb C^{\mb L}} \prod_{x\in \mb L}\ud\phi(x) e^{\mc A(\phi) + \langle J^*,\phi\rangle + \langle J,\phi^*\rangle }\\
\mc A(\phi) &= -\big\langle \phi(x)^*,\left[-\partial_0+\text{h}-\mu\right]\phi(x)\big\rangle  - \half \big\langle \vert \phi\vert^2,v\vert\phi\vert^2\big\rangle \nonumber
\end{align}
In fact, this is the naively discretized version of the usual coherent state path integral used by physicists \cite{NO}\footnote{It might be asked if the actual correlation functions are the limit of the ones defined by (\ref{fnintrep2}) as $\theta\to0 $. Already at the perturbative level, at low temperatures, some differences to the rigorous approach appear and would have to be addressed. The apparently more complicated representation of \cite{BFKT1} owes its justification to the fact that this ultraviolet limit is comparatively simple to take.}. It lacks the characteristic functions and the non-polynomial $\mc D $ term present in (\ref{fnintrep}), and features a manifestly positive quartic interaction term. The quadratic part is still defined by a non-symmetric ``covariance'' with mass (lower bound on the real part of the spectrum) $-\mu$. \\ 
(\ref{fnintrep2}) is an unbounded spin system with positive polynomial interaction and (non-symmetric) massive covariance. More generally, consider $\sf N $ component fields $\phi:\mb L\to \Xi := \mb R^{{\sf N}} $ and a symmetric, but complex covariance $C\in \text{End }\mb C^{\mb L\times \N} $, $\N = \{1,\ldots ,{\sf N}\} $, which is normal and has eigenvalues whose inverses have real part $\geq \upmu>0 $. We will construct (uniformly in $\vert\mb L\vert $)
\begin{align}\label{genfctuss}
\log \mc Z(J) &= \log \int_{\Xi^{\mb L}} \prod_{x\in X}\ud \phi(x) e^{-\half \langle \phi,C^{-1}\phi\rangle + V(\phi) + \langle J,\phi\rangle},
\end{align}
where $ \ud \phi(x)$ is Lebesgue measure on $\Xi $ and 
\begin{align}\label{formtwobodypotential}
V(\phi) &= -\half \sum_{x\in\mb L} \left( \sum_{y\in\mb L} v^\half(x,y)\phi(y)^2\right)^2.
\end{align}
$v^\half $ is the square root of $v$. In fact, our argument works for more general positive polynomials $V(\phi) $, see in particular Remark \ref{rempositivity} below. Note that complex fields have ${\sf N}=2 $, and a non-symmetric real covariance $C$ for complex fields corresponds to the symmetric, but complex covariance $\left(\begin{array}{cc} C+C^T & i(C-C^T) \\ -i(C-C^T) & C+C^T  \end{array}\right) $ for real two component fields.

\subsection{Description of results}

Let $\mb L $ be a finite set and $d$ be a metric on $\mb L$. As far as the algebra of the expansion is concerned, they are arbitrary, but in the bounds they enter in particular through the geometric constants
\begin{align*}
c_{\sf g}(m)&= \sup_{x\in\mb L}\sum_{x'\in\mb L} e^{-md(x,x')}\\
c_{\sf g}'(a)&= \sup_{Q\subset \mb L} \vert Q\vert^{-1} \cdot \Big\vert \big\{x\in\mb L,\, d(x,Q)\leq a \big\}  \Big\vert 
\end{align*}
where $m,a>0 $. We think of these constants as independent of $\vert\mb L\vert $, and this makes our bounds uniform in $\vert\mb L\vert $. For a $D $ dimensional lattice, $c_{\sf g}(m)\sim m^{-D} $ and $c_{\sf g}'(a)\sim a^D $ if $m<1,a>1 $.\\ 
Let $\sf N\in\mb N$ be fixed. We consider real $\sf N$ component fields $\phi:\mb L\to \mb R^{\sf N} =: \Xi  $. Let $C \in \Text{End } \mb C^{\mb L\times \N}$ be a symmetric normal matrix, and assume that $$\lambda_{\text{min}}\big(\Re C^{-1}\big) =: \upmu>0. $$ We will sometimes write $C(x,y) ,\,x,y\in\mb L$, for the (symmetric) ${\sf N}\times {\sf N} $ matrix $C((x,{\sf m}),(y,{\sf n})) $. We define the Gaussian measure with covariance $C$ by
\begin{align*}
\ud \mu_C(\phi) &= \det ( 2\pi C)^{-\half} \prod_{x\in \mb L} \ud\phi(x) \exp\left(-\half\langle \phi,C^{-1}\phi\rangle\right),
\end{align*}
where $\ud \phi(x) = \ud\phi(x,1) \cdots\ud \phi(x,{\sf N}) $ is Lebesgue measure on $\Xi $ and $$\langle \phi,\psi\rangle =\sum_{x\in \mb L} \phi(x)\cdot \psi(x) = \sum_{x\in\mb L} \sum_{{\sf n}=1}^{\sf N} \phi(x,{\sf n})\psi(x,{\sf n}) = \sum_{\xi,\xi'\in\mb L\times\N} \phi(\xi)\psi(\xi'). $$ The square root is canonically defined and the normalization finite since the eigenvalues of $C$ have positive real part. \\
Let $J\in \Xi^{\mb L} $ be a source field. We consider an interaction $V(\phi;J) = V_1(\phi;J) + V_2(\phi;J) $ with $V_1 $ a generic power series without constant term and $V_2$ a two body potential with source term (more general repulsive polynomial interactions are also allowed, see Remark \ref{rempositivity}). More precisely, introduce the notation $\bd L $ for the space of unordered sequences (multisets) of $\mb L\times \N $,
\begin{align*}
\bd L = \bigcup_{n\geq 0} \Big\{(\xi_1,\ldots,\xi_n) =: (\xi_m)_1^n,\xi_m= (x_m,{\sf n}_m)\in\mb L\times\N\Big\}\diagup_{\mc S_n}
\end{align*}
We sometimes write $\{(\xi_m)_1^n\} $ for the equivalence class of $(\xi_m)_1^n $ and $\{\bm\xi\} = \{\xi_m\}_1^n $ for the set of different elements of $\bm\xi\in\bd L $. We also write $$\{(\xi_1,\ldots,\xi_n)\}\circ \{(\xi'_1,\ldots,\xi'_{n'})\} = \{(\xi_1,\ldots,\xi_n,\xi'_1,\ldots,\xi'_{n'})\}   $$ and $\supp\{(\xi_m)_1^n\} = \{x_m\}_1^n \subset \mb L$. Here and in the following, we write $x,x',x_m,\ldots\in\mb L $ for the first (space) component of $\xi,\xi',\xi_m,\ldots\in \mb L\times\N $, and $z,z',z_m,\ldots\in\mb L $ for the first component of $\zeta,\zeta',\zeta_m, \ldots\in \mb L\times\N $, abandoning more unambiguous (as will turn out, overly unambiguous) notation like $\xi\vert_1 $ for brevity. If, for some $X\subset \mb L $, $x\in X $ for all $\xi\in \bm \xi $, we write $\bm\xi\in \bd L\vert_X $, or $\bm\xi\subset X $. We abbreviate $\phi(\bm\xi ) = \prod_{m=1}^n\phi(\xi_m) $ for $\bm\xi = \{(\xi_m)_1^n\} $.\\
With this notation
\begin{align*}
V_1(\phi;J) &= \sum_{\bm\xi,\bm\zeta\in\bd L} v_1(\bm\xi;\bm\zeta) \phi(\bm\xi)J(\bm\zeta)
\end{align*}
is any power series with $v_1(\emptyset;\emptyset) = 0 $ (no constant term) and, with a constant $a $,
\begin{align*}
V_2(\phi;J) &= - \sum_{x\in\mb L} \left( \sum_{y\in\mb L} v^\half(x,y)\phi(y)^2\right)^2 - a \cdot J(x)\cdot\phi(x)
\end{align*}
a two body potential with source term. Let $\mf v_1,\mf v_2\ll 1\ll r\ll R $ be four constants. For some $\dot m>0 $ and $\lambda_J\geq0 $, we assume that (notation for norms is the same as in (\ref{notnorms}))
\begin{align}
\Vert v^\half\Vert_{2\dot m} &\leq \mf v_2^{\half} \nonumber\\
\lambda_{\text{min}}(v^\half)&\geq c_{v}\mf v_2^\half \label{assv2}\\
a&\leq \lambda_J\mf v_2^{\frac14}\nonumber
\end{align}
for some $c_v>0 $ (independent of $\mf v _2$), and that 
\begin{align}\label{assv1}
\Vert V_1\Vert_{r'+\mf v_1^{-1},\lambda_J,\dot m} &\leq \omega(r')
\end{align}
for all $r'\in[r,R] $, for some nondecreasing function $\omega:[r, R]\to\mb R $ with $\omega(r)\leq 1 $ (see Remark \ref{remchoicer} for candidates). Here,
\begin{align*}
\Vert V_1\Vert_{R,\lambda_J,\dot m} &= \sup_{x\in\mb L} \sum_{\substack{ \bm\xi,\bm\zeta\in\bd L  \\ x\in \supp\bm\xi\circ\bm\zeta  }} R^{n(\bm\xi)}\lambda_J^{-n(\bm\zeta)} e^{\dot md_{\text{t}}(\bm\xi\circ\bm\zeta) } \vert v_1(\bm\xi;\bm\zeta)\vert
\end{align*}
is a $1,\infty $ norm with exponential tree decay\footnote{This norm is not exactly the same as (\ref{originalnorm}); ignoring the source field, we have, for any $R'>R $, $\Vert V_1\Vert_{R,\lambda_J,m}\leq (e\cdot \log\frac{R'}{R})^{-1}\Vert V_1\Vert_{R',m} $, so our norm is slightly weaker. We found it more more convenient to work with.}. If $\Vert J\Vert_\infty\leq \lambda_J^{-1} $ and $\phi\in\Xi^{\mb L} $ is in the support of
\begin{align*}
\chi_R(\phi) &= \prod_{x\in\mb L}\chi\Big(\vert \phi(x)\vert\leq R\Big),
\end{align*}
then the series for $V_1(\phi;J) $ converges absolutely. For an analytic function $W(J)$ just of $J$, we define $ \Vert W\Vert_{\lambda_J,m} $ in the same way, with $ \bm\xi = \emptyset$ everywhere. Our main result is

\begin{thm}\label{mainthm}
Let $  m_V, m>0 $ be fixed. Let $C\in\text{End }\mb C^{\mb L\times \N} $ be symmetric and normal with  
\begin{align*}
\lambda_{\text{min}}\big(\Re C^{-1}\big) =&\, \upmu>0 & \text{and}& & \Vert C\Vert_{6m,\infty} &= \mf c_{\infty}<\infty.
\end{align*}
Fix $\lambda_J>0 $. Depending on these data, let $\mf v_1,\mf v_2\geq0 $ be small enough, $R $ and $\frac Rr $ be large enough, but let $r $ satisfy\footnote{This condition will ensure that the kinetic contribution to the action dominates the one of $V_1 $ both for $0\leq\vert\phi(x)\vert< r $ and $r\leq \vert\phi(x)\vert\leq R $. See the proof of Proposition \ref{boundsintint} and Remark \ref{remchoicer}.}
\begin{align}\label{choicer}
\upmu r^2 \geq 16   \cdot\omega(R)\cdot c_{\sf g}'\left(\frac{\log \omega(R)}{m_V}\right).
\end{align}
$R=\infty $ is also allowed (this implies $V_1=0=\omega $). Let $V_1(\phi;J) $ be analytic and satisfy (\ref{assv1}), and let $V_2(\phi;J) $ be a two body potential with source term that satisfies (\ref{assv2}), with $\dot m = 2m+3m_V $. Set $V= V_1+V_2 $ and
\begin{align}\label{logz}
\mc Z(J) &= \int \ud\mu_C(\phi) \chi_R(\phi)e^{V(\phi;J)}.
\end{align}
Then, $\log \mc Z(J) $ exists, is analytic in $J $, and satisfies $$ \Vert \log \mc Z\Vert_{\lambda_J,m}<\infty.   $$ In particular, the truncated correlation functions decay exponentially with mass $m$. $\log \mc Z $ is given by the Mayer series of a cluster expansion.

\end{thm}

\begin{rem}\label{remallowedchi}
The characteristic function $\chi_R $ in our definition of $\mc Z(J) $ enforces only the first of the conditions (\ref{originalchi}). Thanks to this fact, it is an ``ultra-local'' product of characteristic functions $\chi(\vert \phi(x)\vert\leq R) $, each depending only on a single $\phi(x) $. We will exploit this simplification in our proof, and therefore, our result does not apply to more general characteristic functions. See Remark \ref{remnogenchar} for the simple main reason behind this limitation. At least for ``nearest neighbor'' type characteristic functions as in (\ref{originalchi}), a generalization of the results seems feasible, and we leave it to future work in the multi scale setting, where such small field conditions have their natural origin.  \\
Note also that, while the condition $\Vert \phi\Vert_\infty\leq R $ enforced by our characteristic function is necessary for the existence of the interaction $V_1$, the other two conditions of (\ref{originalchi}) are associated to positive terms in the kinetic part of the action, such as $\langle\phi,\partial_0\phi\rangle $ and $\Vert \nabla\phi\Vert_2^2$. Unlike in the symmetry broken phase $\mu>0 $ for which the representation (\ref{fnintrep}) was developed, in the current case of a massive model, the positivity of these contributions is not needed otherwise, and could in principle be used to eliminate the two ``non-ultra-local'' restrictions of (\ref{originalchi}) by an additional large field decomposition that would partly undo the one that let to (\ref{fnintrep}).

\end{rem}

\erem

\begin{rem}\label{remchoicer}
For the parameters $r,R $ and $\mf v_1 $ used in measuring the size of $V_1 $, a common situation is that $$\mf v_1 = \mf w^\epsilon,\;r = \mf w^{-\epsilon'} \Text{ and }R=\mf w^{-\gamma} \qquad \Text{with} \qquad 0<\epsilon,\epsilon'<1 \Text{ and } \gamma>1 ,$$ for a single small parameter $\mf w $, and that $\Vert V_1\Vert_{\mf v_1^{-1}+R,\lambda_J,\dot m} \leq  \mf w^{-\delta} $ with $\delta\geq 0 $. Suppose $V_1 $ has degree at least $d$ in $\phi$ (w.l.o.g. $d\geq 1 $). Then
\begin{align*}
\Vert V_1\Vert_{r'+\mf v_1^{-1},\lambda_J,\dot m} &\leq \omega(r')
\end{align*}
for $$ \omega(r') = \mf w^{-\delta}\left(\frac{\mf v_1^{-1}+r'}{\mf v_1^{-1}+R}\right)^d . $$ Then, the condition $\omega(r) \leq 1 $ becomes a condition on the exponents. Usually $\delta $ is very close to $0$, $\epsilon<1 $ is at (perhaps small but) finite distance to $1$, and $\epsilon' $ can (if necessary) be chosen very close to $1$. In this situation, $ \omega(r) \leq 1$ follows from the smallness of $\mf w $.\\
Condition (\ref{choicer}) now reads
\begin{align*}
\upmu \mf w^{-2\epsilon'} \geq 16\cdot \mf w^{-\delta}\cdot  c_{\sf g}'\left(\frac{\delta}{m_V}\log \mf w^{-1}\right).
\end{align*}
As remarked above, on a $D$ dimensional lattice, $c_{\sf g}'(a) \sim a^D$, so the last factor is subexponentially large. If $\upmu \geq \mf w^{2\epsilon'-\delta+o(1)}, $ (\ref{choicer}) is satisfied for $\mf w $ small enough (and so $r$ and $\frac Rr $ large enough). In single scale models, usually $\upmu = \mc O(1) $ is considered.

\end{rem}

\erem
Even though we tried to obtain a rather general statement on exponential decay of correlation functions, also in this work (like any other on the topic) the emphasis is more on the method than about the sharpest results. Other or more general results can be obtained by tweaking the algebra of the expansion and the methods of bounding to the model at hand.

\section{The Algebra of the expansion}

In this section, we derive a formula for $\log \mc Z(J) $ as a Mayer series expansion of a polymer system, with explicit formulas for the polymer activities in terms of a cluster expansion. We also derive a representation as a Mayer series with large field small field decomposition. The bounds necessary for proving convergence are provided in the later sections.

\subsection{General tools}

We gather the general tools upon which our algebra rests.

\subsubsection{Mayer expansions} 

The final step of our construction of $\log \mc  Z(J) $ is a Mayer expansion for a polymer gas, by which we mean an application of the formula (see, e.g. \cite{Sa})
\begin{align}
\log \sum_{\{X_m\}_1^n\in \mc P(\mb L)} \prod_{m=1}^n A(X_m) &=  \sum_{X\subset \mb L }W(X)\nonumber \\
W(\{x\}) &= \log A(\{x\})\nonumber \\
W(X) &= \sum_{\{X_m\}_1^n\in \mc M(X)} \rho\Big(\{X_m\}_1^n\Big)\prod_{m=1}^n \dot A(X_m) \qquad\vert X\vert\geq 2\label{mealg} \\
\dot A(X) &= A(X)\cdot\prod_{x\in X}  A(\{x\})^{-1}.\nonumber
\end{align}
for some given activities $A: 2^{\mb L}\to \mb C $ (more generally, $\to $ commutative Banach algebra). Here, 
\begin{align*}
\mc P(\mb L) &= \Big\{   \{X_1,\ldots,X_n\} =:\{X_m\}_1^n , X_1\dot\cup \cdots\dot\cup X_n = \mb L,\;\forall m:X_m\neq \emptyset\Big\}\\
\mc M(X) &= \Big\{ \{X_m\}_1^n,\cup_m X_m=X,\;\forall m:\vert X_m\vert\geq 2 \Big\}.
\end{align*}
$X\dot\cup X'=Y $ means $X\cup X'=Y $, $X,X'\neq\emptyset $ and $X\cap X'=\emptyset $. That is, the $\mc P(\mb L) $ is the set of partitions of $\mb L $. We have used the Ursell functions $$\rho(X_1,\ldots,X_n) = \sum_{g\subset \mc G(X_1,\ldots,X_n) \Text{ connected}}(-1)^{\vert g\vert}, $$ with $\mc G(X_1,\ldots,X_n) $ the incidence graph of $X_1,\ldots,X_n $. The following representation is inspired by \cite{RW}, and the bound goes back to \cite{Ro}

\begin{lem}\label{lemursell}
We have
\begin{align*}
\rho(X_1,\ldots,X_n) &=\sum_{T\in \mf T(\ul n)}\rho(T;X_1,\ldots,X_n)\\ 
\rho(T;X_1,\ldots,X_n)&= (-)^{n-1}\delta_{  T\subset \mc G(X_1,\ldots,X_n) }\int_0^1 \prod_{\ell\in T}\ud s(\ell) \prod_{\ell\in \mc G(X_1,\ldots,X_n)\setminus T} \big[1- s^T(\ell)\big]
\end{align*}
Here, $\ul n = \{1,\ldots,n\} $, $\mf T(\ul n) $ is the set of trees on $\ul n$, and
\begin{align*}
s^T(\{i,j\}) &= \min\{s(\ell),\ell \Text{ on the }T\Text{ path linking }i,j\}.
\end{align*}
In particular, $\vert \rho(X_1,\ldots,X_n)\vert\leq  $ number of spanning trees of $\mc G(X_1,\ldots,X_n) $.
\end{lem}

\begin{proof}
For any connected $g$,
\begin{align*}
w(T,g) &= \int_0^1   \prod_{\ell\in T}\ud  s(\ell) \prod_{\ell\in g\setminus T} s^T(\ell)
\end{align*}
satisfies
\begin{align}\label{plemurs}
\sum_{T\Text{ spanning tree of }g}w(T,g) = 1
\end{align}
To see this, note that
\begin{align*}
w(T,g) &= \int_0^1 \prod_{\ell\in g}\ud s(\ell) \prod_{\ell\in g\setminus T}\prod_{\ell'\in P^T(\ell)}\chi\Big(s(\ell)<s(\ell')\Big)
\end{align*}
Here, $P^T(\{i,j\}) $ is the path on $T$ from $i $ to $j $. Introduce an arbitrary ordering $l:\ul N\to g$ with $N=\vert g\vert $ and $l$ bijective. Then,
\begin{align*}
w(T,g)&= \sum_{\pi\in \mc S_N} \int\limits_{0<s(l(\pi(N)) )<\cdots<s(l(\pi(1))<1}\prod_{\ell\in g}\ud s(\ell) \prod_{\ell\in g\setminus T}\prod_{\ell'\in P^T(\ell)}\chi\Big(s(\ell)<s(\ell')\Big)
\end{align*}
For a given $\pi $, on the corresponding domain of integration the integrand is $0$ unless for all $\ell\in g\setminus T $, we have $l^{-1}(\ell) > \max\{l^{-1}(\ell'),\ell'\in P^T(\ell)\}  $, in which case it is $1$. At a given $\pi $, this is the case for exactly one tree $T(\pi)$, namely the one constructed by inductively adding to the edge set the edge $l(\pi(i)) $ with lowest possible $i$ such that at each stage the resulting graph is loop free (Kruskals algorithm). Therefore, 
\begin{align*}
w(T,g)&= \sum_{\pi\in \mc S_N} \delta_{T,T(\pi)}\int \limits_{0<s(l(\pi(N)) )<\cdots<s(l(\pi(1))<1}\prod_{\ell\in g}\ud s(\ell) \\
&= \sum_{\pi\in \mc S_N}\frac{\delta_{T,T(\pi)}}{N!}
\end{align*}
and
\begin{align*}
\sum_{T\subset g} w(T,g) &= \frac1{N!}\sum_{\pi\in \mc S_N}\sum_{T\subset g}\delta_{T,T(\pi)} = 1
\end{align*}
Applying (\ref{plemurs}) to the definition of the Ursell functions, we obtain
\begin{align*}
\rho(X_1,\ldots,X_n) &= (-)^{n-1}\sum_{g\subset \mc G(X_1,\ldots,X_n) \Text{ connected}}\sum_{T\Text{ spanning tree of }g}\int_0^1  \prod_{\ell\in T} \ud s(\ell) \prod_{\ell\in g\setminus T} -s^T(\ell)\\
&= (-)^{n-1}\sum_{T\in \mf T(\ul n)}\sum_{T\subset g\subset \mc G(X_1,\ldots,X_n)}\int_0^1  \prod_{\ell\in T} \ud s(\ell) \prod_{\ell\in g\setminus T} -s^T(\ell)\cdot \prod_{\ell\in \mc G\setminus g} 1\\
&= (-)^{n-1}\sum_{\substack{T\in\mf T(\ul n)  \\  T\subset \mc G(X_1,\ldots,X_n) }}\int_0^1  \prod_{\ell\in T} \ud s(\ell) \prod_{\ell\in \mc G(Y_1,\ldots,Y_n)\setminus T} \big[1-s^T(\ell)\big]
\end{align*}
as required. The bound $\vert\rho(X_1,\ldots,X_n)\vert\leq  $ number of spanning trees of $\mc G(X_1,\ldots,X_n) $ is obvious from this.
\end{proof}

\subsubsection{Mayer expansion with large field small field decomposition}

In case a large field small field decomposition is performed, we will have a modified polymer gas representation for $Z(J) $, and will need a modified Mayer expansion to take the logarithm. Introduce ${2^{\mb L}}' = \{(X,Q),\emptyset\neq Q\subset X\subset \mb L\} $ and 
\begin{align*}
\mc P'(\mb L) = \Big\{\{(X_m,Q_m)\}_1^n,\, (X_m,Q_m)\in  {2^{\mb L}}',\, \{X_m\}_1^n\in\mc P(\mb L) \Big\}.
\end{align*}
Introduce also ${\tilde 2}^{\mb L} = \{(Z,X,Q),\, Z\subset 2^{\mb L}, (X,Q)\in {2^{\mb L}}',Z\dot\cap X\neq \emptyset\} $ where $Z\dot\cap X\neq \emptyset $ means $Z\cap X\neq \emptyset $ or $Z=\emptyset $. We have the following result.

\begin{lem}\label{sflfme}
Let $\mc V: 2^{\mb L}\to\mb C $ and $B:{2^{\mb L}}' \to \mb C $ be activities. Define 
\begin{align*}
\mc Z&= \sum_{\Omega\subset \mb L} \exp\left(\sum_{Z\subset\Omega}\mc V(Z) \right) \sum_{\{(X_m,Q_m)\}_1^n\in \mc P'(\Omega^c)} \prod_{m=1}^n B(X_m,Q_m).
\end{align*}
Then, formally,
\begin{align*}
\log\mc  Z &= \sum_{Z\subset \mb L} \mc V(Z) + \sum_{(Z,X,Q)\in {\tilde 2}^{\mb L} } \mc L(Z,X,Q)\\
\mc L(Z,X,Q) &= \sum_{\{(Z_m,X_m,Q_m)\}_1^n\in \mc M(Z,X,Q)} \rho(\{Z_m\cup X_m\}_1^n)\prod_{m=1}^n \dot B(Z_m,X_m,Q_m)\\
\dot B(Z,X,Q) &= \sum_{(\{Z_{m'}\}_1^{n'},\{(X_{m},Q_{m})\}_1^{n})\in \mc C(Z,X,Q) } \prod_{m'=1}^{n'} -\mc V(Z_{m'})\prod_{m=1}^{n} B(X_{m},Q_{m}).
\end{align*}
Here,
\begin{align*}
\mc M(Z,X,Q) &= \Big\{ \{(Z_m,X_m,Q_m)\}_1^n,\, (Z_m,X_m,Q_m)\in {\tilde 2}^{\mb L},\, \cup M_m=M,\,M=Z,X,Q   \Big\}\\
\mc C(Z,X,Q) &= \Big \{ (\{Z_{m'}\}_1^{n'},\{(X_{m},Q_{m})\}_1^{n}),\, Z_{m'}\subset \mb L, \{(X_{m},Q_{m})\}_1^{n}\in \mc P'(X) , \\&\qquad\qquad Z_{m'}\dot\cap X\neq \emptyset,\,\mc G\Big(\{Z_{m'}\}_1^{n'},\{X_{m}\}_1^{n}\Big)\Text{ conn.},\,\cup M_{m\ob '}=M,M=Z,X,Q \Big\}
\end{align*}
\end{lem}

\begin{proof}

We compute
\begin{align*}
\mc Z&= \exp\Big(\sum_{Z\subset \mb L} \mc V(Z)  \Big) \sum_{\Omega\subset\mb L}\exp\left(\sum_{Z\cap \Omega^c\neq \emptyset}-\mc V_0(Z)\right) \sum_{\{(X_m,Q_m)\}_1^n\in \mc P'(\Omega^c)} \prod_{m=1}^n B(X_m,Q_m) \\
&= \exp\Big(\sum_{Z\subset \mb L} \mc V(Z)  \Big)\sum_{\Omega\subset\mb L} \sum_{\substack{\{Z_{m'}\}_1^{n'}\\ Z_{m'}\cap \Omega^c\neq \emptyset}}\prod_{m'=1}^{n'}-\mc V(Z_{m'})\sum_{\{(X_m,Q_m)\}_1^n\in \mc P'(\Omega^c)} \prod_{m=1}^n B(X_m,Q_m)  \\
&= \exp\Big(\sum_{Z\subset \mb L} \mc V(Z)  \Big) \sum_{\substack{ \{(Z'_{j},X'_{j},Q'_{j})\}_1^{k} \\ (Z'_{j},X'_{j},Q'_{j})\in {\tilde 2}^{\mb L} \\ (Z'_{j}\cup X'_{j})\cap(Z'_{j'}\cup X'_{j'})=\emptyset    }} \prod_{j=1}^k \dot B(Z'_{j},X'_{j},Q'_{j}).
\end{align*}
In the last step, we collected together sets $X'_{j}\cup Z'_{j} $ according to the connected components of the intersection graph of $\{Z_{m'}\}_1^{n'},\{X_m\}_1^n $. Note that any $Z_{m'} $ intersects at least one $X_m$. The small field region $\Omega $ is identified as $\left(\cup_m X_m\right)^c = \left(\cup_{j}X'_{j}\right)^c $. The Proposition now follows from a Polymer expansion analogous to the standard Mayer expansion.
\end{proof}

\subsubsection{BKAR Formula}

We will arrive at a polymer gas representation of $\mc Z(J) $ by interpolating (in a positivity preserving way) between $\mc Z(J)$ and a product of $\sf N $-dimensional integrals (one for each point in $\mb L$), using the BKAR Taylor forest formula \cite{AR}. This formula states the following. Let $H(s)$ be a smooth function of $s \in [0,1]^{{\sf P}(\mb L)} $, with ${\sf P}(\mb L) $ the full graph (set of all pairs) on $\mb L $. Then
\begin{align}\label{BKAR}
H(1) &= \sum_{\{X_m\}_1^n\in\mc P(\mb L)}\sum_{T_m\in\mf T(X_m)}\int_0^1\prod_{m=1}^n \ud \bd s^{T_m} \left(\prod_{m=1}^{n}\prod_{\ell\in T_m}\partial_{s(\ell)}H\right)( s^{T_1,\ldots,T_n}).
\end{align}
Here, $1 $ is the constant function $s(\{x,y\})=1 $, and
\begin{align*}
\ud \bd s^{T_m} &= \prod_{\ell\in T_m}\ud s(\ell)\\
s^{F}(\{x,y\}) &= \min\{s(\ell),\ell \Text{ on the }F\Text{ path linking }x,y\}
\end{align*}
for any forest $F=\{T_1,\ldots,T_n\}\in\mf F(\mb L) $, the set of forests on $\mb L $. The minimum here is set to $0$ if no path exists between $x $ and $y$. In particular, $\mc P(s^{T_1,\ldots,T_n}) = \{X_m\}_1^n $ a.e., where for any $s\in [0,1]^{{\sf P}(\mb L)} $ we define $\mc P(s)\in\mc P(\mb L) $ to be the finest partition with $s(\{x,y\})=0 $ for any $x,y$ in different elements of $\mc P(s) $. We also define $\mc P(F) = \{X_m\}_1^n $ if $F=\{T_m\}_1^n $ with $T_m\in\mf T(X_m) $. Note that $s^{F}(\ell) = s(\ell) $ if $\ell\in F $. In our application, there will be functions $H(X,s\vert_X) $ such that the following factorization holds
\begin{align*}
\left(\prod_{m=1}^{n}\prod_{\ell\in T_m}\partial_{s(\ell)}H\right)(s^{T_1,\ldots,T_n}) &=  \prod_{m=1}^{n}\left(\prod_{\ell\in T_m}\partial_{s(\ell)}H(X_m;s^{T_m})\right)
\end{align*}
Here, $ s\vert_X \in [0,1]^{{\sf P}(X)} $ is the restriction of $s\in [0,1]^{{\sf P}(\mb L)}$ to pairs in $X\subset \mb L$, and the definition of $s^{T_m}$ is the same as above. For later use, we denote by $d^F(x) $ the coordination number of $x$ in any forest $F \in\mf F(\mb L)$.

\subsubsection{Parameter derivatives of Gaussian integrals}

Let $C \in \Text{End } \mb C^{\mb L\times \N}$ be symmetric and normal, and assume that $\Re C^{-1} $ has strictly positive eigenvalues. Consider the Gaussian measure $\ud\mu_C(\phi) $ on $\Xi^{\mb L} $ associated to $C$. The covariance will be interpolated, and for any $s\in\mb C^{{\sf P}(\mb L)} $ we denote $(C_{ s})(x,y) = C(x,y)  s(\{x,y\}) $. Here and in the following, we set $s(\{x,x\}) \equiv 1 $. Note that this matrix is still symmetric, but not necessarily normal. We will later derive bounds on the eigenvalues of $\Re C_{s}^{-1} $ for $s = s^{T_1,\ldots,T_n} $. In particular, they will be strictly positive, so that also $\ud \mu_{C_{s}}(\phi) $ can be defined. Given this, we have the standard integration by parts formula
\begin{align*}
\partial_{s(\{x,y\})} \int \ud \mu_{C_{s}}(\phi) F(\phi) &= \int \ud \mu_{C_{s}}(\phi) C(\{x,y\})\partial_{\{x,y\}}^2F(\phi)
\end{align*}
where $$C(\{x,y\}) \partial_{\{x,y\}}^2 = \sum_{{\sf m},{\sf n}=1}^{\sf N} C\big((x,{\sf m}),(y,{\sf m})\big) \partial_{\phi(x,{\sf m})}\partial_{\phi(y,{\sf n})}. $$ We will use it for compactly supported piecewise smooth functions and for integrable smooth functions. In the first case, the derivatives are taken in the weak sense. Interpolation as described in the previous section allows to restrict $\ud\mu_C(\phi) $ to integration on $\phi \in\Xi^X $ for some $X\subset \mb L $. Being somewhat sloppy, we reflect such restrictions by writing $\ud\mu_C(\phi\vert_X) $.

\subsection{Algebra for the logarithm of the partition function}

We derive a polymer gas representation for $\mc Z(J) $ under the following assumption on the form of the interaction: $V(\phi;J)  =  V_1(\phi;s;J)\vert_{s=1} + V_2(\phi;s;J)\vert_{s=1} $, with functions $ V_i(\phi;s;J) $ of $\phi,J\in\Xi^{\mb L} $, $s\in[0,1]^{{\sf P}(\mb L)} $ satisfying
\begin{align}\label{asintv1}
V_1(\phi;s;J) = \sum_{X\in\mc P(s)} V_1(\phi\vert_{X};s\vert_{X};J\vert_X)\\\label{asintv2}
V_2(\phi;s;J) = \sum_{X\in\mc P(s)} \sum_{x\in X} V_2(x;\phi\vert_X;s\vert_X;J\vert_X)
\end{align}
for some $V_2(x;\phi;s;J) $. The examples we have in mind are as in Theorem \ref{mainthm}:

\begin{prop}\label{aipgr}
(i) Suppose that 
\begin{align*}
V_1(\phi;J) &= \sum_{\bm\xi,\bm\zeta\in\bd L} v_1(\bm\xi;\bm\zeta) \phi(\bm\xi)J(\bm\zeta)
\end{align*}
is a generic power series. Then, $V_1(\phi;J) = V_1(\phi;s;J)\vert_{s=1} $, where 
\begin{align*}
V_1(\phi;s;J) &= \sum_{\bm\xi,\bm\zeta\in\bd L} \prod_{\ell\in \mc T(\bm\xi\circ\bm\zeta) } s(\ell) v_1(\bm\xi;\bm\zeta) \phi(\bm\xi)J(\bm\zeta)
\end{align*}
satisfies (\ref{asintv1}). Here $\mc T(\bm\xi) $ is an arbitrary minimal spanning tree of $\supp\bm\xi \subset \mb L$.\\
(ii) Suppose that
\begin{align*}
V_2(\phi;J) &= -\sum_{x\in\mb L} \sum_{\substack{\bm\xi,\bm\zeta\in\bd L \\ n(\bm\xi\circ\bm\zeta)  \leq 2M }} v_2(x;\bm\xi;\bm\zeta) \phi(\bm\xi)J(\bm\zeta)
\end{align*}
is a polynomial. Then, $V_2(\phi;J) = V_2(\phi;s;J)\vert_{s=1} $, where 
\begin{align*}
V_2(\phi;s;J) &= -\sum_{x\in\mb L}  \sum_{\substack{\bm\xi,\bm\zeta\in\bd L \\ n(\bm\xi\circ\bm\zeta)  \leq 2M }} v_2(x;\bm\xi;\bm\zeta) \prod_{\xi'\in \bm\xi\circ\bm\zeta}s(\{x,x'\}) \phi(\bm\xi)J(\bm\zeta)
\end{align*}
satisfies (\ref{asintv2}). 
\end{prop}
\begin{proof}
Immediate.
\end{proof}
\noindent

\begin{rem}\label{rempositivity}
As in the assumptions of Theorem \ref{mainthm}, the part $V_1 $ of the interaction will be small in comparison to the kinetic energy in the domain of integration. It is small absolutely in the region $\Vert\phi\Vert_\infty\leq r $ for $ r\ll R$, and small relatively to the kinetic energy in the region $r\leq\vert\phi(x)\vert\leq R $, by an appropriately large choice (\ref{choicer}) of $r$ (in the region $r\ll R $).\\ 
The part $V_2$ is assumed to satisfy a positivity property of the kind
\begin{align*}
\Re V_2(\phi;0)\leq -\lambda_\phi^{2M}\cdot c_{pos}\cdot \sum_{x\in \mb L}\phi(x)^{2M} + c_{pos}'\cdot \vert X\vert
\end{align*}
for some small $\lambda_\phi>0 $ and constants $c_{pos},c_{pos}' $. It is further assumed that the interpolation of (ii) does not destroy this positivity. Note that $V_2(\phi;J)-V_2(\phi;0) $ has degree at most $2M-1 $ in $\phi. $\\
As an example, let 
\begin{align*}
V_2(\phi) &= -\sum_{x,x'\in\mb L} \phi(x)^2 v(x,x') \phi(x')^2 
\end{align*}
be a two body interaction with $v\in\Text{End }\mb C^{\mb L} $ positive definite with minimal eigenvalue $c_v^2\mf v_2\geq \lambda_\phi^4 $. This satisfies the above positivity property with
\begin{align*}
c_{pos} &= 1 & c_{pos}' &= 0
\end{align*}
and so does
\begin{align*}
V_2(\phi; s) &= -\sum_{x\in \mb L} \left(\sum_{x'\in \mb L}v^\half(x,x')s(\{x,x'\})^2 \phi(x')^2\right)^2,
\end{align*}
for $s = s^{T_1,\ldots,T_n} $, because $v^\half(x,x')s(\{x,x'\})^2 $ has minimal eigenvalue at least $c_v\mf v_2^\half $ (compare Lemma \ref{lemhadamard} below). More generally, also the ``effective action style'' interaction $V_2(\phi;J) = V_2(\phi+ DJ) $ for $D\in\text{End }\Xi^{\mb L} $ and its interpolation according to the proposition has the positivity property, and so does, of course, $$ V_2(\phi;J)= \sum_{x\in X} p\Big(x;\lambda_\phi\phi(x);J(x)\Big) $$ for any uniformly bounded above polynomial $p $ of degree $2M$ in $\phi $ (interpolation is trivial here).

\erem
\end{rem}

\begin{rem}\label{remdel}

Sometimes (usually not in single scale models) one studies potentials $V(\phi;J) $ that depend on ``discrete derivatives'' $\partial\phi $ of the field, such as $\partial\phi(x) = \epsilon^{-1}(\phi(x+\epsilon e_\mu) -\phi(x)) $. The above interpolation essentially introduces Dirichlet boundary conditions and is therefore undesirable in this situation. \\
Although we do not investigate this situation further, we sketch a way to find a polymer gas representation for
\begin{align*}
A=\exp\sum_{x_1,\ldots,x_k\in\mb L} v(x_1,\ldots,x_k) \partial^1\phi(x_1)\cdots\partial^n\phi(x_k),
\end{align*}
where $\partial^j\phi(x) $ is linear in $\phi $, depends on $\phi(y) $ for $y\in \square^j(x)\ni x $, and must not be interpolated. For simplicity, here ${\sf N}=1 $. Define the interpolation
\begin{align*}
A(s) &= \exp \sum_{x_1,\ldots,x_k\in\mb L}  v(x_1,\ldots,x_k)\prod_{\ell\in \mc T(\{x_1,\ldots,x_k\})} \mc S(\ell;s)\partial^1\phi(x_1)\cdots\partial^k\phi(x_k)
\end{align*}
with
\begin{align*}
\mc S(\{x_j,x_{j'}\};s) &= \prod_{\{z,z'\}\in {\sf P}\big(\square^j(x_j),\square^{j'}(x_{j'})\big)} s(\{z,z'\}).
\end{align*}
Here, $ {\sf P}\Big(\square^j(x_j),\square^{j'}(x_{j'})\Big) = \{\{z,z'\},z\in\square^j(x_j), z'\in\square^{j'}(x_{j'}),z\neq z' \} $. Note that no factor $ s(\{z,z'\})$ appears more than linearly in $\mc S $. The BKAR formula gives
\begin{align*}
A&=  \sum_{\{X_m\}_1^n\in\mc P(\mb L)} \prod_{m=1}^n A(X_m)\\
A(X) &= \sum_{T\in\mf T(X)}\int \ud \bd s^{T} \prod_{\ell\in T}\partial_{s(\ell)}  \exp \Bigg[\sum_{x_1,\ldots,x_k\in X}  v(x_1,\ldots,x_k) \\& \qqquad\times\prod_{\ell\in \mc T(\{x_1,\ldots,x_k\})} \mc S(\ell;s^{T})\partial^1\phi(x_1)\cdots\partial^k\phi(x_k)\Bigg]
\end{align*}
This is because whenever $x_j\in X_m $ and $x_{j'}\in X_{m'} $, $m\neq m' $, by $\square^j(x)\ni x$, $$ \mc S(\{x_j,x_{j'}\};s^{T_1,\ldots,T_n}) \propto s_{\{x_j,x_{j'}\}} = 0. $$ In particular, for $\ell\in{\sf P}(X_m), $ $ \mc S(\ell;s^{T_1,\ldots,T_n}) $ only depends on $s_{\{z,z'\}} $ for $z,z'\in X_m $. This also implies that the above sum over the $x_j $ actually only includes terms such that not only $x_j\in X $, but $\square^j(x_j)\subset X $. This gives the desired polymer gas representation.

\erem
\end{rem}

\noindent
We now state the expansion formula for $\log \mc Z(J) $ in the case when no large field small field decomposition is performed.

\begin{thm}\label{algebra}
Under the assumptions (\ref{asintv1}), (\ref{asintv2}), we have the formal series expansion
\begin{align*}
\log \mc Z(J) &=  \sum_{X\subset\mb L } W(X;J\vert_X)
\end{align*}
with $W(X;J\vert_X) $ given by (\ref{mealg}) (argument $J$ suppressed there) with 
\begin{align*}
A(X;J\vert_X) &= \sum_{T\in \mf T(X)} \bm\int\limits_{T}\ud\phi\vert_X\ud s\vert_X A'(X;\phi\vert_X;s^T;J\vert_X)\\
A'(X;\phi\vert_X;s^T;J\vert_X)&= \exp\left( V_1(\phi\vert_{X};s^T;J\vert_X)+\sum_{x\in X} V_2(x;\phi\vert_X;s^T;J\vert_X) \right).
\end{align*}
Here, we use the operator
\begin{align*}
\bm\int\limits_{T} \ud \phi\vert_X\ud s\vert_X&=\sum_{F_C\dot\cup F_V=T}  \int_0^1 \ud\bd s^{T}  \int \ud\mu_{C_{s^{T}}}(\phi\vert_X) \\&\qquad\qquad\qquad \prod_{\ell\in F_C}C(\ell)\partial_\ell \chi_R(\phi\vert_{X})\prod_{\ell\in F_V}\partial_{s(\ell)} .
\end{align*}

\end{thm}

\begin{proof}
Apply the BKAR formula to the interpolation
\begin{align*}
\mc Z(J) &= \mc Z(J;s)\vert_{s=1}\\
\mc Z(J;s) &= \int  \ud \mu_{C_s}(\phi) \chi_{R}(\phi)e^{V(\phi;s;J)}.
\end{align*}
Use the factorization of the integrand implied by (\ref{asintv1}), (\ref{asintv2}), the ultralocal structure of $\chi_{R}(\phi) $ and
\begin{align*}
\int\ud \mu_{C_{s^{T'_1,\ldots,T'_k} } }(\phi) = \prod_{j=1}^k \int\ud \mu_{C_{s^{T'_j} } }(\phi\vert_{Y_j}),
\end{align*}
with $\{Y_j\}_1^k = \mc P(s^{T'_1,\ldots,T'_k}) $

\end{proof}

\noindent
The expansion formula with large field small field decomposition is:

\begin{thm}\label{algebralfsf}
Assume (\ref{asintv1}) and (\ref{asintv2}), and introduce the integral operator
\begin{align*}
\bm\int\limits_{T,Q } \ud \phi\vert_X\ud s\vert_X&=\sum_{F_C\dot\cup F_V=T}  \int_0^1 \ud\bd s^{T}  \int \ud\mu_{C_{s^{T}}}(\phi\vert_X) \\&\qquad\qquad\qquad \prod_{\ell\in F_C}C(\ell)\partial_\ell \chi_{Q^c}(\phi\vert_{Q^c})\chi_Q^c(\phi\vert_Q) \prod_{\ell\in F_V}\partial_{s(\ell)} 
\end{align*}
with the characteristic functions ($\vert \phi(x)\vert $ the euclidean norm on $\Xi $)
\begin{align*}
\chi_{Q^c}(\phi\vert_{Q^c}) &= \prod_{x\in Q^c}\chi\Big(\vert \phi(x)\vert\leq r\Big)\\
\chi_{Q}^c(\phi\vert_{Q}) &= \prod_{x\in Q}\chi\Big(r<\vert \phi(x)\vert \leq R\Big).
\end{align*}
($r>0$ is arbitrary at this stage and will be chosen later.) Then, $\log\mc  Z(J) $ is given as in Lemma \ref{sflfme} (argument $J$ suppressed there) with $\mc V $ given by (\ref{mealg}) from 
\begin{align*}
A_s(Z;J\vert_Z) &= \sum_{T\in\mf T(Z)} \bm\int\limits_{T,\emptyset } \ud \phi\vert_Z\ud s\vert_Z A'(Z;\phi\vert_Z;s^T;J\vert_Z)
\end{align*}
and
\begin{align*}
B(X,Q;J\vert_X) &= \sum_{T\in \mf T(X)} \bm\int\limits_{T,Q}\ud\phi\vert_X\ud s\vert_X A'(X;\phi\vert_X;s^T;J\vert_X)
\end{align*}
\end{thm}

\begin{proof}
Proceed as in the proof of Theorem \ref{algebra}. Insert the identity
\begin{align*}
\chi_R(\phi\vert_X) &= \sum_{Q\subset X} \chi_{Q^c}(\phi\vert_{Q^c})\chi_Q^c(\phi\vert_Q).
\end{align*}
This gives
\begin{align*}
\mc Z(J) &= \sum_{\Omega\subset \mb L} \sum_{\{Z_m\}_1^n\in\mc P(\Omega)} \prod_{m=1}^n A_s(Z_m;J\vert_{Z_m}) \times \sum_{\{(X_{m'},Q_{m'})\}_1^{n'}\in \mc P'(\Omega^c)} \prod_{m'=1}^{n'} B(X_{m'},Q_{m'};J\vert_{X_{m'}})
\end{align*}
Use (\ref{mealg}) to exponentiate the first factor in the sum over $\Omega $ as required by Lemma \ref{sflfme}.
\end{proof}

\begin{rem}\label{remcharlittler}
The introduction of the characteristic functions $\chi_Q^c(\phi\vert_Q)  $ are a tool for treating bounds on the $V_1 $ part of the interaction. In this way, they are implicitly used also for showing the convergence of the expansion in Theorem \ref{algebra}. \\
They could also be used to extract somewhat better bounds for the expansion of Theorem \ref{algebralfsf} as compared to Theorem \ref{algebra}. Indeed, we have for the first term contributing to $\log\mc  Z $ in Lemma \ref{sflfme} $$\sum_{Z\subset \mb L} \mc V(Z) = \log \int \ud\mu_C(\phi) \chi_r(\phi)e^{V(\phi;J)} , $$ and our choice of $r \ll R$ will be so that the ordinary perturbation theory (expand the exponential) of this integral is absolutely convergent. For the other terms in Lemma \ref{sflfme}, thanks to the mass of $C$, the characteristic functions $\chi_Q^c(\phi\vert_Q) $ will allow to extract one nonperturbatively small factor for each point in $Q$.\\ 
It might be desirable to also get small factors from other positive contributions to the kinetic energy $\langle \phi, C^{-1}\phi\rangle $. This would involve characteristic functions $\chi_Q^c $ which are not a product of local contributions. Similarly to the situation described before in Remark \ref{remallowedchi}, this can at the moment be handled only if $R=\infty$. For single scale models, it is unnecessary. We leave generalizations in this direction to future work in the multi scale setting.
\erem
\end{rem}

\begin{rem}\label{remflowchart}
Theorems \ref{algebra} and \ref{algebralfsf} as a flow chart:
\begin{align*}
\begin{array}{lllll}
&&V(\phi;J) &  &  \vspace{10pt}\\
&&\downarrow \Text{Prop. \ref{aipgr}} & & \vspace{10pt}\\
&\phantom{A'(X;\phi;J)}&A'(X;\phi;s;J) & &\vspace{10pt}\\
&\swarrow\Text{Thm. \ref{algebra}} & \downarrow\Text{Thm. \ref{algebralfsf}} & \searrow&\Text{Thm. \ref{algebralfsf}, (\ref{mealg})} \vspace{10pt}\\
A(X;J)& & B(X,Q;J) && \mc V(Z;J) \vspace{10pt}\\
\downarrow\Text{(\ref{mealg})} & & \downarrow \Text{Lem. \ref{sflfme}} &\swarrow &\downarrow  \vspace{10pt}\\
\log \mc Z(J) & & \mc L(Z,X,Q;J) & & \downarrow\vspace{10pt}\\
& & \downarrow \Text{Lem. \ref{sflfme}} & \swarrow & \vspace{10pt}\\
& & \log \mc Z(J) & & 
\end{array}
\end{align*}
Here and from now on, we omit the $\vert_X $ in expressions like $A'(X;\phi\vert_X;s\vert_X;J\vert_X) $.
\end{rem}

\section{Norms}

Involved in our construction of $\log \mc Z(J) $ are
\begin{itemize}
\item An interaction $V(\phi;J) $, analytic in its arguments $\phi,J\in\Xi_{\mb C}^{\mb L} $ in an appropriate domain. Here, $\Xi_{\mb C} = \mb C^{{\sf N}} $ is short for the complexified state space of fields.
\item An interpolated interaction $V(\phi;s;J) $, analytic in its arguments $\phi,J\in\Xi_{\mb C}^{\mb L} $ and $s\in\mb C^{{\sf P}(\mb L) } $ in an appropriate domain. In the case of $V_2 $, we also have a family, indexed by $x\in\mb L $, of analytic functions $V_2(x;\phi;s;J)  $, which depends only on $s(\{x,y\}),\,y\in \mb L $ , and only on $\phi(y) ,J(y)$ with $y\in \mc P(s)(x) $.
\item A family, indexed by $X\subset \mb L $, of analytic functions $A'\big(X;\phi;s;J\big) $. It depends on $\phi(x),J(x) $ only for $x\in X $.
\item Three families, indexed by $X\subset \mb L  $, of analytic functions $A\big(X;J\big) $ , $\dot A\big(X;J\big) $ and $W(X;J) $. All these families only depend only on $J(x)$ for $x\in X$.
\item Three more families, indexed by $Z\subset \mb L  $, of analytic functions $A_s\big(Z;J\big) $ , $\dot A_s\big(Z;J\big) $ and $\mc V(Z;J) $. All these families only depend only on $J(x)$ for $x\in Z$.
\item A family, indexed by $(X,Q)\in {2^{\mb L}}' $, of analytic functions $B\big(X,Q;J\big)$. 
\item Families, indexed by $(Z,X,Q)\in \tilde 2^{\mb L} $, of analytic functions $\dot B(Z,X,Q;J)$, $\mc L(Z,X,Q;J) $. By abuse of notation, we also sometimes consider $B\big(X,Q;J\big)\equiv \delta_{Z,\emptyset}B\big(X,Q;J\big) $ as of this kind.
\item $\log \mc Z(J) $ itself, analytic in $J\in\Xi_{\mb C}^{\mb L} $ in an appropriate domain.
\end{itemize}

\subsection{Structure of the norms}

We want to define norms for each of these objects and use them to control the various steps of our construction. Because of the (asymptotically) infinite number of derivatives in $\phi $ needed in Theorem \ref{algebra}, our norms are analyticity norms in this variable. Further, as motivated by interest in $n$-point correlation functions for arbitrary $n$, we also need analyticity norms in the source field $J$. In the interpolation parameters, our norms need only control one derivative. \\
To define them efficiently, we introduce some notation. Define 
\begin{align*}
\mc F_{1}&= \Xi_{\mb C}^{\mb L}\times \Xi_{\mb C}^{\mb L}     &\mc S_{1} &= \emptyset    &\aleph_{1}&= \bd L\times \bd L\\ 
\mc F_{2}&= \Xi_{\mb C}^{\mb L}\times \Xi_{\mb C}^{\mb L}     &\mc S_{2} &= \mb L    &\aleph_{2}&= \mb L\times \bd L\times \bd L\\
\mc F_{\phi}&= \Xi_{\mb C}^{\mb L}\times {\mb C}^{{\sf P}(\mb L)}\times \Xi_{\mb C}^{\mb L}     &\mc S_{\phi} &= 2^{\mb L}   &\aleph_{\phi}&= 2^{\mb L}\times\bd L\times \mf F(\mb L)\times \bd L\\
\mc F^\bullet_J&= \Xi_{\mb C}^{\mb L}     &\mc S^\bullet_J &= {2^{\mb L}}^\bullet   &\aleph^\bullet_{J}&= {2^{\mb L}}^\bullet   \times\bd L\\
\mc F&= \Xi_{\mb C}^{\mb L}     &\mc S &= \emptyset    &\aleph&= \bd L
\end{align*}
Here and in the following, $\bullet $ is a universal wildchart character, employed above to assume the values $\bullet =  $ void, $'$ (prime) or $\tilde{\phantom 2} $ (tilde). We use generic names $\Psi $ for elements of $\mc F_\bullet^\bullet $, $\varsigma $ for elements of $\mc S_\bullet^\bullet $, and $\alpha $ for elements of $\aleph_\bullet^\bullet $. That is,
\begin{align*}
\Psi &= (\phi;J)   & \varsigma &= -  &  \alpha &= (\bm\xi;\bm\zeta)   & n(\alpha) &= (n(\bm\xi),n(\bm\zeta))  \\
 &= (\phi;J)   &  &= x  &   &= (x;\bm\xi;\bm\zeta)   &   &= (n(\bm\xi),n(\bm\zeta))  \\
&= (\phi;s;J)   &  &=  X  &   &= (X;\bm\xi;F;\bm\zeta)   &   &= (\vert X\vert,n(\bm\xi),n(\bm\zeta))  \\
&= J   &  &= X  &   &= (X;\bm\zeta)   &   &= (\vert X\vert,n(\bm\zeta))  \\
&= J    &  &= (X,Q)  &   &= (X,Q;\bm\zeta)   &   &= (\vert X\vert,\vert Q\vert,n(\bm\zeta))  \\
&= J    &  &= (Z,X,Q)  &   &= (Z,X,Q;\bm\zeta)   &   &= (\vert Z\vert,\vert X\vert,\vert Q\vert,n(\bm\zeta))  \\
&= J   &  &= -  &   &= \bm\zeta   &   &= n(\bm\zeta)  
\end{align*}
in the case of $\aleph_1,\aleph_2,\aleph_\phi,\aleph_J,\aleph_J',\tilde\aleph_J,\aleph $, respectively. We also introduced the notation $n(\alpha) $.\\
The objects we need norms for are certain families, indexed by elements $\varsigma \in \mc S_\bullet^\bullet $, of analytic functions of $\Psi\in \mc F_\bullet^\bullet $. The norms will be in terms of their power series in $\Psi $, which can be thought of as functions of $\alpha\in \aleph_\bullet^\bullet $. For this, for a family, indexed by $\varsigma \in\mc S_\bullet^\bullet $, of functions $A(\varsigma;\Psi) $ of $\Psi\in \mc F_\bullet^\bullet $, introduce $A(\alpha;\Psi) $ by
\begin{align*}
A(\alpha;\Psi) &= \nabla_{\phi,\bm\xi}\nabla_{J,\bm\zeta} A(\varsigma;\Psi) & \alpha&= (\varsigma;\bm\xi;\bm\zeta) \in \aleph_1,\aleph_2\\
A(\alpha;\Psi) &= \nabla_{\phi,\bm\xi}\nabla_{s,F}\nabla_{J,\bm\zeta} A(\varsigma;\Psi) & \alpha&= (\varsigma;\bm\xi;F;\bm\zeta) \in \aleph_\phi\\
&= \nabla_{J,\bm\zeta} A(\varsigma;\Psi) & \alpha&= (\varsigma;\bm\zeta) \in \aleph_J^\bullet,\aleph
\end{align*}
Here,
\begin{align*}
\nabla_{\phi,\bm\xi} &= \prod_{\xi\in\{\bm\xi\}}\frac1{n(\bm\xi,\xi)!} \frac{\partial^{n(\bm\xi,\xi)}}{\partial\phi(\xi)^{n(\bm\xi,\xi)}}\\
\nabla_{s,F} &= \prod_{x\in \mb L}\frac1{{\sf N}^{d^F(x)}d^F(x)!} \prod_{\{x,x'\}\in F} \partial_{s(\{x,x'\})}\\
\nabla_{J,\bm\zeta} &= \prod_{\zeta\in\{\bm\zeta\}}\frac1{n(\bm\zeta,\zeta)!} \frac{\partial^{n(\bm\zeta,\zeta)}}{\partial J(\zeta)^{n(\bm\zeta,\zeta)}}
\end{align*}
with $n(\bm\xi,\xi) = \vert\{m,\xi_m=\xi\}\vert $ (the reason for the normalization in $\nabla_{s,F} $ will become clear later). For $F\in\mb C^{\aleph_\bullet^\bullet} $, we define
\begin{align*}
\sum_{\alpha\in\aleph_\bullet^\bullet} &= \sum_{\bm\xi,\bm\zeta\in\bd L}  && \aleph_1 \\
&= \sum_{x\in\mb L} \sum_{\bm\xi,\bm\zeta\in\bd L} && \aleph_2\\
&= \sum_{\substack{X\subset\mb L \\ \vert X\vert \geq 2}}   \sum_{\substack{\bm\xi,\bm\zeta\in\bd L\vert_X \\ F\in \mf F(X) }}  && \aleph_\phi\\
&= \sum_{\substack{X\subset\mb L \\ \vert X\vert \geq 2}}  \sum_{\bm\zeta\in\bd L\vert_X} && \aleph_J\\
&= \sum_{ (X,Q)\in {2^{\mb L}}'   }  \sum_{\bm\zeta\in\bd L\vert_X} && \aleph_J'\\
&= \sum_{ (Z,X,Q)\in {\tilde2^{\mb L}}  }  \sum_{\bm\zeta\in\bd L\vert_X} && \tilde \aleph_J\\
&= \sum_{\bm\zeta\in\bd L} && \aleph.
\end{align*}
Even though not explicitly written, the sums over $X\subset\mb L $ above always exclude the empty set, where the activities from the last section are not defined.

\begin{defin}
Let $\bm\Lambda $ be a space of parameters (masses, coupling constants), to be specified later. For each $\bm\lambda\in\bm\Lambda $ and all values of $\bullet $, let 
\begin{align*}
\mc B_{\bullet,\bm\lambda}^\bullet\subset \mb C^{\aleph_\bullet^\bullet}
\end{align*}
be a suitable set of test functions. Let $\mc D_{\bullet,\bm\lambda}^\bullet  \subset \mc F_\bullet^\bullet $ be to be specified sets of field/interpolation parameter configurations. Then, on a family $A(\varsigma;\Psi) $, indexed by $\varsigma\in\mc S_\bullet^\bullet $, of analytic functions of $\Psi\in \mc F_\bullet^\bullet $, we define the norm
\begin{align*}
\bm\vert  A\bm\vert_{\bullet,\bm\lambda}^\bullet &= \sup_{\eta\in \mc B_{\bullet,\bm\lambda}^\bullet  }  \sum_{\alpha\in \aleph_\bullet^\bullet} \sup_{\Psi\in \mc D_{\bullet,\bm\lambda}^\bullet}  \vert \eta(\alpha)A(\alpha;\Psi) \vert   .
\end{align*}
For a family $A(\varsigma;\Psi) $, indexed by $\varsigma=X\in\mc S_\phi $, of analytic functions of $\Psi=(\phi;s;J)\in \mc F_\phi $, we use instead a second norm (two fat vertical lines), defined by
\begin{align*}
\bm\Vert  A\bm\Vert_{\phi,\bm\lambda}  &= \sup_{ \eta\in \mc B_{\phi,\bm\lambda}   }\sup_{Q\subset \mb L}   \sum_{\alpha\in \aleph_\phi}  \sup_{\Psi\in \mc D_{Q,\bm\lambda} }   \mf G_{Q,\bm\lambda}(\varsigma;\Psi) \vert \eta(\alpha)A(\alpha;\Psi) \vert .
\end{align*}
Here, for each $Q\subset \mb L $, $\mc D_{Q,\bm\lambda} \subset \mc F_\phi $ is another set of field/interpolation parameter configurations, and $\mf G_{Q ,\bm\lambda} :\mc S_\phi\times \mc F_\phi \to\mb R_{>0} $ is a fixed function (somewhat analogous to a ``large field regulator'').

\erem
\end{defin}

\begin{rem}
For later use, we also denote for $F\in\mb C^{\aleph_\bullet^\bullet} $
\begin{align*}
\Vert F\Vert_{\bm\lambda} &= \sup_{\eta\in\mc B_{\bullet,\bm\lambda}^\bullet}  \sum_{\alpha\in \aleph_\bullet^\bullet}\vert \eta(\alpha)F(\alpha)\vert.
\end{align*}
This has the spirit of a weak definition of norms, except that there is no real dual pairing since the absolute values are inside all the sums encoded in $\sum_{\alpha\in\aleph_\bullet^\bullet} $. For single scale problems, this poses no problem. For multi scale problems, where the issue of Remark \ref{remdel} becomes important, it is desirable to have some of these sums inside the absolute values. In this case, somewhat different forms of the cluster expansion are necessary.

\erem
\end{rem}

\subsection{Properties of the norms}

We will choose the sets $\mc B_{\bullet,\bm\lambda}^\bullet $ explicitly in Section 5. They will satisfy certain properties that allow us to control the expansions related to: the passage $V(\phi;J)  \to A'(X;\phi;s;J) $; the cluster expansion of Theorems \ref{algebra} or \ref{algebralfsf}; and the Mayer expansion of (\ref{mealg}) or Lemma \ref{sflfme}.

\subsubsection{Properties of test functions}

The first property is for $V(\phi;J)\to A'(X;\phi;s;J) $. We use the notations $\supp\{(\xi_1,\ldots,\xi_n)\} = \{x_1,\ldots,x_n\} $ and
\begin{align*}
\supp (\bm\xi;\bm\zeta) &= \supp\bm\xi\circ\bm\zeta\\
\supp (X;\bm\xi;F;\bm\zeta) &= X\\
\supp (X;\bm\zeta) &= X\\
\supp (X,Q;\bm\zeta) &= X\\
\supp (Z,X,Q;\bm\zeta) &= Z\cup X
\end{align*}
We also define $\alpha\circ\alpha'  $ componentwise, with $X\circ X' = X\cup X' $ and $F\circ F' = F\cup F' $ if $F\cap F'=\emptyset $ ($F\circ F' $ is undefined otherwise).\vspace{10pt}

\noindent
\textbf{Property 1.} Let $ \lambda_\bullet = \lambda_\bullet(\bm\lambda) $, $\bullet = \phi,J,s,1,2,T,X$, and $m_\bullet= m_\bullet(\bm\lambda) $, $\bullet=\text{void,}V$, be components of $\bm\lambda $, specified later. We will have $\lambda_\bullet\leq 1 $ and $m_\bullet >0 $.\\ 
(i) Let $\Lambda^1 $ be a to be specified transformation on the space $\bm\Lambda $ of parameters, and denote $\dot{\bm\lambda} = \Lambda^1\bm\lambda $. Then,
\begin{align*}
\sum_{ X\subset \mb L  } \sum_{\substack{\bm\xi,\bm\zeta\in \bd L\vert_X  \\  F\in \mf F(X)}}  \dot\lambda_X^{-\vert X\vert}   \dot\lambda_\phi^{n(\bm\xi)} \dot\lambda_s^{2\vert F\vert} \dot\lambda_J^{n(\bm\zeta)} e^{- (m+  2m_V)d(F)}\vert \eta(X;\bm\xi;F;\bm\zeta)\vert \leq 1
\end{align*}
for any $\eta\in \mc B_{\phi,\bm\lambda} $.\\
(ii) For a fixed nondecreasing $\omega(r') $ with $\omega(r)\leq 1 $ we have $\eta\in\mc B_{1,\bm\lambda} $ whenever
\begin{align*}
\vert \eta(\bm\xi;\bm\zeta)\vert \leq \omega(r')^{-1}  \left(r'+\lambda_1^{-1}\right)^{n(\bm\xi)}\lambda_J^{-n(\bm\zeta)} \left[\prod_{\ell\in \mc T(\bm\xi\circ\bm\zeta)} \lambda_s^{-2}e^{  m d(\ell)}\right] \frac{\delta_{ \bm\xi\circ\bm\zeta\subset X}}{\vert X\vert}
\end{align*}
for some fixed $X\subset \mb L $ and $r'\in [r,R] $.\\
(iii) We have $\eta\in\mc B_{2,\bm\lambda}$ whenever
\begin{align*}
\vert \eta(x;\bm\xi;\bm\zeta)\vert \leq  \lambda_2^{-n(\bm\xi)}\lambda_J^{-n(\bm\zeta)}  \prod_{\xi'\in \bm\xi\circ\bm\zeta}e^{md(x,x')} f(x) \delta_{  \bm\xi\circ\bm\zeta\subset X} 
\end{align*}
for some fixed $X\subset \mb L $ and some $f\in\mb R_{\geq 0}^{X} $ with $\Vert f\Vert_1\leq 1 $.

\erem

\begin{rem}\label{rembint}
Clearly, conditions (ii) and (iii) would be inappropriate in the setting touched upon in Remark \ref{remdel}. In the single scale case we are interested in, they are sufficient, and allow for simplified proofs of the bound on the interpolation of the interaction. In particular, the use of (i) will be possible, while in more general contexts, (i) would have to be replaced by something analogous to Property 3 below.

\end{rem}

\erem
The properties for Theorem \ref{algebra} and Theorem \ref{algebralfsf} are somewhat abstract at this stage and will be clarified in the next section:\vspace{10pt}\\
\textbf{Property 2.} For any $\eta\in \mc B_{J,\bm\lambda}^\bullet $ and functions $\gamma(T;X;\bm\xi;F) $ and $\gamma'(T;X,Q;\bm\xi;F),Q\subset X\subset\mb L,\bm\xi\in\bd L\vert_X,T\in\mf T(X),F\in\mf F(X) $ as defined in the proof of Proposition \ref{boundce}, we have
\begin{align*}
\sum_{T\in \mf T(X)}\gamma(T;X;\bm\xi;F)\eta(X;\bm\zeta) \in \mc B_{\phi,\Lambda^2\bm\lambda}\\
\sum_{\substack{T\in \mf T(X) \\ \emptyset\neq Q\subset X}}\gamma'(T;X,Q;\bm\xi;F)\eta(X,Q;\bm\zeta) \in \mc B_{\phi,\Lambda^2\bm\lambda}
\end{align*}
for some map $\Lambda^2 $ on the space $\bm\Lambda $ of parameters, specified later. For simple reasons, it is actually not globally defined, see Proposition \ref{propprop2}.

\erem
To state the properties for the Mayer expansions (\ref{mealg}) and Lemma \ref{sflfme} in a language that is also useful for alternative versions of the cluster expansion, we introduce, for values of $\bullet $ as needed, for any $n\geq 2$ and any $T\in\mf T(\ul n) $, sets of test functions
\begin{align*}
\mc B_{J,\bm\lambda}^\bullet(T) \subset \mb C^{\left(\aleph_J^\bullet\right)^n}.
\end{align*}
They are specified later. We will prove\vspace{10pt}\\
\textbf{Property 3.} 
(i) Let $\Lambda_J^\bullet $ be a map on the space of parameters $\bm\Lambda $, specified later. Let $F_m\in \mb C^{\aleph_J^\bullet},m=1,\ldots,n, $ be such that $\Vert F_m\Vert_{\Lambda_{J}^\bullet\bm\lambda}<\infty $. Then, for any $T\in\mf T(\ul n) $,
\begin{align*}
\sup_{\eta\in \mc B_{J,\bm\lambda}^\bullet(T)} \sum_{\alpha_m,m\in\ul n} \eta(\alpha_1,\cdots,\alpha_n) \prod_{m=1}^n \vert F_m(\alpha_m)\vert &\leq \left(\frac{c_{J}^\bullet(\bm\lambda)}8\right)^n  \prod_{m=1}^n  d^T(m)! \Vert F_m\Vert_{\Lambda_{J}^\bullet\bm\lambda}
\end{align*}
for some constant $c_{J}^\bullet(\bm\lambda)\geq 1 $, depending on the values of $\bullet$.\\
(ii) Let $ \eta\in \mc B_{J,\bm\lambda}^\bullet$, and define
\begin{align*}
\delta_T\eta(\alpha_1,\ldots,\alpha_n) &=  \eta(\alpha_1\circ\cdots\circ\alpha_n) \prod_{\{m,m'\}\in T}\delta_{\supp\alpha_{m}\cap\supp\alpha_{m'} \neq\emptyset}
\end{align*}
Then $ \delta_T\eta\in\mc B_{J,\bm\lambda}^\bullet(T) $. Further, for $\eta $ as before and for $a\in \mb C^{\bd L} $ so that $$ \sum_{\bm\zeta\in\bd L\vert_X}\lambda_J^{-n(\bm\zeta)}\vert a(\bm\zeta)\vert \leq \big(c^3_{\Text{ii}}\big)^{\vert X\vert} ,$$ define $$ \eta a(\alpha) = \sum_{\bm\zeta'\in \bd L\vert_{ \supp\alpha}}\eta(\varsigma;\bm\zeta\circ\bm\zeta')a(\bm\zeta').  $$ Then $ \eta a\in\mc B_{J,\Lambda^3_{\Text{ii}}\bm\lambda }$, for some transformation $\Lambda^3_{\Text{ii}} $ on $\bm\Lambda $, specified later, and depending on $c^3_{\Text{ii}} $. \\
(iii) Let $\eta\in \tilde{\mc B}_{J,\bm\lambda} $, and denote 
\begin{align*}
\mc R\eta(Z;\bm\zeta) &= \eta(Z,\emptyset,\emptyset;\bm\zeta) & \mc R\eta &\in \mb C^{\aleph_J}  \\
\mc R'\eta(X,Q;\bm\zeta) &= \eta(\emptyset,X,Q;\bm\zeta) & \mc R\eta &\in \mb C^{\aleph_J'} .
\end{align*}
Then $\mc R^\bullet\eta\in\mc B_{J,\Lambda^3_{\Text{iii}}\bm\lambda}^\bullet $ for some $\Lambda^3_{\Text{iii}} $ specified later.\\
(iv) For $\eta\in \mc B_{\bm\lambda} $, denote, with $\upmu = \lambda_{\text{min}}(\Re C^{-1}) $,
\begin{align*}
\mc I\eta(X;\bm\zeta) &= \delta_{\supp\bm\zeta\subset X}\eta(\bm\zeta) & \mc I\eta\in \mb C^{\aleph_J}\\
\tilde{\mc I}\eta(Z,X,Q;\bm\zeta) &= \delta_{\supp\bm\zeta\subset X\cup Z}\delta_{\emptyset\neq Q\subset X} e^{\frac\upmu8r^2\vert Q\vert}\eta(\bm\zeta) & \mc I\eta\in \mb C^{\tilde \aleph_J}
\end{align*}
Then $\mc I^\bullet\eta\in\mc B^\bullet_{J,\Lambda_{\Text{iv}}^3\bm\lambda} $ for some $\Lambda^3_{\Text{iv}} $ specified later. \\
(v) There is a transformation $\Lambda^3_{\Text{v}} $ on $\bm\Lambda $ such that, with $\dot{\bm\lambda} = \Lambda^3_{\Text{v}} \bm\lambda $, we have
\begin{align*}
\sum_{x\in \mb L}\sup_{\supp \bm\zeta=\{x\}}\dot\lambda_J^{n(\bm\zeta)}\vert \eta(\{x\};\bm\zeta)\vert \leq \dot\lambda_X^{-1}.
\end{align*}
for any $\eta\in\mc B_{J,\bm\lambda} $.

\erem

\subsubsection{Choice of test field configurations and regulators}

Our bounds further rely on choices $\mc D_{\bullet,\bm\lambda}^\bullet $ of test field configurations and the ``large field regulator'' $\mf G_{Q,\bm\lambda}(X;\phi;s;J) $ which, in this context, have to be compatible in much simpler ways than the choices of the $\mc B_{\bullet,\bm\lambda}^\bullet $. We set $\mc D_{1,\bm\lambda} = \mc D_{2,\bm\lambda} = \{(0\,;0)\}, $ the zero field configuration, $\mc D_{J,\bm\lambda}^\bullet = \mc D_{\bm\lambda} = \{J\in \Xi^{\mb L}, \,\Vert J\Vert_\infty\leq \kappa_J^{-1}\} $ for a fixed radius of analyticity $\kappa_J^{-1} $ as in Theorem \ref{mainthm}, and 
\begin{align*}
\mc D_{Q,\bm\lambda}\subset \Big\{(\phi;s;J),\, \phi,J\in \Xi^{\mb L}, &\vert\phi(x)\vert\leq r,x\in Q^c,\,r<\vert\phi(x)\vert\leq R,\,x\in Q, \,\Vert J\Vert_\infty\leq \kappa_J^{-1}, \\ & s\in [0,1]^{{\sf P}(\mb L)},\, s=s^F \Text{ for some } F\in\mf F(\mb L),s(\ell)\in [0,1]\forall \ell\in F \Big\} 
\end{align*}
$\mc D_{\phi,\bm\lambda} $ is never needed. Our choice for the ``large field regulator'' $\mf G_{Q,\bm\lambda}(X;\phi;s;J) $ is
\begin{align}
\mf G_{Q,\bm\lambda}(X;\phi;s;J) &= \exp\left(- \frac{\upmu}8\sum_{x\in X} \phi(x)^2\right).
\end{align}
For the current problem, it is actually independent of $Q,s,J$. For later use, we abbreviate $\upmu_{\leq 1} = \upmu\wedge 1. $

\section{Bounds for the expansion}

In this section, we prove bounds on our expansion. They rest on the properties of our norms stated in the last section. These properties will be verified in the next section.

\subsection{Interpolation of the interaction}

\begin{prop}\label{boundsintint}
Suppose that $V(\phi;J) = V_1(\phi;s;J)\vert_{s=1} + V_2(\phi;s;J)\vert_{s=1} $, with $V_1,V_2 $ as in Proposition \ref{aipgr}. As in Theorem \ref{algebra}, set
\begin{align*}
A'(X;\phi;s;J) &= \exp \left( V_1(\phi\vert_X;s\vert_X;J\vert_X) + \sum_{x\in X} V_2(x;\phi\vert_X;s\vert_X;J\vert_X) \right).
\end{align*}
Let the transformation $\Lambda_{\Text{P.\ref{boundsintint}}}$ on the space $\bm \Lambda $ of parameters be so that the components of $\dot{\bm\lambda} = \Lambda_{\Text{P.\ref{boundsintint}}}\bm\lambda $ satisfy
\begin{align}\nonumber
\dot\lambda_1^{-1}&\geq \ddot\lambda_\phi^{-1} &\dot\lambda_2^{-1} \geq  \dot\lambda_\phi^{-1} &\geq \ddot\lambda_\phi^{-1} + 6\upmu_{\leq 1}^{-1}\ddot\lambda_s^{-2}  \\  \dot\lambda_J^{-1}  &\geq \ddot\lambda_J^{-1} + \kappa_J^{-1}  & \dot\lambda_s^{-2} &\geq  1+ 6\upmu_{\leq 1}^{-1}\ddot\lambda_s^{-2} \label{asboundsintint} \\         \dot m &\geq m+3 m_V  \nonumber
\end{align}
with $\ddot{\bm\lambda} = \Lambda^1\bm\lambda $. Assume condition (\ref{choicer}) from Theorem \ref{mainthm}:
\begin{align}\label{aspos}
\upmu r^2 \geq 16   \cdot\omega(R)\cdot c_{\sf g}'\left(\frac{\log \omega(R)}{m_V}\right).
\end{align}
Assume further that $\dot\lambda_\phi $ is so small that, for any $F\in\mf F(X) $ and $s\in[0,1]^F $,
\begin{align}\label{condpos}
\sum_{x\in X}\Re V_2(x;\phi\vert_X;s^F;0)\leq -\dot\lambda_\phi^{2M} \cdot c_{pos}\cdot \sum_{x\in X}\phi(x)^{2M} + c_{pos}' \cdot \vert X\vert
\end{align}
for some fixed $c_{pos}>0$, $c_{pos}'\geq0 $ (cf. Remark \ref{rempositivity}) and all $\phi\in\Xi^{\mb L} $. Finally, assume that $\ddot\lambda_{X}\leq c_{\Text{P.\ref{boundsintint}} }^{-1} $, with $ c_{\Text{P.\ref{boundsintint}} } = c_{\Text{P.\ref{boundsintint}} } ({\sf N},m_V,r,R,\omega,M,c_{pos},c_{pos}')$ defined at the end of the proof. If $ \bm\vert V_1\bm\vert_{1,\dot{\bm\lambda}} ,\, \bm\vert V_2\bm\vert_{2,\dot{\bm\lambda}}<1$, then $ \bm\Vert A'  \bm\Vert_{\phi,\bm\lambda} <1$.

\end{prop}

\begin{proof}
Fix $X\subset \mb L $ for the moment. Note that $V_2(x;\phi;s;J) $ depends on $\phi $ and $s$ only through the combination $s_x\phi := s(\{x,x'\})\phi(x') $. We write sloppily $V_2(x;\phi;s;J)= V_2(x;s_x\phi;J) $. Define $X_1:= X\sqcup \{1\} $ and, for $x_1\in X_1 $, set $V(x_1;\phi;s;J) = V_1(\phi\vert_X;s\vert_X;J\vert_X) $ if $x_1=1 $ and $ V(x_1;\phi;s;J) = V_2(x;s_x\phi;J\vert_X)$ if $x_1=x\in X $. Also set
\begin{align*}
A'(x_1;X;\phi;s;J) &= e^{V(x_1;\phi;s;J)}
\end{align*}
so that
\begin{align*}
A'(X;\phi;s;J) &= \prod_{x_1\in X_1} A'(x_1;X;\phi;s;J).
\end{align*}
By the product rule, for $\alpha = (X;\bm\xi;F;\bm\zeta)\in \aleph_\phi $ ,
\begin{align*}
A'(\alpha;\phi;s;J) &= \sum_{\substack{\alpha_{x_1'}\in \aleph_\phi \\  \circ_{x_1'}\alpha_{x_1'} = \alpha }} c\big((\alpha_{x_1'})_{x_1'\in X_1};\alpha\big) \prod_{x_1\in X_1}A'(x_1;\alpha_{x_1};\phi;s;J)\\
c\big((\alpha_{x_1'})_{x_1'\in X_1};\alpha\big) &= \prod_{x_1\in X_1}\delta_{X_{x_1},X} \prod_{x\in X} \frac{\prod_{x_1\in X_1} d^{F_{x_1}}(x)! }{d^F(x)!} 
\end{align*}
Note that, by definition, $F $ is the disjoint union of the $F_{x_1} $. For $x_1=1 $, we have
\begin{align*}
A'(1;\alpha_1;\phi;s;J) &= \prod_{\xi\in\{\bm\xi_1\}} \frac1{2\pi i}\oint\limits_{\vert \psi(\xi)\vert = \ddot\lambda_{\phi}^{-1}} \frac{\ud\psi(\xi)}{\psi(\xi)^{n(\bm\xi_1,\xi)+1 }} \times \prod_{\ell\in {\sf P}(X)}\frac1{2\pi i} \oint\limits_{\vert t(\ell)\vert =   \tilde\lambda_s^{-2}e^{ \tilde md(\ell)}} \frac{\ud t(\ell)}{t(\ell)^{1+\delta_{\ell\in F_1}}}\\&\qquad\qquad\qquad\times  \prod_{\zeta\in\{\bm\zeta_1\}} \frac1{2\pi i}\oint\limits_{\vert K(\zeta)\vert= \ddot\lambda_J^{-1} }\frac{\ud K(\zeta)}{K(\zeta)^{n(\bm\zeta_1,\zeta)+1}}  e^{V_1(\phi+\psi;s+t;J+K)}
\end{align*}
We set here $\tilde\lambda_s = \frac16 \upmu_{\leq 1}^{\half}\cdot \ddot\lambda_s $ and $\tilde m = m+2m_V $. For $x_1=x\in X $, $A'(x;\alpha_{x};\phi;s;J) $ vanishes if $ x\not\in \ell$ for some $\ell\in F_x $. It is non-zero if $$F_x = \big\{\{x,x'\},x'\in N_x\big\} =: \{x,N_x\} $$ with $N_x=\supp F_x\setminus \{x\}$ the set of neighbors of $x$ in $F_x$ (local notation), and then
\begin{align*}
A'(x;\alpha_{x};\phi;s;J) &=\prod_{\xi'\in\{\bm\xi_x\}} \frac{s(\{x,x'\})^{n(\bm\xi_x,\xi')}}{n(\bm\xi_x,\xi')!}\frac{\partial^{n(\bm\xi_x,\xi')}}{\partial s_x\phi(\xi')^{n(\bm\xi_x,\xi')}}\\&\qquad\times \frac1{{\sf N}^{\vert F_x\vert}\vert F_x\vert!} \prod_{x'\in N_x } \phi(x')\cdot\pderi{s_x\phi(x')} \\&\qquad\times  \prod_{\zeta\in\{\bm\zeta_x\}} \frac1{n(\bm\zeta_x,\zeta)!}\frac{\partial^{n(\bm\zeta_x,\zeta)}}{\partial J(\zeta)^{n(\bm\zeta_x,\zeta)}} e^{V_2(x;s_x\phi;J)}\\
&=  \prod_{ \xi'\in \{\bm\xi_x\} } \frac{s(\{x,x'\})^{n(\bm\xi_x,\xi')}}{2\pi i }\oint\limits_{\vert \psi_x(\xi') \vert =  \ddot\lambda_\phi^{-1}e^{  m_Vd(x,x')}} \frac{\ud\psi_x(\xi')}{\psi_x(\xi')^{n(\bm\xi_x,\xi')+1}} \\&\qquad \times  \frac1{\vert F_x\vert!} \prod_{x'\in N_x}\sum_{{\sf m}\in\N} \frac{2\,\phi(x',{\sf m})}{\sf N} \prod_{{\sf n}\in\N} \frac{1}{2\pi i } \!\!\oint\limits_{\vert \psi_x'(x',{\sf n}) \vert =  \tilde\lambda_s^{-2}e^{ \tilde md(x,x')}}  \!\!\!\frac{\ud\psi_x'(x',{\sf n})}{\psi_x'(x',{\sf n})^{ \delta_{{\sf n},{\sf m}}+1} } \\&\qquad\times  \prod_{\zeta'\in\{\bm\zeta_x\}} \oint\limits_{\vert K_x(\zeta') \vert =  \ddot\lambda_J^{-1}e^{m_Vd(x,z')} }\frac{\ud K_x(\zeta')}{K_x(\zeta')^{n(\bm\zeta_x,\zeta')+1}} \;\; e^{V_2(x;s_x\phi+\psi_x+\psi_x';J+K_x)}.
\end{align*}
By definition,
\begin{align*}
V_1(\phi+\psi;s+t;J+K) &= \sum_{\alpha\in \aleph_1} \eta_1(\alpha)V_1(\alpha;\Psi)\Big\vert_{\Psi=0}
\end{align*}
with
\begin{align*}
\eta_1(\bm\xi;\bm\zeta) &=  (\phi+\psi)(\bm\xi) \cdot(J+K)(\bm\zeta) \prod_{\ell\in\mc T( \bm\xi\circ\bm\zeta)} (s+t)(\ell) \;\delta_{\bm\xi\circ\bm\zeta\subset X}.
\end{align*}
Fix any $Q\subset \mb L $, $(\phi,s,J)\in \mc D_{Q,\bm\lambda} $, and $(\psi;t;K) $ in the region of integration specified above. Write 
\begin{align*}
\eta_1(\bm\xi;\bm\zeta) &= \eta_1(\bm\xi;\bm\zeta)\delta_{\bm\xi\circ\bm\zeta\subset \ol Q} + \eta_1(\bm\xi;\bm\zeta)\delta_{\bm\xi\circ\bm\zeta\not\subset \ol Q}\delta_{\bm\xi\circ\bm\zeta\subset Q^c} + \eta_1(\bm\xi;\bm\zeta)\delta_{\bm\xi\circ\bm\zeta\not\subset \ol Q}\delta_{\bm\xi\circ\bm\zeta\not\subset Q^c}
\end{align*}
where 
\begin{align*}
\ol Q &= \Big\{ x\in \mb L,\, d(x,Q)\leq \frac{\log \omega(R)}{m_V}  \Big\}.
\end{align*}
By Property 1 (iii) and (\ref{asboundsintint}),
\begin{align*}
\vert \eta_1(\bm\xi;\bm\zeta)\delta_{\bm\xi\circ\bm\zeta\subset \ol Q} \vert \in \vert Q\vert \cdot\omega(R)\cdot c_{\sf g}'\left(\frac{\log \omega(R)}{m_V}\right) \cdot \mc B_{1,\dot{\bm\lambda}}.
\end{align*}
Similarly,
\begin{align*}
\vert \eta_1(\bm\xi;\bm\zeta)\delta_{\bm\xi\circ\bm\zeta\not\subset \ol Q}\delta_{\bm\xi\circ\bm\zeta\subset Q^c} \vert \in \vert X\vert\cdot \mc B_{1,\dot{\bm\lambda}}.
\end{align*}
Finally, note that
\begin{align*}
\vert\eta_1(\bm\xi;\bm\zeta) \delta_{\bm\xi\circ\bm\zeta\not\subset \ol Q}\delta_{\bm\xi\circ\bm\zeta\not\subset Q^c}\vert &\leq e^{-m_Vd(Q,{\ol Q}^c)}\Big(R+\ddot\lambda_\phi^{-1}\Big)^{n(\bm\xi)}\\&\qquad\qquad\qquad\times\Big(\kappa_J^{-1}+\ddot\lambda_J^{-1}\Big)^{n(\bm\zeta)} \left[ \prod_{\ell\in\mc T(\bm\xi\circ\bm\zeta) } \tilde\lambda_s^{-2}e^{(\tilde m+m_V)d(\ell)}\right]\delta_{ \bm\xi\circ\bm\zeta\subset X} . 
\end{align*}
Therefore, by the choice of $\ol Q $,
\begin{align*}
\vert\eta_1(\bm\xi;\bm\zeta) \vert \in \vert X\vert\cdot \mc B_{1,\dot{\bm\lambda}}.
\end{align*}
It follows that, for $Q\subset \mb L $, $(\phi,s,J)\in \mc D_{Q,\bm\lambda} $, and $(\psi;t;K) $ in the above region of integration
\begin{align*}
\vert V_1(\phi+\psi;s+t;J+K)  \vert \leq \left[2\vert X\vert +   \vert Q\vert \cdot\omega(R)\cdot c_{\sf g}'\left(\frac{\log \omega(R)}{m_V}\right)\right]\; \bm\vert V_1\bm\vert_{1,\dot{\bm\lambda}}.
\end{align*}
For $V_2 $, note that
\begin{align*}
V_2(x;s_x\circ\phi+\psi_x+\psi_x';J+K_x) &= V_2(x;s_x\circ\phi;0) + V_2'(x;s_x\circ\phi,\psi_x+\psi'_x;J+K_x) \\
V_2'(x;s_x\circ\phi,\psi_x+\psi_x';J+K_x)  &= -\sum_{\substack{\bm\xi,\bm\upsilon,\bm\zeta \in\bd L\vert_X  \\ n(\bm\xi\circ\bm\upsilon\circ\bm\zeta)\leq 2M \\ n(\bm\xi)<2M }} v_2(x;\bm\xi\circ\bm\upsilon;\bm\zeta) \prod_{\xi'\in\bm\xi\circ\bm\upsilon}s(\{x,x'\}) \\&\qqquad\times \phi(\bm\xi)(\psi_x+\psi_x')(\bm\upsilon)(J+K_x)(\bm\zeta).
\end{align*}
$V_2' $ is a polynomial in $\phi $ with degree $<2M $. By definition,
\begin{align*}
\sum_{x\in X}V_2'(x;s_x\circ\phi,\psi_x+\psi_x';J+K_x) &= \sum_{\alpha\in\aleph_2}\eta_2(\alpha) V_2(\alpha;\Psi)\Big\vert_{\Psi=0}
\end{align*}
with
\begin{align*}
\eta_2(x;\bm\xi;\bm\zeta) &=  \sum_{\substack{\bm\xi'\circ\bm\upsilon=\bm\xi  \\n(\bm\xi')<2M   }} \phi(\bm\xi')(\psi_x+\psi_x')(\bm\upsilon)\prod_{\xi'\in\bm\xi}s(\{x,x'\}) (J+K_x)(\bm\zeta) \;\; \delta_{x\in X}\delta_{\bm\xi\circ\bm\zeta\subset X}
\end{align*}
We distinguish the terms with $n(\bm\xi')=0 $ and $n(\bm\xi')>0 $ to get $\eta_2= \eta_2'+\eta_2'' $. For $\Vert J\Vert_\infty\leq \kappa_J^{-1},s\in[0,1]^{{\sf P}(X)} $ and $\psi_x,\psi_x',K_x $ in the regions of integration above, we have, by (\ref{asboundsintint}),
\begin{align*}
\vert \eta_2'(x;\bm\xi;\bm\zeta)\vert &\in  \vert X\vert \cdot \mc B_{2,\dot{\bm\lambda}}.
\end{align*}
For $\eta_2'' $, by Young's inequality and a crude bound,
\begin{align*}
\vert  \phi(\bm\xi') \vert \leq \sum_{\xi'\in\{\bm\xi'\}} \frac{\vert\phi(\xi')\vert^{n(\bm\xi')} }{n(\bm\xi',\xi')} \leq  \big(2\dot\lambda_2\big)^{-n(\bm\xi')} \sum_{\xi'\in X\times\N} 1+\big(2\dot\lambda_\phi\vert\phi(\xi')\vert\big)^{2M-1}.
\end{align*}
Therefore, for $\Vert J\Vert_\infty\leq \kappa_J^{-1},s\in[0,1]^{{\sf P}(X)} $ and $\psi_x,\psi_x',K_x $ in the corresponding regions of integration,
\begin{align*}
\vert \eta_2''(x;\bm\xi;\bm\zeta)\vert &\leq  \dot\lambda_2^{-n(\bm\xi)}\dot\lambda_J^{-n(\bm\zeta)} \prod_{\xi'\in\bm\xi\circ\bm\zeta}e^{\tilde md(x,x')} \;\delta_{x\in X}\delta_{\bm\xi\circ\bm\zeta\subset X} \\&\qqquad \times    \sum_{\xi'\in X\times\N} e^{-m_Vd(x',x)}\left[ 1+\big(2\dot\lambda_\phi\vert\phi(\xi')\vert\big)^{2M-1}  \right]
\end{align*}
and so
\begin{align*}
\vert \eta_2''(x;\bm\xi;\bm\zeta)\vert & \in c_{\mf g}(m_V) \sum_{\xi'\in X\times\N} 1+\big(2\dot\lambda_\phi\vert\phi(\xi')\vert\big)^{2M-1}   \,\cdot\, \mc B_{2,\dot{\bm\lambda}} 
\end{align*}
We conclude
\begin{align*}
\left\vert\sum_{x\in X}V_2'(x;s_x\circ\phi,\psi_x+\psi_x';J+K_x)  \right\vert\leq \left[\vert X\vert +  c_{\mf g}(m_V) \sum_{\xi'\in X\times\N} 1+ \big(2\dot\lambda_\phi\vert\phi(\xi')\vert\big)^{2M-1} \right] \; \bm\vert V_2\bm\vert_{2,\dot{\bm\lambda}}
\end{align*}
Using (\ref{condpos}), we conclude that, for $(\phi;s;J)\in\mc D_{Q,\bm\lambda} $,
\begin{align*}
\Big\vert \prod_{x_1\in X_1}A'(x_1;&\alpha_{x_1};\phi;s;J ) \Big\vert\leq \ddot\lambda_\phi^{n(\bm\xi_1)} \ddot\lambda_J^{n(\bm\zeta_1)}  \tilde\lambda_s^{2\vert F_1\vert}e^{-\tilde md(F_1)} \\& \times \exp\left(\left[2\vert X\vert +  \vert Q\vert \cdot\omega(R)\cdot c_{\sf g}'\left(\frac{\log \omega(R)}{m_V}\right)\right]\; \bm\vert V_1\bm\vert_{1,\dot{\bm\lambda}}\right) \\ &\times\prod_{x\in X} \ddot \lambda_\phi^{n(\bm\xi_x)} \ddot \lambda_J^{n(\bm\zeta_x)}e^{-m_Vd(x,\bm\xi_x\circ\bm\zeta_x)} \frac{\delta_{F_x=\{x,N_x\}}}{\vert F_x\vert!} \prod_{x'\in N_x} \frac{2}{\sf N}\vert \phi(x')\vert_1\tilde\lambda_s^2e^{-\tilde md(x,x')}\\&  \times \exp\left( c_{pos}'\cdot \vert X\vert +\left[\vert X\vert +   \sum_{\xi'\in X\times\N} P\big(\dot\lambda_\phi \phi(\xi')\big) \right] \; \bm\vert V_2\bm\vert_{2,\dot{\bm\lambda}}\right)
\end{align*}
where $d(x,\bm \xi) = \sum_{\xi'\in\bm\xi}d(x,x') $ and
\begin{align*}
P(u) &= c_{\mf g}(m_V)\left[1+\big(2\vert u\vert\big)^{2M-1}\right]-c_{pos}u^{2M}
\end{align*}
is bounded above by $\tilde c_{\Text{P.\ref{boundsintint}}} = c_{\mf g}(m_V) + \frac{c_{\mf g}(m_V)^{2M}}{c_{pos}^{2M-1}} 2^{4M^2}  $. We use this to eliminate $ P\big(\dot\lambda_\phi \phi(\xi')\big)$ and obtain
\begin{align*}
\left\vert \prod_{x_1\in X_1}A'(x_1;\alpha_{x_1};\phi;s;J ) \right\vert&\leq \ddot \lambda_\phi^{n(\bm\xi)}\tilde\lambda_s^{2\vert F\vert}\ddot \lambda_J^{n(\bm\zeta)} e^{-\tilde md(F)} \\&\qquad\times\prod_{x\in X} \frac{e^{-m_Vd(x,\bm\xi_x\circ\bm\zeta_x)}\delta_{F_x=\{x,N_x\}}}{\vert F_x\vert!}\prod_{x'\in N_x} \frac{2}{\sf N}\vert \phi(x')\vert_1 \\&\qquad\times \exp\left[\vert X\vert \cdot\left[3+{\sf N}\,\tilde c_{\Text{P.\ref{boundsintint}}}  +c_{pos}'\right]+  \vert Q\vert \cdot\omega(R)\cdot c_{\sf g}'\left(\frac{\log \omega(R)}{m_V}\right)\right]
\end{align*}
We have
\begin{align*}
\sum_{\substack{\alpha_{x_1'}\in\aleph_\phi  \\ \circ_{x_1'} \alpha_{x_1'}=\alpha }} & c\big((\alpha_{x_1'})_{x_1'\in X_1};\alpha\big)\prod_{x\in X} \frac{e^{-m_Vd(x,\bm\xi_x\circ\bm\zeta_x)}\delta_{F_x=\{x,N_x\}}}{\vert F_x\vert!}\prod_{x'\in N_x} \frac{2}{\sf N}\vert \phi(x')\vert_1 \\ 
&\leq \sum_{\substack{m_x,n_x\geq 0   \\ \sum_{x}m_x\leq n(\bm\xi) \\ \sum_{x}n_x\leq n(\bm\zeta) }}  \prod_{x\in X} \frac1{m_x!n_x!}  \left[ \sum_{\xi'\in \{\bm\xi\circ\bm\zeta\}} e^{-m_Vd(x,x')}\right]^{m_x+n_x} \\&\qqquad\times \prod_{x\in X}\frac1{d^F(x)!}\sum_{F'\subset F} \prod_{\{x,x'\}\in F'} \frac2{\sf N} \big(\vert\phi(x)\vert_1+\vert\phi(x')\vert_1\big)\\
&\leq e^{2{\sf N}c_{\sf g}(m_V)\cdot\vert X\vert } \times \upmu_{\leq1}^{-\vert F\vert} \times \prod_{x\in X}\frac{\left[ 1+\frac 2{\sf N} \upmu_{\leq1}\vert\phi(x)\vert_1\right]^{d^F(x)}}{d^F(x)!} .
\end{align*}
We have, for any $\phi\in \Xi $
\begin{align*}
\mf G_{Q,\bm\lambda}(X;\phi;s;J)\prod_{x\in X}&\frac{\left[ 1+\frac 2{\sf N}\upmu_{\leq1}\vert\phi(x)\vert_1\right]^{d^F(x)}}{d^F(x)!}\\& \leq \exp\left(-\frac{\upmu}{16} \sum_{x\in X}\phi(x)^2\right) \prod_{x\in X} \frac{1}{d^F(x)!} \sqrt{1+\frac{32 d^F(x)}{\sf N}}^{d^F(x)} e^{ -\frac34 d^F(x)}\\
&\leq \exp\left(-\frac{\upmu}{16} \sum_{x\in X}\phi(x)^2\right) \prod_{x\in X} \frac{ 6^{d^F(x)}}{\sqrt{d^F(x)!}}
\end{align*}
By (\ref{aspos}), for $\phi\in \mc D_{Q,\bm\lambda} $
\begin{align*}
\exp\left(-\frac{\upmu}{16} \sum_{x\in X}\phi(x)^2\right)\leq \exp\left[-  \vert Q\vert \cdot\omega(R)\cdot c_{\sf g}'\left(\frac{\log \omega(R)}{m_V}\right)   \right]
\end{align*}
Putting everything together,
\begin{align*}
\bm\Vert A'\bm\Vert_{\phi,\bm\lambda} &= \sup_{\eta\in \mc B_{\phi,\bm\lambda}} \sup_{Q\subset \mb L} \sum_{\alpha\in\aleph_{\phi}} \sup_{\Psi\in\mc D_{Q,\bm\lambda}} \mf G_{Q,\bm\lambda} (\varsigma;\Psi) \vert \eta(\alpha)A'(\alpha;\Psi)\vert\\
&\leq \sup_{\eta\in \mc B_{\phi,\bm\lambda}}  \sum_{\alpha\in\aleph_{\phi}} \ddot \lambda_\phi^{n(\bm\xi)} 32^{\vert F\vert}\upmu_{\leq 1}^{-\vert F\vert} \tilde\lambda_s^{2\vert F\vert}\ddot  \lambda_J^{n(\bm\zeta)} e^{-\tilde md(F)} e^{\vert X\vert [3+{\sf N} \tilde c_{\Text{P.\ref{boundsintint}}} + c_{pos}' + 2{\sf N}c_{\sf g}(m_V) ]} \vert \eta(\alpha)\vert \\& <1
\end{align*}
by the assumption on $\ddot\lambda_X $ with $c_{\Text{P.\ref{boundsintint}}} =  e^{3+{\sf N} \tilde c_{\Text{P.\ref{boundsintint}}} + c_{pos}' + 2{\sf N}c_{\sf g}(m_V) } $ and the Properties 1 of the norms. This proves the claim.

\end{proof}

\subsection{Cluster Expansion}

\begin{prop}\label{boundce}
Let $A'\in \mb C^{\mc S_\phi\times \mc F_\phi} $ be a family, indexed by $X\in\mc S_\phi $, of analytic functions of $\Psi= (\phi;s;J)\in \mc F_\phi$. As in Theorems \ref{algebra} and \ref{algebralfsf}, define
\begin{align*}
A(X;J\vert_X) &= \sum_{T\in\mf T(X)} \bm\int\limits_T \ud\phi\vert_X\ud s\vert_X A'(X;\phi\vert_X;s^T;J\vert_X)\\
B(X,Q;J\vert_X) &= \sum_{T\in\mf T(X)} \bm\int\limits_{T,Q} \ud\phi\vert_X\ud s\vert_X A'(X;\phi\vert_X;s^T;J\vert_X).
\end{align*}
Let the transformation $\Lambda^2 $ be the transformation on the space $\bm\Lambda $ of parameters from Property 2, set $\dot{\bm\lambda} =\Lambda^2\bm\lambda  $. Assume that $\bm\Vert A'\bm\Vert_{\phi,\dot{\bm\lambda} }<\infty $. Then $\bm\vert A\bm\vert_{J,\bm\lambda},\bm\vert A_s\bm\vert_{J,\bm\lambda},\bm\vert B\bm\vert_{J,\bm\lambda}'\leq \bm\Vert A'\bm\Vert_{\phi,\dot{\bm\lambda} }. $
\end{prop}

\begin{proof}
We prove the statement about $B $ only (the proof for $A $, $A_s $ is simpler and largely identical). For a forest $F\in\mf F(X),X\subset \mb L $, we define the set of legs of $F$,
\begin{align*}
\mc L(F) &= \Big\{\big(\ell,x\big),x\in \ell,\ell\in F  \Big\} ,
\end{align*}
and, for its subsets, the projection onto their second component (as a multiset)
\begin{align*}
\mc L\vert_2 &= \Big\{\Big(x,(\ell,x)\in \mc L  \Big)\Big\}
\end{align*}
For $\ul{\sf n}\in \N^{\mc L(F)} $ an assignment of an index $\ul{\sf n}((\ell,x))\in\N $ to every leg of $F$, we define $(\ell,x)^{\ul{\sf n}} = (\ell,(x,\ul{\sf n}(\ell,x))) $ and $\ell^{\ul{\sf n}} = \big\{(x,\ul{\sf n}(\ell,x)),(x',\ul{\sf n}(\ell,x')) \big\} $ for $\ell=\{x,x'\}\in F $. We also define, for every $\mc L\subset \mc L(F) $ $$\mc L\vert_2^{\ul{\sf n}} = \Big\{\Big((x,\ul{\sf n}(\ell,x)),(\ell,x)\in \mc L  \Big)\Big\}\in\bd L. $$ Note that
\begin{align*}
\prod_{\ell\in F'} C(\ell)\partial_\ell &= \sum_{\ul{\sf n}\in \N^{\mc L( F')}} \prod_{\ell\in F'} C(\ell^{\ul{\sf n} })  \prod_{\xi\in\mc L(F')\vert_s^{\ul{\sf n}} } n(\mc L(F')\vert_s^{\ul{\sf n}},\xi)! \cdot \nabla_{\mc L(F')\vert_s^{\ul{\sf n}} }
\end{align*}
In the expression for $A$, a subset $\mc L_1\vert_2^{\ul{\sf n}}\subset\mc L(F')\vert_s^{\ul{\sf n}}   $ hits $A'$ and the remaining derivatives hit the characteristic functions. For each $x\in \{(\mc L(F')\setminus\mc L_1)\vert_2\}=: \mc L_4\vert_2 $ where the latter happened, we interpret multiple derivatives of a single characteristic function weakly and ``integrate by parts''. Because of the ultralocal structure of the characteristic functions, the derivatives coming from integration by parts can either hit $A' $, or the Gaussian (call the corresponding set of legs $ \mc L_2$ and $\mc L_3$, respectively). We evaluate the latter with the Cauchy formula (the radius $\tilde r = \sqrt{\frac{\upmu\lambda_{\Text{min}}(\Re C) }{64{\sf N}} }r $ of the corresponding circle of integration is optimized in Proposition \ref{propprop2}) and obtain
\begin{align*}
\bm\int\limits_{T,Q} \ud\phi\vert_X\ud s\vert_X A'(X;\phi\vert_X;s^T;J\vert_X) &=  \sum_{\substack{ \bm\xi\in \bd L\vert_X  \\ F \in\mf F(X) }} \int_0^1 \ud\bd s^T \int_{\Xi^X} \prod_{x\in X}\ud\phi(x)\,\tilde \gamma\big(T;X,Q;\bm\xi;F;\phi;s^T\big)\\&\qquad\times \mf G_{Q,\bm\lambda}(X;\phi\vert_X;s^T;J\vert_X)\nabla_{\phi,\bm\xi}\nabla_{s,F} A'(X;\phi\vert_X;s^T;J\vert_X)
\end{align*}
with
\begin{align*}
\tilde \gamma &= \delta_{F\subset T} \prod_{x\in X}{\sf N}^{d^F(x)}d^F(x)!d^{T\setminus F}(x)! \sum_{\ul{\sf n}\in \N^{\mc L(T\setminus F)} }    \prod_{x\in X}\frac{\prod_{{\sf n}\in \N} n\big(\mc L(T\setminus F)\vert_2^{\ul{\sf n}},(x,{\sf n})\big)!}{d^{T\setminus F}(x)!} \\&\qquad \times\prod_{\ell\in T\setminus F} C(\ell^{\ul{\sf n} })  \sum_{\substack{ \mc L_1\dot\cup\cdots\dot\cup\mc L_4 = \mc L(T\setminus F)  \\ \mc L_4\vert_2 = \{\mc L_2\vert_2\circ\mc L_3\vert_2\} \\ \mc L_1\vert_2^{\ul{\sf n}}\circ\mc L_2\vert_2^{\ul{\sf n}} = \bm\xi }} (-)^{\vert \mc L_2\circ\mc L_3\vert} \prod_{\xi\in \{\mc L_3\vert_2^{\ul{\sf n}} \}}\frac1{2\pi i} \oint\limits_{\vert \psi(\xi)\vert = \tilde r}\frac{\ud\psi(\xi)}{\psi(\xi)^{1+n(\mc L_3\vert_2^{\ul{\sf n}} ,\xi)}}\\ &\qquad\times \prod_{\xi\in \mc L_4\vert_2^{\ul{\sf n}} } \frac{\phi(\xi)}{\vert\phi(x)\vert} \chi_{Q,\mc L_4\vert_2}(\phi) \cdot \det (2\pi C_{s^T})^{-\half}\exp\left(-\half \big\langle \phi+\psi ,C_{s^T}^{-1} (\phi+\psi) \big\rangle + \frac{\upmu}{8} \Vert\phi\vert_X\Vert_2^2\right) 
\end{align*}
where (note again that $\mc L_4\vert_2 $ is by construction a multiset without repeated elements)
\begin{align*}
\chi_{Q,S}(\phi) &= \chi_{Q^c\setminus S}(\phi)\chi_{Q\setminus S}^c(\phi)\prod_{x\in Q^c\cap S} \delta\Big(\vert\phi(x)\vert=r \Big) \prod_{x\in Q^c\cap S} \left[\delta\Big(\vert\phi(x)\vert= R \Big) - \delta\Big(\vert\phi(x)\vert= r \Big)\right]
\end{align*}
The delta functions restrict the integration of their variables to the sphere, with the usual surface measure (with appropriate signs if ${\sf N} =1$). We now have
\begin{align*}
\bm\vert B\bm\vert_{J,\bm\lambda} &\leq \sup_{\eta\in\mc B_{J,\bm\lambda}' } \sum_{(X,Q;\bm\zeta)\in \aleph_{J}'} \sup_{ \Vert J\Vert_{\infty}\leq\kappa_J^{-1}}  \vert \eta(X,Q;\bm\zeta)\vert\sum_{T\in\mf T(X)} \sum_{\substack{ \bm\xi\in \bd L\vert_X  \\ F \in\mf F(X) }}\\&\qquad\times\Bigg\vert  \int_0^1 \ud\bd s^T \int_{\Xi^X} \prod_{x\in X}\ud\phi(x)\,\tilde \gamma\big(T;X,Q;\bm\xi;F;\phi;s^T\big)  \\&\qqquad\times\mf G_{Q,\bm\lambda}(X;\phi\vert_X;s^T;J\vert_X)\nabla_{\phi,\bm\xi}\nabla_{s,F}\nabla_{J,\bm\zeta} A'(X;\phi\vert_X;s^T;J\vert_X)\Bigg\vert\\
&\leq \sup_{\eta\in\mc B_{J,\bm\lambda}' } \sum_{(X,Q;\bm\zeta)\in \aleph_{J}'}   \sum_{\substack{ \bm\xi\in \bd L\vert_X  \\ F \in\mf F(X) }}\vert \eta(X,Q;\bm\zeta)\vert\sum_{T\in\mf T(X)}  \gamma' \big(T;X,Q;\bm\xi;F\big)  \\&\qquad\times \sup_{(\phi;s;J)\in \mc D_{Q,\bm\lambda}} \mf G_{Q,\bm\lambda}(X;\phi\vert_X;s;J\vert_X)\Big\vert\nabla_{\phi,\bm\xi}\nabla_{s,F}\nabla_{J,\bm\zeta} A'(X;\phi\vert_X;s;J\vert_X)\Big\vert\\
&\leq \bm\Vert A'\bm\Vert_{\phi,\dot{\bm\lambda}}
\end{align*}
by Property 2 with 
\begin{align*}
\gamma' \big(T;X,Q;\bm\xi;F\big) &=  \int_0^1 \ud\bd s^T \int_{\Xi^X} \prod_{x\in X}\ud\phi(x) \;\Big\vert \tilde \gamma\big(T;X,Q;\bm\xi;F;\phi;s^T\big)\Big\vert.
\end{align*}
This proves the claim about $B$. The proof for $A,A_s $ is essentially the same, using 
\begin{align*}
\gamma \big(T;X;\bm\xi;F\big) &=  \delta_{\vert X\vert\geq 2}\int_0^1 \ud\bd s^T \int_{\Xi^X} \prod_{x\in X}\ud\phi(x) \;\Big\vert \tilde \gamma\big(T;X;\bm\xi;F;\phi;s^T\big)\Big\vert
\end{align*}
with $\tilde \gamma $ as above, with $r=R $ (which implies $Q=\emptyset $).

\end{proof}

\begin{rem}\label{remnogenchar}
The proof of Proposition \ref{boundce} is limited to ``ultra-local'' characteristic functions because the integration by parts procedure that removes all but one derivative per point $x\in\mb L $ from the characteristic functions does not generalize in an easy way to the ``non-ultra-local'' case. We leave this for future work.
\end{rem}

\erem
We also need a bound for the purely local activity $A\big(\{x\};J(x)\big) $ and its logarithm:
\begin{prop}\label{blocact}
Set
\begin{align*}
\mu_{C(x,x)}(B_r) &= \int\limits_{\vert\phi(x)\vert\leq r}\ud\mu_{C(x,x)}(\phi(x)) & \vert \mu_{C(x,x)}\vert &= \int_{\Xi}\ud\vert\mu_{C(x,x)}\vert(\phi(x)) ,
\end{align*}
and assume that
\begin{align}\label{poslocact}
e^{-\frac{\upmu}4 r^2} \leq  2^{-\sf N}\left[\max_{x\in\mb L}\vert \mu_{C(x,x)}\vert\right]^{-1} e^{-c_{\Text{P.\ref{blocact}}}}.
\end{align}
where $ c_{\Text{P.\ref{blocact}}}\geq 1$ is some constant defined in the proof\footnote{This assumption implies in particular $\vert \mu_{C(x,x)}(B_r)-1\vert\leq e^{-\frac{\upmu}4 r^2-1}  $}. Let $V_1,V_2 $ satisfy
\begin{align}\label{smallness}
\bm\vert V_1\bm\vert_{1,\dot{\bm\lambda}}  +\bm\vert V_2\bm\vert_{2,\dot{\bm\lambda}}< \frac16\left[\max_{x\in\mb L}\vert \mu_{C(x,x)}\vert\right]^{-1} =: c_{\Text{(\ref{smallness})}}
\end{align}
with a $\dot{\bm\lambda} $ so that $\dot\lambda_J^{-1} \geq \lambda_J^{-1}+\kappa_J^{-1} $ and\footnote{This relation between the size of $V_2$ and $r$ is not invariant under scaling (as opposed to all others, such as (\ref{poslocact}), which only depend on $V_2$ through $c_{pos},c_{pos}' $). Like in the proof of Proposition \ref{boundsintint}, it could be replaced by a scale invariant relation, plus an absolute smallness condition on $\dot\lambda_\phi $, obtained, e.g., from Wick's rule for the Gaussian integral of a polynomial. This is somewhat more complicated, and turns out to have no benefit, so we omit it.}
\begin{align}\label{relvr}
\dot\lambda_2^{-1}\geq r .
\end{align}
Assume also (\ref{aspos}). Define the norm
\begin{align*}
 \Vert F(J(x))\Vert_{\lambda_J} := \sum_{\substack{\bm\zeta\in\bd L\\ \supp\bm\zeta\subset \{x\} } } \lambda_J^{-n(\bm\zeta)}\big\vert \nabla_{J,\bm\zeta}F(J(x))\big\vert .
\end{align*}
Then, whenever $\vert J(x)\vert\leq \kappa_J^{-1} $, we have 
\begin{align*}
\left\Vert  A_s\big(\{x\};J(x)\big)- \mu_{C(x,x)}(B_r) \right\Vert_{\lambda_J} &\leq e\vert \mu_{C(x,x)}\vert\Big(\bm\vert V_1\bm\vert_{1,\dot{\bm\lambda}}  +\bm\vert V_2\bm\vert_{2,\dot{\bm\lambda}}\Big)  \\
\left\Vert  A\big(\{x\};J(x)\big)- \mu_{C(x,x)}(B_r) \right\Vert_{\lambda_J} &\leq e\vert \mu_{C(x,x)}\vert\Big(\bm\vert V_1\bm\vert_{1,\dot{\bm\lambda}}  +\bm\vert V_2\bm\vert_{2,\dot{\bm\lambda}}\Big)  + e^{-\frac\upmu4r^2}
\end{align*}
In particular,
$$ A\big(\{x\};J(x)\big)^{-1},\log A\big(\{x\};J(x)\big),A_s\big(\{x\};J(x)\big)^{-1}\Text{ and }\log A_s\big(\{x\};J(x)\big) $$ exist and 
\begin{align*}
\big\Vert  A_s\big(\{x\};J(x)\big)^{-1} \big\Vert_{\lambda_J} &\leq \left[\vert \mu_{C(x,x)}(B_r)\vert - e\vert \mu_{C(x,x)}\vert\Big(\bm\vert V_1\bm\vert_{1,\dot{\bm\lambda}}  +\bm\vert V_2\bm\vert_{2,\dot{\bm\lambda}}\Big) \right]^{-1}    \\   
\big\Vert  A\big(\{x\};J(x)\big)^{-1} \big\Vert_{\lambda_J} &\leq\left[ \vert \mu_{C(x,x)}(B_r)\vert - e\vert \mu_{C(x,x)}\vert\Big(\bm\vert V_1\bm\vert_{1,\dot{\bm\lambda}}  +\bm\vert V_2\bm\vert_{2,\dot{\bm\lambda}}\Big) -e^{-\frac\upmu4r^2}  \right]^{-1}     \\
\big\Vert\log A_s\big(\{x\};J(x)\big) \big\Vert_{\lambda_J} &\leq   \Big\vert \log  \mu_{C(x,x)}(B_r)\Big\vert \\&\qquad + \log 1- \frac{e\vert \mu_{C(x,x)}\vert}{\vert\mu_{C(x,x)}(B_r)\vert } \Big(\bm\vert V_1\bm\vert_{1,\dot{\bm\lambda}}  +\bm\vert V_2\bm\vert_{2,\dot{\bm\lambda}}\Big)\\
\big\Vert\log A\big(\{x\};J(x)\big) \big\Vert_{\lambda_J} &\leq   \Big\vert \log  \mu_{C(x,x)}(B_r)\Big\vert \\&\qquad + \log 1- \frac{e\vert \mu_{C(x,x)}\vert}{\vert\mu_{C(x,x)}(B_r)\vert } \Big(\bm\vert V_1\bm\vert_{1,\dot{\bm\lambda}}  +\bm\vert V_2\bm\vert_{2,\dot{\bm\lambda}}\Big) -\frac{e^{-\frac\upmu4r^2}}{\vert\mu_{C(x,x)}(B_r)\vert}  
\end{align*}

\end{prop}
\begin{proof}
Note that 
\begin{align*}
A_s\big(\{x\};J(x)\big) &= \mu_{C(x,x)}(B_r) + A_{0,r}\big(\{x\};J(x)\big)\\
A\big(\{x\};J(x)\big) &= \mu_{C(x,x)}(B_r)+ A_{0,r}\big(\{x\};J(x)\big) + A_{r,R}\big(\{x\};J(x)\big)
\end{align*}
with
\begin{align*}
A_{0,r}\big(\{x\};J(x)\big) &= \int\limits_{0\leq\vert \phi(x)\vert\leq r} \ud \mu_{C(x,x)} \big(\phi(x)\big) \Big[A'\big(x;\phi(x);J(x)\big)-1\Big]\\
A_{0,r}\big(\{x\};J(x)\big) &= \int\limits_{r<\vert \phi(x)\vert\leq R} \ud \mu_{C(x,x)} \big(\phi(x)\big) A'\big(x;\phi(x);J(x)\big)\\
A'\big(x;\phi(x);J(x)\big)&= \exp\left(\sum_{\supp\bm\xi,\supp\bm\zeta\subset\{x\}}\big(v_1-v_2\big)(\bm\xi,\bm\zeta) \phi(\bm\xi)J(\bm\zeta)    \right)  
\end{align*}
Here, we set $v_2 (\bm\xi,\bm\zeta) = 0  $ if $n(\bm\xi\circ\bm\zeta)>2M $. Let $\delta = 0 $ or $1$. We have
\begin{align*}
\Bigg\Vert\int\limits_{r<\vert \phi(x)\vert\leq r'} \ud \mu_{C(x,x)}& \big(\phi(x)\big) \Big[A'\big(x;\phi(x);J(x)\big) -\delta\Big] \Bigg\Vert_{\lambda_J} \\& \leq \int\limits_{r<\vert \phi(x)\vert\leq r'} \ud \vert\mu_{C(x,x)}\vert \big(\phi(x)\big) \Vert  A'\big(x;\phi(x);J(x)\big)  -\delta \Vert_{\lambda_J} ,
\end{align*}
where in the norm on the right hand side, we regard $\phi(x) $ as a parameter. By the product rule,
\begin{align*}
\Vert  A'\big(x;\phi(x);J(x)\big) -\delta \Vert_{\lambda_J} \leq \exp\left(\Re\!\!\!\sum_{\supp\bm\xi \subset\{x\} }\!\!\!\big(v_1-v_2\big)(\bm\xi,-)\phi(\bm\xi) +  \Vert V'(x;\phi(x);J(x))\Vert_{\lambda_J} \right)-\delta
\end{align*}
Here, $$ V'(\phi(x)) =\sum_{\supp\bm\xi\subset\supp\bm\zeta=\{x\} }\big(v_1-v_2\big)(\bm\xi,\bm\zeta) \phi(\bm\xi)J(\bm\zeta) .$$ In the case of $A_{0,r} $, we conclude by the assumptions
\begin{align*}
\sup_{\vert \phi(x)\vert\leq r} \Vert  A'\big(x;\phi(x);J(x)\big)-1  \Vert_{\lambda_J} \leq   \Big(\bm\vert V_1\bm\vert_{1,\dot{\bm\lambda}}  +\bm\vert V_2\bm\vert_{2,\dot{\bm\lambda}}\Big)\exp\Big(\bm\vert V_1\bm\vert_{1,\dot{\bm\lambda} }  +  \bm\vert V_2\bm\vert_{2,\dot{\bm\lambda}}  \Big)
\end{align*}
and so $$\Vert A_s(\{x\};J(x))\Vert_{\lambda_J}\leq e\vert \mu_{C(x,x)}\vert \Big(\bm\vert V_1\bm\vert_{1,\dot{\bm\lambda}}  +\bm\vert V_2\bm\vert_{2,\dot{\bm\lambda}}\Big)  .$$ For $A_{r,R} $, we have, as in the proof of Proposition \ref{boundsintint},
\begin{align*}
\sup_{r<\vert \phi(x)\vert\leq R}\Vert  A'\big(x;\phi(x);J(x)\big)  \Vert_{\lambda_J} \leq  \exp\Big[\omega(R)\cdot \bm\vert V_1\bm\vert_{1,\dot{\bm\lambda} }  + \tilde c_{\Text{P.\ref{blocact}}}  \bm\vert V_2\bm\vert_{2,\dot{\bm\lambda}} + c_{pos}'\Big]
\end{align*}
with $\tilde c_{\Text{P.\ref{blocact}}} = 1+\frac{2^{4M^2}} {c_{pos}^{2M-1}}. $ Using (\ref{aspos}),
\begin{align*}
\lambda_{\Text{min}}(\Re C(x,x)^{-1})\geq \upmu,
\end{align*}
(cf. Lemma \ref{lemhadamard}), and that $ \vert \mu_{C(x,x)}\vert $ is the ratio of two determinants, this proves 
\begin{align*}
\Vert A_{r,R}\big(\{x\};J(x)\big)\Vert_{\lambda_J}\leq  e^{-\frac{\upmu}{2}r^2 } \vert \mu_{C(x,x)}\vert e^{c_{\Text{P.\ref{blocact}}} }.
\end{align*}
with $c_{\Text{P.\ref{blocact}}} =  4\log {\sf N}+\tilde c_{\Text{P.\ref{blocact}}} + c_{pos}' $. The Proposition follows.

\end{proof}

\subsection{Mayer expansion}

We start the bounds needed for the Mayer expansion of (\ref{mealg}) and Lemma \ref{sflfme} by controlling the normalization $A\to\dot A $.

\begin{prop}\label{bnormalization}
Let $A(\varsigma;J) $ be a family, indexed by $\varsigma\in \mc S_J^\bullet $, of analytic functions of $J \in \Xi_{\mb C}^{\mb L}$, and let $a\big(\{x\};J(x)\big) $ be so that, for every $x\in\mb L $, $\sup_{\vert J(x)\vert\leq\kappa_J^{-1}}\Vert a\big(\{x\},J(x)\big)\Vert_{\lambda_J}\leq c^3_{\Text{ii}}. $ Define, as in (\ref{mealg}),
\begin{align*}
\dot A(\varsigma;J) &= A(\varsigma;J)\cdot\prod_{x\in X}  a\big(\{x\},J(x)\big).
\end{align*}
Then
\begin{align*}
\bm\vert \dot A\bm\vert_{J,\bm\lambda}^\bullet \leq \bm\vert  A\bm\vert_{J,\Lambda^3_{\Text{ii}}\bm\lambda}^\bullet
\end{align*}
\end{prop}

\begin{proof}
This is obvious from the product rule and Property 3 (ii).
\end{proof}

\noindent
Next come the bounds for the actual Mayer expansion:

\begin{prop}\label{bme}
(i) Let $\dot A(X;J) $ be a family, indexed by $X\subset\mb L$, of analytic functions of $J\in\Xi_{\mb C}^{\mb L} $, such that $\bm\vert \dot A\bm\vert_{J,\dot{\bm\lambda}} <c_J(\bm\lambda)^{-1} $, where $\dot{\bm\lambda} = \Lambda_J\bm\lambda $ with $\Lambda_J $ and $c_J $ from Property 3 (i). Define the family $W(X;J) $ ($\vert X\vert\geq 2 $) by (\ref{mealg}). Then
\begin{align*}
\bm\vert W\bm\vert_{J,\bm\lambda}\leq \frac18 \frac{\bm\vert \dot A\bm\vert_{J,\dot{\bm\lambda}}}{c_J(\bm\lambda)^{-1}-\bm\vert \dot A\bm\vert_{J,\dot{\bm\lambda}}}.
\end{align*}
(ii) Let $\mc V(Z;J) $ be a family, indexed by $Z\subset\mb L $, and $B(X,Q;J) $ a family, indexed by $(X,Q)\in {2^{\mb L}}' $, of analytic functions of $J\in\Xi_{\mb C}^{\mb L} $. Assume that $\bm\vert \mc V\bm\vert_{J,\dot{\bm\lambda}}  + \bm\vert B\bm\vert'_{J,\dot{\bm\lambda}}<\tilde c_J(\bm\lambda)^{-1} $, where $\dot{\bm\lambda} = \Lambda^3_{\Text{iii}}\tilde\Lambda_J\bm\lambda $. Define the family $\dot B(Z,X,Q;J) $ as in Lemma \ref{sflfme}. Then 
\begin{align*}
\bm\vert \dot B\bm\vert^{\sim}_{J,\bm\lambda}\leq \frac18 \frac{\bm\vert \mc V\bm\vert_{J,\dot{\bm\lambda}} + \bm\vert B\bm\vert'_{J,\dot{\bm\lambda}}  }{\tilde c_J(\bm\lambda)^{-1}-\bm\vert \mc V\bm\vert_{J,\dot{\bm\lambda}} -\bm\vert B\bm\vert'_{J,\dot{\bm\lambda}} }.
\end{align*}
(iii) Let $\dot B(Z,X,Q;J) $ be a family, indexed by $(Z,X,Q)\in \tilde 2^{\mb L}$, of analytic functions of $J\in\Xi_{\mb C}^{\mb L} $, such that $\bm\vert \dot B\bm\vert^{\sim}_{J,\dot{\bm\lambda}} < \tilde c_J(\bm\lambda)^{-1} $ with $\dot{\bm\lambda}=\tilde \Lambda_J\bm\lambda $. Define the family $\mc L(Z,X,Q;J) $ as in Lemma \ref{sflfme}. Then
\begin{align*}
\bm\vert \mc L\bm\vert^{\sim}_{J,\bm\lambda}\leq \frac18 \frac{\bm\vert \dot B\bm\vert^{\sim}_{J,\dot{\bm\lambda}}}{\tilde c_J(\bm\lambda)^{-1}-\bm\vert \dot B\bm\vert^{\sim}_{J,\dot{\bm\lambda}}}.
\end{align*}
\end{prop}

\begin{proof}
(i) By Lemma \ref{lemursell}, we have
\begin{align*}
\big\vert \rho\big(\{X_m\}_1^n\big)\big\vert &\leq \sum_{T\subset \mc G(\{X_m\}_1^n)} \prod_{\{m,m'\}\in T}\delta_{X_m\cap X_{m'}\neq\emptyset}.
\end{align*}
Therefore
\begin{align*}
\bm\vert W\bm\vert_{J,\bm\lambda} &\leq \sup_{\eta\in \mc B_{J,\bm\lambda}} \sum_{n\geq 1}\frac1{n!} \sum_{T\in\mf T(\ul n)}\sum_{\substack{X_m\subset\mb L \\ \bm\zeta_m\in\bd L\vert_{X_m} \\ m\in \ul n}}  \left\vert \eta\big(\cup_m X_m,\circ_m\bm\zeta_m\big) \prod_{\{m,m'\}\in T}\delta_{X_m\cap X_{m'}\neq\emptyset} \right\vert \\&\qqquad\qqquad\times \prod_{m=1}^n \sup_{\Vert J\Vert_\infty\leq\kappa_J^{-1}}\vert\nabla_{J,\bm\zeta_m} \dot A(X_m;J)\vert\\
&\leq \sum_{n\geq 1}\frac1{n!} \sum_{T\in\mf T(\ul n)}\sup_{\eta\in \mc B_{J,\bm\lambda}(T)}\sum_{\substack{X_m\subset\mb L \\ \bm\zeta_m\in\bd L\vert_{X_m} \\ m\in \ul n}}  \vert \eta\big((X_m,\bm\zeta_m)_1^n\big) \vert   \prod_{m=1}^n \sup_{\Vert J\Vert_\infty\leq\kappa_J^{-1}}\vert\nabla_{J,\bm\zeta_m} \dot A(X_m;J)\vert \\
&\leq  \sum_{n\geq 1}\frac1{n!}\left(\frac{c_J(\bm\lambda)}8 \bm\vert \dot A\bm\vert_{J,\dot{\bm\lambda}} \right)^n \sum_{T\in\mf T(\ul n)} \prod_{m=1}^n d^T(m)! \\
&\leq \frac18 \frac{\bm\vert \dot A\bm\vert_{J,\dot{\bm\lambda}}}{c_J(\bm\lambda)^{-1}-\bm\vert \dot A\bm\vert_{J,\dot{\bm\lambda}}}
\end{align*}
as claimed. The argument for (iii) is identical. For (ii), we first regard $$\mc V(Z;J) \equiv \delta_{X,\emptyset}\delta_{Q,\emptyset}\mc V(Z;J)\qquad\text{and}\qquad\dot B(X,Q;J) \equiv \delta_{Z,\emptyset}\dot B(X,Q;J) $$ as families, indexed by $(Z,X,Q)\in \tilde 2^{\mb L} $. By Property 3 (iii), their $\bm\vert \cdot\bm\vert^\sim_{J,\bm\lambda} $ norms are bounded by their natural $\bm\vert\cdot\bm\vert^\bullet_{J,\Lambda^3_{\Text{iii}}\bm\lambda } $ norms ($\bullet =  $ void or $\phantom{}' $). We then use the a crude bound of the type
\begin{align*}
\sum_{(\{Z_{m'}\}_1^{n'},\{(X_m,Q_m)\}_1^n )\in\mc C(Z,X,Q) } \cdots \leq \sum_{n+n'\geq 1} \frac1{n!n'!} \sum_{T\in\mf T(\ul{n+n'})} \sum_{\substack{Z_m\subset\mb L \\ (X_{m'},Q_{m'})\in {2^{\mb L}}' }}&\prod_{\{a,a'\}\in T}\delta_{Y_a\cap Y_{a'}\neq\emptyset} \cdots \\ Y_a &= \left\{\begin{array}{cc} X_a & a\leq n \\ Z_{a-n} & a>n \end{array} \right.
\end{align*}
and argue as before. This concludes the proof of the proposition.

\end{proof}

\subsection{The bound for the logarithm}
The bounds from the previous propositions and the algebra of Theorems \ref{algebra} and \ref{algebralfsf} combine into the following theorem on $\log\mc Z(J) $, whose proof is a meditation on the flow chart of Remark \ref{remflowchart}, the bounds of this section, and Properties 3 (iv) and (v).

\begin{thm}\label{thmabound}
Let $C\in\Text{End }\mb C^{\mb L\times\N} $ be a covariance with $\lambda_{\Text{min}}(\Re C^{-1})=:\upmu>0 $. Let $V_1(\phi;J)$ be a power series in its arguments, and let $V_2(\phi;J) $ be polynomial of maximal total degree $2M$, as in Proposition \ref{aipgr}. Suppose that $\bm\vert V_1\bm\vert_{1,\dot{\bm\lambda}} +\bm\vert V_2\bm\vert_{2,\dot{\bm\lambda}}<c_{\Text{(\ref{smallness})}}$ for some $\dot{\bm\lambda}\in\bm\Lambda $. Assume
\begin{enumerate}
\item The positivity conditions (\ref{poslocact}) and (\ref{stabassumr}) below for $r$, associated to the large field decomposition.
\item The stability condition (\ref{aspos}) associated to $V_1$, for some $R>r>0 $ and nondecreasing $ \omega:[r,R]\to\mb R$ with $\omega(r)\leq 1 $.
\item The stability condition (\ref{condpos}) associated to $V_2$, for some $c_{pos}>0,c_{pos}'\geq 0. $
\item The condition $r\leq \dot\lambda_2^{-1} $ of (\ref{relvr}).
\end{enumerate}
Fix $\kappa_J>0 $.\\
(i) For the cluster expansion without explicit large field decomposition, if:
\begin{align*}
\dot{\bm\lambda} &= \Lambda_{\Text{P.\ref{boundsintint}}}\bm\lambda\ob 1 &\Text{ for some }&\bm\lambda\ob 1 ; \\
\bm\lambda\ob 1 &= \Lambda^2\bm\lambda\ob 2 &\Text{ for some }&\bm\lambda\ob 2 ;\\
\bm\lambda\ob 2 &= \Lambda^3_{\Text{ii}}\bm\lambda\ob 3  &\Text{ for some }&\bm\lambda\ob 3  &&\Text{ (with }\Lambda^3_{\Text{ii}} \Text{ determined }\\
\bm\lambda\ob 3 &= \Lambda_J\bm\lambda\ob 4 &\Text{ for some }&\bm\lambda\ob 4 ; &&\qquad\Text{by the choice } c^3_{\Text{ii}} =  2 \sup_{x\in\mb L}\big\vert \mu_{C(x,x)}(B_r)\big\vert^{-1}\Text{)};\\
\bm\lambda\ob 4 &=\Lambda^3_{\Text{iv}}\bm\lambda &\Text{ for some }&\bm\lambda;
\end{align*}
and if we have
\begin{align*}
\dot\lambda\ob 1_X&\leq \half c_J(\bm\lambda\ob 4)^{-1}  c_{\Text{P.\ref{boundsintint}}}^{-1} &\Text{ with }&\dot{\bm\lambda}\ob1 = \Lambda^1\bm\lambda\ob 1 \\ 
\dot\lambda\ob 4_J&\geq \dot \lambda_{J} &\Text{ with }&\dot{\bm\lambda}\ob4 = \Lambda^3_{\Text{v}}\bm\lambda\ob4 
\end{align*}
then, for any $J $ with $\Vert J\Vert_\infty\leq \kappa_J^{-1} $, $\mc Z(J) $ as defined by (\ref{logz}) has a logarithm, and
\begin{align*}
\bm\vert \log \mc Z\bm\vert_{\bm\lambda} \leq   1 + \frac1{\dot\lambda\ob 4_X}\sup_{x\in\mb L}\Big\vert \log  \mu_{C(x,x)}(B_r)\Big\vert.
\end{align*}
(ii) For the cluster expansion with large field decomposition, assume the same as in (i), but with the stronger $\dot\lambda_X\ob1\leq \half \tilde c_J(\bm\lambda\ob 4)^{-1} c_J(\bm\lambda\ob 4)^{-1}  c_{\Text{P.\ref{boundsintint}}}^{-1} $  and $\bm\vert V_1\bm\vert_{1,\dot{\bm\lambda}} +\bm\vert V_2\bm\vert_{2,\dot{\bm\lambda}}<\half\tilde c_J(\bm\lambda\ob 4)^{-1} c_{\Text{(\ref{smallness})}} $. A fortiori, the pure small field partition function
\begin{align*}
\mc Z_s(J) &=  \int \ud\mu_C(\phi) \chi_r(\phi)e^{V(\phi;J)} = \exp\left(\sum_{Z\subset\mb L}\mc V(Z;J) \right)
\end{align*}
is analytic in $J$ for $\Vert J\Vert_\infty\leq\kappa_J^{-1} $ and has a logarithm with $\bm\vert \mc V\bm\vert_{J,\bm\lambda\ob4}<\half \tilde c_J(\bm\lambda\ob 4)^{-1} $. If: 
\begin{align*}
\dot{\bm\lambda} &=\Lambda_{\Text{P.\ref{boundsintint}}}\tilde{\bm\lambda}\ob 1&\Text{ for some }&\tilde{\bm\lambda}\ob 1  ; \\
\tilde{\bm\lambda}\ob 1 &= \Lambda^2\tilde{\bm\lambda}\ob 2 &\Text{ for some }&\tilde{\bm\lambda}\ob 2;\\
\tilde{\bm\lambda}\ob 2 &= \Lambda^3_{\Text{iii}}\tilde \Lambda_J\tilde{\bm\lambda}\ob 3  &\Text{ for some }&\tilde{\bm\lambda}\ob 3  ; \\
\tilde{\bm\lambda}\ob 3 &= \tilde\Lambda_J\tilde{\bm\lambda}\ob 4 &\Text{ for some }&\tilde{\bm\lambda}\ob 4  ; \\
\tilde{\bm\lambda}\ob 4 &=\Lambda^3_{\Text{iv}}\tilde{\bm\lambda} &\Text{ for some }&\tilde{\bm\lambda};
\end{align*}
and if we have 
\begin{align*}
\dot{\tilde\lambda}\ob 1_X&\leq \half\tilde c_J(\tilde{\bm\lambda}\ob 4)^{-1} c_J(\tilde{\bm\lambda}\ob 4)^{-1}  c_{\Text{P.\ref{boundsintint}}}^{-1} &\Text{ with }&\dot{\tilde{\bm\lambda}}\ob1 = \Lambda^1\tilde{\bm\lambda}\ob 1
\end{align*}
and, for simplicity, $\tilde{\bm\lambda}\ob2 = \bm\lambda\ob4 $; then
\begin{align*}
\bm\vert \log \mc Z - \log \mc Z_s\bm\vert_{\bm\lambda} \leq   e^{-\frac\upmu8r^2}.
\end{align*}

\end{thm}

\section{Choice of norms}

\subsection{Definition of test function spaces}

We now define explicitly the test function spaces $\mc B_{\bullet,\bm\lambda}^\bullet $ and thereby complete the definition of our norms from section 3. \\
In many cases, the Properties 1 - 3 could in principle be used as a definition of our set of test functions, but in order to make contact with the conventional $1-\infty $ type norms used in the literature, we now implement them in our setting. As a motivation, consider, for an activity $A(X) $, the norm
\begin{align*}
\Vert A\Vert = \sup_{x\in\mb L}\sum_{\substack{X\subset \mb L \\ X\ni x }} K^{\vert X\vert} e^{m d_{\text{t}}(X) }\vert A(X)\vert
\end{align*}
for some big constant $K$, some mass $m$, and $$ d_{\text{t}}(X) = \min_{T\in \mf T(X)} \sum_{\ell\in T}d(\ell)  $$ the spanning tree size of $X$. This is easily written in the form 
\begin{align*}
\sup_{\eta\in \mc B}\sum_ {X\subset \mb L }  \vert \eta(X)A(X)\vert
\end{align*}
required by our setup by introducing the operator $\Sigma: \mb C^{\mb L\times 2^{\mb L}}\to \mb C^{2^{\mb L}} $ through
\begin{align*}
\big(\Sigma(\check{\eta})\big)(X) &=  \sum_{x\in \mb L} \check{\eta}(x;X).
\end{align*}
Indeed, we then have
\begin{align*}
\Vert A\Vert = \sup_{\eta\in \Sigma (\check{\mc B}_{K,m})} \sum_ {X\subset \mb L  }  \vert \eta(X)A(X)\vert,
\end{align*}
where $\check{\mc B}_{K,m}\subset  \mb C^{\mb L\times2^{\mb L}}$ is the unit ball in the norm
\begin{align*}
\Vert\check{\eta}\Vert_{K,m} &=\sum_{x\in \mb L}   \sup_{X\subset \mb L}K^{-\vert X\vert} e^{-md_{\text{t}}(X) }e^{ \tilde m d(x,X)}\vert \check{\eta}(x;X)\vert
\end{align*}
with $\tilde m = \infty $ (other choices of $\tilde m $ generalize our class of norms). \\
With this in mind, we set $\check\aleph_\bullet^\bullet = \mb L\times \aleph_\bullet^\bullet $, and define $\Sigma_\bullet^\bullet :\mb C^{\check\aleph_\bullet^\bullet}\to\mb C^{\aleph_\bullet^\bullet}$ to be the sum over the first argument. Except for $\aleph_2 $, we further set $d_t(\alpha) = d_t(\supp\alpha) $, $d(x;\alpha) = \inf_{x'\in\supp\alpha }d(x,x') $ and $d(\alpha,\alpha') = \inf_{x\ob{'}\in \supp\alpha\ob{'}} d(x,x') $, with $\supp\alpha $ as before Property 1. For $\alpha=(x;\bm\xi;\bm\eta)\in\aleph_2$, in order to use a unified formalism, $d_t(\alpha) = \sum_{\xi'\in\bm\xi\circ\bm\eta}d(x;x') $ and $d(x';\alpha) = d(x,x') $. For $F\in\mf F(X) $, we also call $\mc T^c(F;X) $ any forest $F'\in\mf F(X) $ such that $F\dot\cup F'\in \mf T(X) $ and $F' $ has minimal length, as always denoted by $d(\mc T^c(F;X) ) $.

\begin{defin}\label{defnorms}
Let the components of $\bm\lambda \in \bm\Lambda=\mb R_{>0}^{11} $ be denoted by $$ \lambda_\bullet \;(\bullet= \phi,J,s,1,2,T,X,Q ),m,m_V,\tilde m,\mu . $$ Then, except for $\mc B_{1,\bm\lambda}$, the $\mc B_{\bullet,\bm\lambda}^\bullet  = \Sigma_{\bullet}^\bullet(\check{\mc B}_{\bullet,\bm\lambda}^\bullet) $, where $\check{\mc B}_{\bullet,\bm\lambda}^\bullet $ are unit balls of $\mb C^{\check\aleph_\bullet^\bullet}  $ in the norms 
\begin{align*}
\Vert \check{\eta}\Vert_{\bullet,\bm\lambda}^\bullet &= \sum_{x\in\mb L}\sup_{ \alpha\in \aleph_\bullet^\bullet  } \bm\lambda_\bullet^\bullet(\alpha) e^{-md_{\text{t}}(\alpha) }  e^{\tilde m d(x;\alpha)} \vert  \check{\eta}(x;\alpha)\vert
\end{align*}
Here, the weight factors are 
\begin{align*}
\bm\lambda_\bullet^\bullet(\alpha)&= \lambda_2^{n(\bm\xi)}\lambda_{J}^{n(\bm\zeta)}    && \alpha\in \aleph_2  \\
&= \lambda_X^{-\vert X\vert}\lambda_T^{-\vert X \vert+1}\lambda_\phi^{n(\bm\xi)}\lambda_J^{n(\bm\zeta)} \lambda_s^{2\vert F\vert}e^{2md(\mc T^c(F;X) )-m_Vd(F)} && \alpha\in\aleph_\phi\\
&= \lambda_X^{\vert X\vert}\lambda_T^{\vert X \vert-1} \lambda_J^{n(\bm\zeta)}  && \alpha\in\aleph_J\\
&= \lambda_X^{\vert X\vert}\lambda_T^{\vert X \vert-1}\lambda_Q^{\vert Q\vert} \lambda_J^{n(\bm\zeta)}  && \alpha\in\aleph_J'\\
&= \lambda_X^{\vert X\vert+\vert Z\vert}\lambda_T^{\vert X \vert+\vert Z\vert-1}\lambda_Q^{\vert Q\vert} \lambda_J^{n(\bm\zeta)}  && \alpha\in\tilde \aleph_J\\
&= \lambda_J^{n(\bm\zeta)} && \alpha\in\tilde \aleph
\end{align*}
For $\mc B_{1,\bm\lambda} $, $ \mc B_{1,\bm\lambda} = \cup_{r'\in[r,R]}\Sigma_1(\check{\mc B}_{1,\bm\lambda,r'}) $ with $ \check{\mc B}_{1,\bm\lambda,r'}   $ as above with the weight factor
\begin{align*}
\bm\lambda_1(\bm\xi;\bm\zeta;r') &=  \omega(r')\left(r'+\lambda_1^{-1}\right)^{-n(\bm\xi)}\lambda_{J}^{n(\bm\zeta)}\lambda_s^{2\vert \supp\bm\xi\circ\bm\zeta\vert}   
\end{align*}

\end{defin}
\erem

\begin{rem}
The mass $\mu  $ is used below. For simplicity, we will always assume $\tilde m,\mu\geq m,m_V>0 $ (in fact, $\tilde m =\mu= \infty $ is the standard choice). In application, we only need $\lambda_\bullet\leq 1 $. Note that with our choice of test functions,
\begin{align*}
\sup_{\eta\in \mc B_{\bullet,\bm\lambda}^\bullet} \sum_{\alpha\in\aleph_\bullet^\bullet} \vert\eta(\alpha)F(\alpha) \vert &= \sup_{\eta\in \mc B_{\bullet,\bm\lambda}^\bullet} \left\vert \sum_{\alpha\in\aleph_\bullet^\bullet} \eta(\alpha)F(\alpha)\right\vert.
\end{align*}
We will from now on sometimes use the RHS in the definition of our norms. \\
The weight $\bm\lambda_\phi(\alpha) $ is special as compared to the others. This is related to the way we obtain bounds on the interpolation, cf. also Remark \ref{rembint}. The $2$ in its $ e^{2md(\mc T^c(F;X) )} $ factor could be replaced by any other constant, upon changing the corresponding $ 2$ in Property 1 (i). 

\end{rem}

\erem
The verification of both Property 1 (i) and Property 2 with the above choice of $\mc B_{\phi,\bm\lambda} $ rests on Lemma \ref{lemvolumeargument} below (or some analogue of it), which induces an $m $ dependence of the corresponding maps $\Lambda^1 $ and $\Lambda^2 $ on $\bm\Lambda $. In the case of $\Lambda^2 $, this can be avoided if we use the following alternative definition of $\mc B_{\phi,\bm\lambda} $:

\begin{defin}\label{altdef} \textbf{(alternative definition of $\mc B_{\phi,\bm\lambda} $).}
Define alternatively $\check\aleph_\phi = \mb L\times \mf F(\mb L)\times \aleph_\phi $, and
\begin{align*}
\Sigma_\phi\check\eta(X;\bm\xi;F;\bm\zeta) &= \sum_{\substack{x\in\mb L  \\ F'\in\mf F(\mb L)  \\ F'\dot\cup F \in\mf T(X)}} \check\eta(x;F';X;\bm\xi;F;\bm\zeta).
\end{align*}
Then,
$\mc B_{\phi,\bm\lambda} = \Sigma_\phi(\check{\mc B}_{\phi,\bm\lambda}) $, where $\check{\mc B}_{\phi,\bm\lambda} $ is the unit ball of $\mb C^{\check\aleph_\phi}  $ in the norm 
\begin{align*}
\Vert \check{\eta}\Vert_{\phi,\bm\lambda} &=  \sum_{x\in\mb L  }  \sup_{\substack{ \alpha\in \aleph_\phi \\  F'\in\mf F(\mb L) \\ F'\dot\cup F \in\mf T(X)}}   \bm\lambda_\phi(\alpha;F') e^{-md_{\text{t}}(\alpha) }  e^{\tilde m d(x;\alpha)} \vert  \check{\eta}(x;F';\alpha)\vert.
\end{align*}
with the weight
\begin{align*}
\bm\lambda_\phi(\alpha;F')&=  \lambda_X^{-\vert X\vert}\lambda_T^{-\vert X \vert+1}\lambda_\phi^{n(\bm\xi)}\lambda_J^{n(\bm\zeta)} \lambda_s^{2\vert F\vert}e^{2md(F' )}\prod_{x\in X}d^F(x)!^{-1}d^{F'}(x)!^{-1}
\end{align*}

\end{defin}

\erem
In order to fully state Property 3, we also need to define the ``multilinear'' test function spaces $\mc B_J^\bullet(T) $ for $T\in\mf T(\ul k),k\geq 2 $. We set $\check\aleph_J^\bullet(k) = \mb L\times (\aleph_J^\bullet)^k $, and $\Sigma_J^\bullet $ is still the sum over the first component.

\begin{defin}
$\mc B^\bullet_{J,\bm\lambda}(T) = \Sigma_J^\bullet(\check{\mc B}^\bullet_{J,\bm\lambda}(T)) $, where $\check{\mc B}^\bullet_{J,\bm\lambda}(T) $ are the unit balls of $\mb C^{\check\aleph_J^\bullet(k)} $ in the norms
\begin{align*}
\Vert \check{\eta}\Vert_{\bullet,T,\bm\lambda}^\bullet &= \sum_{x\in\mb L}\sup_{\substack{\alpha_j\in \aleph_\bullet^\bullet \\j=1,\ldots,k}} \prod_{j=1}^k\left[\bm\lambda_J^\bullet(\alpha_j) e^{-md_{\text{t}}(\alpha_j) } \right] e^{\tilde m d(x;\circ_j\alpha_j) }\exp\left(\sum_{\ell\in T}\mu d(\alpha_\ell)\right)  \left\vert  \check{\eta}(x;\alpha_1,\ldots,\alpha_j)\right\vert.
\end{align*}
Here, $d(\alpha_{\{m,m'\}}) = d(\alpha_m,\alpha_{m'}) $

\erem
\end{defin}

\subsection{Verification of the Properties}

We now verify that our norms have the Properties 1 - 3 of section 3. The transformations $\Lambda^1,\Lambda^2,\Lambda_J^\bullet, $ and $\Lambda^3_{\Text{ii}} $ - $\Lambda^3_{\Text{v}}  $ that characterize these properties are by no means unique (already because no norm involves all parameters of $\bm\lambda\in\bm \Lambda $ at once), and there is sometimes no natural choice for them. Whenever there is no loss of generality, we make a simple arbitrary choice anyway.

\subsubsection{Property 1}

\begin{prop}
The Properties 1 hold with $\Lambda^1$ given by
\begin{align*}
\Lambda^1\bm\lambda &=  \left(\frac{\lambda_\phi}{2{\sf N}} ,\lambda_J,\lambda_s\frac{c_{\sf g}(m)^\half}{c_{\sf g}(m_V)^\half},\lambda_1,\lambda_2,\lambda_T, 48\lambda_X\left[ c_{\sf g}(\tilde m)  \vee   \lambda_T c_{\sf g}(m) \right],\lambda_Q,m,m_V,\tilde m,\mu\right) && \text{or}\\
&=\left(\frac{\lambda_\phi}{2{\sf N}} ,\lambda_J,\lambda_s\frac{c_{\sf g}(m)^\half}{c_{\sf g}(m_V)^\half},\lambda_1,\lambda_2,\lambda_T, 128 \lambda_X\left[ c_{\sf g}(\tilde m)  \vee   \lambda_T c_{\sf g}(m) \right] ,\lambda_Q,m,m_V,\tilde m,\mu\right)   &&
\end{align*}
for the two alternative choices of $\mc B_{\phi,\bm\lambda} $. 

\end{prop}

\begin{proof}
Properties (ii) and (iii) are easy to check using $$ \check\eta(x;\bm\xi;\bm\zeta) = \frac{\delta_{x\in\supp\bm\xi\circ\bm\zeta}}{\vert \supp\bm\xi\circ\bm\zeta\vert }\eta(\bm\xi;\bm\zeta)  $$ for (ii) and $$ \check\eta(x';x;\bm\xi;\bm\zeta) = \delta_{x,x'}\eta(x;\bm\xi;\bm\zeta) $$ for (iii). We turn to Property (i) for the first alternative definition of $\mc B_{\phi,\bm\lambda} $. Set $\dot{\bm\lambda} = \Lambda^1\bm\lambda $ as usual, and abbreviate $c=c_{\sf g}(\tilde m)c_{\sf g}(m)^{-1}\lambda_T^{-1} $. Let $\eta\in\mc B_{\phi,\bm\lambda} $. Then, by definition, our choice of $\Lambda^1 $, and $d_t(X)\leq d(F)+d(\mc T^c(F;X)) $ for any $F\in \mf F(X) $, we have
\begin{align*}
\sum_{X\subset \mb L} \sum_{\substack{\bm\xi,\bm\zeta\in \bd L\vert_X  \\  F\in \mf F(X)}}  \dot\lambda_X^{-\vert X\vert} & \dot\lambda_\phi^{n(\bm\xi)} \dot\lambda_s^{2\vert F\vert} \dot\lambda_J^{n(\bm\zeta)} e^{- (m+2m_V)d(F)}\vert \eta(X;\bm\xi;F;\bm\zeta)\vert \\ & \leq \sup_{x\in\mb L} c_{\sf g}(\tilde m)^{-1}c\sum_{X\subset \mb L} [48(c\vee 1)]^{-\vert X\vert}\sum_{\bm\xi,\bm\zeta\in \bd L\vert_X }(2{\sf N})^{-n(\bm\xi\circ\bm\zeta)} \\&\qquad\times\sum_{  F\subset T\in \mf T(X)  } c_{\sf g}(m_V)^{-\vert   F\vert} c_{\sf g}(m)^{-\vert T\setminus F\vert}e^{-m_Vd(F)-md(T\setminus F) -\tilde md(x;X)}\\
&\leq   1
\end{align*}
by the Lemma below. For the second alternative, we have in the same way
\begin{align*}
\sum_{X\subset \mb L} \sum_{\substack{\bm\xi,\bm\zeta\in \bd L\vert_X  \\  F\in \mf F(X)}}  &\dot\lambda_X^{-\vert X\vert}  \dot\lambda_\phi^{n(\bm\xi)} \dot\lambda_s^{2\vert F\vert} \dot\lambda_J^{n(\bm\zeta)} e^{- (m+2m_V)  d(F)}\vert \eta(X;\bm\xi;F;\bm\zeta)\vert \\ & \leq\sup_{x\in\mb L} c_{\sf g}(\tilde m)^{-1}c\sum_{X\subset \mb L}  [128(c\vee 1)]^{-\vert X\vert}\sum_{\bm\xi,\bm\zeta\in \bd L\vert_X }(2{\sf N})^{-n(\bm\xi\circ\bm\zeta)} \\&\qquad\times\sum_{  F\subset T\in \mf T(X)  } c_{\sf g}(m_V)^{-\vert   F\vert} c_{\sf g}(m)^{-\vert T\setminus F\vert}e^{-m_Vd(F)-md(T\setminus F) -\tilde md(x;X)}\prod_{x\in X}d^T(x)!
\end{align*}
The claim follows from the second statement in the Lemma.

\end{proof}

\begin{lem}\label{lemvolumeargument}
We have
\begin{align*}
\sup_{x\in\mb L}\sum_{\substack{ X\subset\mb L \\ F\subset T\in\mf T(X) }} 12^{-\vert X\vert} c_{\sf g}(m_V)^{-\vert F\vert}c_{\sf g}(m)^{-\vert T\setminus F\vert} e^{-m_Vd(F)-md(T\setminus F) }e^{-\tilde m d(x;X)} \leq c_{\sf g}(\tilde m)
\end{align*}
and 
\begin{align*}
\sup_{x\in\mb L}\sum_{\substack{ X\subset\mb L \\ T\in\mf T(X) }} 32^{-\vert X\vert} c_{\sf g}(m_V)^{-\vert F\vert}c_{\sf g}(m)^{-\vert T\setminus F\vert} e^{-m_Vd(F)-md(T\setminus F) }e^{-\tilde m d(x;X)}\prod_{x\in X}d^T(x)! \leq c_{\sf g}(\tilde m)
\end{align*}
\end{lem}

\begin{proof}
We have for every $x\in \mb L $
\begin{align*}
\sum_{\substack{ X\subset\mb L \\ F\subset T\in\mf T(X) }} 6^{-\vert X\vert}& c_{\sf g}(m_V)^{-\vert F\vert}c_{\sf g}(m)^{-\vert T\setminus F\vert} e^{-m_Vd(F)-md(T\setminus F) }e^{-\tilde m d(x;X)} \\&\leq \sum_{q\geq 1}\frac{6^{-q}  }{q!} \sum_{\substack{   F\subset T\in\mf T(\ul q) \\ p_0\in\ul q}} \sum_{x_1,\ldots,x_q\in X } e^{-\tilde md(x_{p_0},x)} \\&\qquad\qquad\times\prod_{\{p,p'\}\in F} c_{\sf g}(m_V)^{-1}e^{-m_Vd(x_p,x_{p'})}\prod_{\{p,p'\}\in T\setminus F} c_{\sf g}(m)^{-1}e^{-md(x_p,x_{p'})}
\end{align*}
By a pinch and sum argument, working inwards from the leafs of $T$ to $p_0 $, we have
\begin{align*}
\sum_{x_1,\ldots,x_q\in X } e^{-\tilde md(x_{p_0},x)}  \prod_{\{p,p'\}\in F} c_{\sf g}(m_V)^{-1}e^{-m_Vd(x_p,x_{p'})}\prod_{\{p,p'\}\in T\setminus F} c_{\sf g}(m)^{-1}e^{-md(x_p,x_{p'})} \leq c_{\sf g}(\tilde m).
\end{align*}
By Cayleys formula, $\vert \mf T(\ul q)\vert = q^{q-2} \leq (q-1)! 3^q $, and in fact
\begin{align*}
\sum_{T\in\mf T(\ul q)} \prod_{p=1}^q d^T(q)! \leq (q-1)!8^q.
\end{align*}
Both parts of the lemma follow.

\end{proof}

\subsubsection{Property 2}

\begin{prop}\label{propprop2}
(i) With the first alternative definition of $\mc B_{\phi,\bm\lambda} $, suppose that $\lambda_Q\geq e^{-\frac\upmu8r^2} $. 
Let $\Lambda^2 $ be any map of $\bm\Lambda $ such that, with $\dot{\bm\lambda} = \Lambda^2\bm\lambda $, we have
\begin{align*}
 \dot\lambda_X \dot\lambda_T^\half  &\geq \frac{96{\sf N}^2c_{\Text{L.\ref{lemhadamard}}} }{\lambda_X \lambda_T^\half e^{\frac\upmu{16}r^2} }\\
\dot\lambda_s&\leq c_{\sf g}(m_V)^{-\half} \cdot \big[1\wedge c_{\Text{L.\ref{lemhadamard}}}^{-1}\big]\cdot \frac{\lambda_X \lambda_T^\half  \dot\lambda_T^\half \dot\lambda_X  }{96{\sf N}^2 } \\ 
\dot\lambda_\phi&\leq  \Vert C\Vert_{3m,\infty}^{-\half}c_{\sf g}(m_V)^{-\half} \cdot \big[1\wedge c_{\Text{L.\ref{lemhadamard}}}^{-1}\big]\cdot \frac{\lambda_X \lambda_T^\half  \dot\lambda_T^\half \dot\lambda_X }{96{\sf N}^2 } 
\end{align*}
(The first of these bounds will not be saturated in application). Then, Property 2 holds for this choice of $\Lambda^2 $ if 
\begin{align*} 
c_{\Text{P.\ref{propprop2}}}  \leq \frac{\left(\dot\lambda_T\lambda_T\right)^\half}{3{\sf N}^2c_{\sf g}(m)^\half\Vert C\Vert_{3m,\infty}^{\half}}
\end{align*}
with $$ c_{\Text{P.\ref{propprop2}}}  = \sqrt{\frac{64{\sf N}}{\upmu\lambda_{\Text{min}}(\Re C) } }r^{-1} .$$ 
(ii) With the second alternative definition of $\mc B_{\phi,\bm\lambda} $, the same is true with every $c_{\sf g}(m),c_{\sf g}(m_V) $ above replaced by $1$. In particular, the constraint on $c_{\Text{P.\ref{propprop2}}} $ is 
\begin{align}\label{stabassumr}
c_{\Text{P.\ref{propprop2}}}  \leq \frac{\left(\dot\lambda_T\lambda_T\right)^\half}{3{\sf N}^2 \Vert C\Vert_{3m,\infty}^{\half}}
\end{align}

\end{prop}

\begin{proof}
We prove only the statement for
\begin{align*}
\gamma'\eta(X;\bm\xi;F;\bm\zeta) &= \sum_{\substack{ T\in\mf T(X) \\ \emptyset\neq Q\subset X }} \gamma'(T;X,Q;\bm\xi;F)\eta(X,Q;\bm\zeta).
\end{align*}
for the first alternative definition of $\mc B_{\phi,\bm\lambda} $. The other statements are simpler and largely identical. Note that $\gamma'\eta = \Sigma_{\phi}\gamma'\check\eta $, so we have to bound $\Vert \gamma'\check\eta\Vert_{\phi,\dot{\bm\lambda}} $. We have
\begin{align*}
\Vert \gamma'\check\eta\Vert_{\phi,\dot{\bm\lambda}} &\leq   \sup_{\substack{X\subset \mb L \\ \bm\xi\in\bd L\vert_X,F\in\mf F(X)}}  \sum_{\substack{ T\in\mf T(X) \\ \emptyset\neq Q\subset X }} \left( \dot\lambda_X \lambda_X \right)^{-\vert X\vert}\left(\dot\lambda_T\lambda_T\right)^{-\vert X\vert+1}\lambda_Q^{-\vert Q\vert} \\&\qqquad\times\dot\lambda_\phi^{n(\bm\xi)} \dot\lambda_s^{2\vert F\vert}e^{2m d(\mc T^c(F;X))-m_Vd(F)} \vert \gamma'(T;X,Q;\bm\xi;F)\vert 
\end{align*}
Quickly recall the definition of $\gamma' $. For $\phi\in\Xi^{X} $, $\psi\in \Xi_{\mb C}^{X} $ and $s\in [0,1]^T $, we have, by a simple computation using that $\Re C^{-1}_{s^T} $ is nonnegative (see the Lemma below),
\begin{align*}
- \Re \big\langle \phi+\psi ,C^{-1}_{s^T}(\phi+\psi)  \big\rangle + \frac{\upmu}8 \Vert\phi \Vert_2^2 \leq -\half \big\langle\phi,\left[\Re C^{-1}_{s^T}-\frac{\upmu}{4}\right] \phi \big\rangle + 4\lambda_{\Text{min}}(\Re C_{s^T})^{-1}\Vert \psi\Vert_2^2.
\end{align*}
Again by the Lemma, for any $Q,S\subset X $
\begin{align*}
\vert \det 2\pi C_{s^T}\vert^{-\half} \int_{\Xi^{X}} \prod_{x\in X}\ud\phi(x) \chi_{Q,S}(\phi) \exp \left[-\half \big\langle\phi,\left[\Re C^{-1}_{s^T}-\frac{\upmu}{4}\right] \phi \big\rangle \right] \leq c_{\Text{L.\ref{lemhadamard}}}^{\vert X\vert}   e^{-\vert Q\cup S\vert\frac\upmu4r^2 } 
\end{align*}
Performing some elementary bounds and using Lemma \ref{lemvolumeargument}, we conclude
\begin{align*}
\Vert \gamma'\check\eta\Vert_{\phi,\dot{\bm\lambda}} &\leq \sup_{\substack{X\subset\mb L\\ F\in\mf F(X)}}{\sf N}^{2\vert X\vert-2} \left( \dot\lambda_X \lambda_X \right)^{-\vert X\vert}\left(\dot\lambda_T\lambda_T\right)^{-\vert X\vert+1}  c_{\Text{L.\ref{lemhadamard}}}^{\vert X\vert}\\&\qquad\times \sum_{\substack{F\subset T\in \mf T(X) \\ \emptyset\neq Q\subset X }} e^{-\vert Q\vert\frac\upmu{8}r^2 }\lambda_Q^{-Q}e^{-m_Vd(F)-md(T\setminus F)} \prod_{x\in X}d^T(x)!\\&\qquad\times\dot\lambda_s^{2\vert F\vert} \Vert C\Vert_{3m,\infty}^{\vert T\setminus F\vert}\sum_{\substack{ \mc L_1\dot\cup\cdots\dot\cup\mc L_4 = \mc L(T\setminus F)  \\ \mc L_4\vert_2 = \{\mc L_2\vert_2\circ\mc L_3\vert_2\} }} c_{\Text{P.\ref{propprop2}}}^{\vert \mc L_3\vert}\dot\lambda_\phi^{\vert \mc L_1\vert + \vert\mc L_2\vert}    e^{ -\vert  \mc L_4 \vert\frac\upmu{16}r^2  } \\
&\leq     \sup_{X\subset\mb L}  \left( {\sf N}^{2} \left(\dot\lambda_T\lambda_T\right)^{-\half} \left[c_{\sf g}(m_V)^\half\dot\lambda_s +c_{\sf g}(m)^\half\Vert C\Vert_{3m,\infty}^{\half}\left[ \dot\lambda_\phi +  c_{\Text{P.\ref{propprop2}}} \right] \right] \right)^{ \vert X\vert -2}\\&\qquad\times \left( \frac{ 32{\sf N}^{2} c_{\Text{L.\ref{lemhadamard}}}  }{\dot\lambda_X \lambda_X  \left(\dot\lambda_T\lambda_T\right)^{\half}} \left[c_{\sf g}(m_V)^\half\dot\lambda_s +c_{\sf g}(m)^\half\Vert C\Vert_{3m,\infty}^{\half} \left[\dot\lambda_\phi +  e^{ - \frac\upmu{16}r^2  } \right] \right] \right)^{\vert X\vert}\\& \leq 1.
\end{align*}
The Proposition follows.

\end{proof}

\begin{lem}\label{lemhadamard}
Let $X\subset \mb L $, $T\in\mf T(X) $ and $s\in[0,1]^T $ be fixed. Let $\sigma(A) $ denote the spectrum of a symmetric real matrix $A$ on $X\times\N$. Set as before $\upmu = \lambda_{\Text{min}}(\Re C^{-1}) $ and $\mf a = \lambda_{\Text{min}}(\Re C) $. \\
(i) $\sigma(\Re C_{s^T}^{-1})\subset [\upmu,\mf a^{-1}] $ and $\sigma(\Re C_{s^T})\subset [\mf a,\upmu^{-1}] $.\\
(ii) For any $Q,S\subset X $,
\begin{align*}
\vert \det 2\pi C_{s^T}\vert^{-\half} \int_{\Xi^{X}} \prod_{x\in X}\ud\phi(x) \chi_{Q,S}(\phi) \exp \left[-\half \big\langle\phi,\left[\Re C^{-1}_{s^T}-\frac{\upmu}{4}\right] \phi \big\rangle \right] \leq c_{\Text{L.\ref{lemhadamard}}}^{\vert X\vert}   e^{-\vert Q\cup S\vert\frac\upmu4r^2 } 
\end{align*}
\end{lem}

\begin{proof}
(i) $C_{s^T} $ is the Hadamard product $\bm 1_{s^T}\circ C\vert_{X}$, where $\bm 1 $ is the operator on $\mb R^{X\times\N} $ with constant kernel equal to $1$. $\bm 1_{s^T}$ is positive semidefinite and has $1$'s on the diagonal. Let $\lambda_k $ be the eigenvalues of $C$ and $\mc F_{k\xi} $ the unitary matrix that diagonalizes $C$ and $\rho_p^{s^T} $ be the eigenvalues of $\Re C_{s^T} $ and $\mc F_{k\xi}^{s^T} $ the unitary matrix that diagonalizes $\Re C_{s^T}$. Then,
\begin{align*}
\rho_p^{s^T} &= \sum_{k}\mc R^{s^T}_{pk}\Re\lambda_k\\
\mc R^{s^T}_{pk} &= \sum_{\xi,\xi'\in X\times\N} \mc F^{s^T}_{p\xi}\ol{\mc F_{k\xi}} \bm 1_{s^T}(\xi,\xi') \ol{\mc F^{s^T}_{p\xi'}}\mc F_{k\xi'}\geq 0
\end{align*}
We have $\sum_{k}\mc R^{\bd s(T)}_{pk}=1$ and so the eigenvalues of $\Re C_{\bd s(T)} $ are convex combinations of the eigenvalues of $\Re C $. The same statement is true for $\Im C_{s^T} $ and for $C_{s^T} $ itself (but note that $C_{s^T} $ might not be normal). In particular, $\Re C_{s^T} $ is positive definite with $\sigma(\Re C_{s^T})\subset [\mf a,\upmu^{-1}] $. Note that $$ \Re C_{s^T}^{-1} = \Big(\Re C_{s^T} + \Im C_{s^T} (\Re C_{s^T})^{-1}\Im C_{s^T}\Big)^{-1}. $$ It follows that $$ \lambda_{\text{max}}(\Re C_{s^T}^{-1}) \leq \mf a^{-1} . $$ The same is true for $\Im C_{s^T}^{-1} $. Further, by the Schur determinant formula, we have
\begin{align*}
\det \Big[\lambda - (\Re C_{s^T} + \Im C_{s^T}&(\Re C_{s^T})^{-1}\Im C_{s^T}) \Big] \\& = \det \left(\begin{array}{cc} \Re C &\Im C \\ \Im C & \lambda- \Re C \end{array}\right)\circ \left(\begin{array}{cc} \bm 1_{s^T} & \bm 1_{s^T} \\ \bm 1_{s^T} & \bm 1_{s^T} \end{array}\right)  \cdot \det (\Re C_{s^T})^{-1}
\end{align*}
We have for the eigenvalues 
\begin{align*}
\text{spec}  \left(\begin{array}{cc} \Re C &\Im C \\ \Im C & \lambda- \Re C \end{array}\right) = \left\{\half \lambda\pm \half \left[4\frac{1-\lambda \cdot \Re a }{\vert a\vert^2}+\lambda^2 \right]^\half\,,\; a\in \text{spec }C^{-1} \right\}
\end{align*}
which is $\subset \mb R_{>0} $ whenever $\lambda >\frac 1{\upmu} $, and so it follows from Oppenheims inequality that the above characteristic polynomial is nonzero for $\lambda >\frac 1{\upmu}$, i.e. $\lambda_{\text{min}}(\Re C_{s^T}^{-1} )\geq \upmu $.\\
(ii) Call $\mf I=\mf I(X,Q,S,T,s) $ the quantity we have to bound. For $A\in\text{End }\mb C^{X\times\N}$, call $I(A) = I(A;X,Q,S) $ the integral
\begin{align*}
I(A) &= \int_{\Xi^X}\prod_{x\in X}\ud\phi(x) \chi_{Q,S}(\phi) e^{-\half \langle \phi,A\phi\rangle}.
\end{align*}
Then
\begin{align*}
\mf I = \vert \det 2\pi C_{s^T}\vert^{-\half} I(\Re C^{-1}_{s^T}-\frac{\upmu}{4}).
\end{align*}
By (i), we have
\begin{align*}
I\Big(\Re C^{-1}_{s^T}-\frac{\upmu}{4}\Big) \leq  e^{-\vert Q\cup S\vert\frac\upmu4r^2 } I\Big(\half \Re C^{-1}_{s^T}\Big).
\end{align*}
Using the identity
\begin{align*}
\langle \phi ,A \phi\rangle &= \left\langle \phi\vert_P, \left(A_{P,P} - A_{P,X\setminus P} A_{X\setminus P,X\setminus P}^{-1}A_{X\setminus P,P}\right)\phi\vert_P\right\rangle\\
& \qquad+ \left\langle  \phi\vert_{X\setminus P} + A_{X\setminus P,X\setminus P}^{-1} A_{X\setminus P,P}\phi\vert_P, A_{X\setminus P,X\setminus P}\left[\phi\vert_{X\setminus P} + A_{X\setminus P,X\setminus P}^{-1}A_{X\setminus P,P}\phi\vert_P\right]  \right\rangle
\end{align*}
with $P=Q\cup S $ and $A_{P,P'}(\xi,\xi') =\delta_{x\in P}\delta_{x'\in P'}A(\xi,\xi') $, and noting the eigenvalue inequalities for the Schur complement $A_{P,P} - A_{P,X\setminus P} A_{X\setminus P,X\setminus P}^{-1}A_{X\setminus P,P} $ analogous to (i), one shows easily
\begin{align*}
I\Big(\half \Re C^{-1}_{s^T}\Big)\leq \left(2\omega_{\sf N}r^{\sf N}e^{-\frac\upmu2r^2}\right)^{\vert P\vert} \det \Big(4\pi\Re C^{-1}_{s^T,X\setminus P,X\setminus P} \Big)^{-\half}
\end{align*}
where $\omega_{\sf N} $ is the volume of the $\sf N $ sphere. We claim that
\begin{align}\label{cinfty}
\det \Big(\Re C^{-1}_{s^T,X\setminus P,X\setminus P} \Big)^{-\half} &\leq  c_{\infty}^{\half \vert X\setminus P\vert }
\end{align}
with $c_{\infty} = \sup_{\xi\in \mb L\times\N} \Re C(\xi,\xi) + \big[\Im C(\Re C)^{-1}\Im C\big](\xi,\xi).  $ Indeed, let $\mc G_{p\xi} $ be the matrix diagonalizing $\Re C^{-1}_{s^T,X\setminus P,X\setminus P}  $, and $\rho_p $ be the corresponding eigenvalues. Repeating the argument from (i) and using the concavity of the logarithm, we get
\begin{align*}
\det \Big(\Re C^{-1}_{s^T,X\setminus P,X\setminus P} \Big)^{-\half} &\leq \exp\left[-\half \sum_{p}\sum_{\xi\in X\setminus P\times\N}\vert \mc G_{p\xi}\vert^2 \log \rho_p \right] = \det   \Big(\Re C^{-1}_{s^T}\Big)_{X\setminus P,X\setminus P} ^{-\half}\\
&\leq \left[\prod_{\xi\in X\setminus P\times\N}  \Re C_{s^T}(\xi,\xi) + \big[\Im C_{s^T}(\Re C_{s^T})^{-1}\Im C_{s^T}\big](\xi,\xi)\right]^\half
\end{align*}
by the Hadamard inequality. However, 
\begin{align*}
\Im C_{s^T} (\Re C_{s^T})^{-1}\Im C_{s^T}(\xi,\xi) &\leq  \Im C (\Re C)^{-1}\Im C(\xi,\xi)
\end{align*}
Indeed, the matrix 
\begin{align*}
\left(\begin{array}{cc} \Re C_{s^T} & \Im C_{s^T} \\ \Im C_{s^T} & [\Im C  (\Re C)^{-1}\Im C]_{s^T} \end{array}\right) = \left(\begin{array}{cc} \Re C & \Im C \\ \Im C & \Im C (\Re C)^{-1}\Im C \end{array}\right)\circ \left(\begin{array}{cc} \bd 1_{s^T} & \bd 1_{s^T} \\ \bd 1_{s^T} & \bd 1_{s^T} \end{array}\right)
\end{align*}
is positive semidefinite by the Schur product theorem, since both of its Hadamard factors are (for the first, we use that $\Re C $ is positive definite and its Schur complement in the block matrix is $0$). The Schur complement of the above, namely $[\Im C  (\Re C)^{-1}\Im C]_{s^T} - \Im C_{s^T}(\Re  C_{s^T})^{-1} \Im C_{s^T} $, is therefore positive definite, hence has nonnegative diagonal entries, which are equal to $\Im C  (\Re C)^{-1}\Im C(\xi,\xi)- \Im C_{s^T}(\Re  C_{s^T})^{-1} \Im C_{s^T}(\xi,\xi)$, as stated.\\
Putting together the results obtained so far, we get the bound
\begin{align*}
\mf I \leq \left(2^{\half\sf N}c_\infty^\half +2 r^{\sf N}e^{-\frac\upmu2r^2} \right)^{\vert X\vert}e^{-\vert Q\cup S\vert\frac\upmu4r^2 }\cdot \vert \det   C_{s^T}\vert^{-\half} .
\end{align*}
Arguing like before for the determinant, we see
\begin{align*}
\vert \det   C_{s^T}\vert^{-\half} &\leq \exp\left[-\half \sum_k \sum_{\xi\in X\times\N} \vert \mc F_{k\xi}\vert^2\log\Re\lambda_k\right],
\end{align*}
where, as in (i), $\mc F_{k\xi} $ diagonalizes $C$ and $\lambda_k $ are the eigenvalues of $C$. Abbreviating
\begin{align}\label{cd}
c_d &= \sup_{X\subset \mb L} \exp\left[-\frac1{2\vert X\vert} \sum_k \sum_{\xi\in X\times\N} \vert \mc F_{k\xi}\vert^2\log\Re\lambda_k\right],
\end{align}
(obviously, $c_d\leq \mf a^{-\frac{\sf N}2} $, but see Remark \ref{remdestroycancel} just below) the claim follows with $$c_{\Text{L.\ref{lemhadamard}}} = c_d\left(2^{\half\sf N}c_\infty^\half +2 r^{\sf N}e^{-\frac\upmu2r^2} \right). $$

\end{proof}

\begin{rem}\label{remdestroycancel}
An inequality of the type $$ I\Big(\half \Re C^{-1}_{s^T}\Big)\leq \left(\const\upmu^{-\frac{\sf N}2}\right)^{\vert X\vert}  $$ would easily follow from (i), but we want to get a better $\upmu $ - dependence in the case when most of the eigenvalues of $C^{-1} $ are larger than $\upmu$. In the same way, the bound $\vert \det   C_{s^T}\vert^{-\half} \leq\mf a^{-\frac{\sf N}2 \vert X\vert} $ follows easily from the one derived above, but also generally has a worse $\upmu $ dependence.\\ 
For example, let ${\sf N}=2 $, $\mb L= \mb Z^{D+1}/L\mb Z^{D+1} $ the discrete torus, and set $\mb L^* = \frac{2\pi}L\mb Z^{D+1}/2\pi \mb Z^{D+1} $. Then, the many Boson propagator defined in (\ref{mbbfktp}) with $\theta=1 $ is 
\begin{align}\label{propmbfourier}
C\big((x,{\sf n}),(x',{\sf n'})\big) &= L^{-D-1}\sum_{k=(k_0,\bd k)\in\mb L^*}   \frac{ e^{ik(x-x')}}{2\cosh (\hat{\text{h}}(\bd k)-\mu) -2\cos k_0}\\&\qqquad\times \left(\begin{array}{cc} e^{\hat{\text{h}}(\bd k)-\mu}-\cos k_0 & -\sin k_0\\ \sin k_0&  e^{\hat{\text{h}}(\bd k)-\mu}-\cos k_0\end{array}\right)_{\sf n,n'}
\end{align}
where $ \hat{\text{h}}(\bd k)$ is the spatial Fourier transform of $\text{h}$. For $\text{h} = -\Delta $ (or some reasonable approximation thereof), $ \hat{\text{h}}(\bd k) $ is positive with a quadratic zero at $\bd k=0 $. Therefore, the singularity (\ref{propmbfourier}) at $\upmu=k=0 $ is integrable for any $D\geq 1$, and we have $c_\infty\leq \const $ independent of $\mu\leq 0 $. Also, the matrix $\mc F_{k\xi} $ diagonalizing $C$ has $ \vert \mc F_{k\xi}\vert^2\sim L^{-D-1}$ (discrete Fourier transform tensorized with a $2\times 2$ matrix), and we deduce in a similar way $c_d\leq \const, $ independent of $\mu\leq 0$. Finally, an appropriate choice for $r $ will be $r\sim \mf v^{-\frac14+\epsilon} $ for $\mf v $ the (small) coupling constant (cf. below (\ref{notnorms})). As long as $\upmu\geq \mf v^{\half -3\epsilon} $, we conclude that $c_{\Text{L.\ref{lemhadamard}}} $ is bounded uniformly.

\erem
\end{rem}

\subsubsection{Property 3}

\begin{prop}\label{propprop3}
The Properties 3 hold with $\Lambda^3_{\Text{iii}} \bm\lambda = \bm\lambda  $ and 
\begin{align*}
\Lambda_J^\bullet \bm\lambda &= \left(\lambda_\phi,\lambda_J,\lambda_s,\lambda_1,\lambda_2,\lambda_T,\frac{\lambda_X}{3},\lambda_Q,m,m_V,\tilde m,\mu\right) \\
\Lambda^3_{\Text{ii}} \bm\lambda &= \left(\lambda_\phi,\lambda_J,\lambda_s,\lambda_1,\lambda_2,\lambda_T,\frac{\lambda_X}{c^3_{\Text{ii}}},\lambda_Q,m,m_V,\tilde m,\mu\right) \\
\Lambda^3_{\Text{iv}}  \bm\lambda &= \left(\vphantom\half\lambda_\phi,\lambda_J,\lambda_s,\lambda_1,\lambda_2,\lambda_T\wedge 1,\lambda_X\wedge 1,\lambda_Q\wedge e^{-\frac\upmu8r^2},2m,m_V,\tilde m,\mu\right) \\
\Lambda^3_{\Text{v}}  \bm\lambda &= \left(\lambda_\phi,\lambda_J,\lambda_s,\lambda_1,\lambda_2,\lambda_T,\frac{\lambda_X}{c_{\sf g}(\tilde m)},\lambda_Q,m,m_V,\tilde m,\mu\right) 
\end{align*}
and $$c_J^\bullet(\bm\lambda) = 16\; \sup_{x\in\mb L}\sum_{y\in \mb L} e^{-\mu d(x,y)}.  $$

\end{prop}

\begin{proof}
The statements (ii) to (v) are easy verifications, using for (iv) $d_t(X\cup Y)\geq \half d_t(X) $ for any two sets $X,Y\subset\mb L$\footnote{The choice $\lambda_Q\wedge e^{-\frac\upmu8r^2} $ in $\Lambda^3_{\Text{iv}} $ is, of course, arbitrary. It is inspired by the domain of $ \Lambda^2$, cf. Proposition \ref{propprop2}, and also necessary to get the factor $e^{-\frac\upmu8r^2} $ in Theorem \ref{thmabound} (ii).}. We prove (i). Assume w.l.o.g. that $1$ is a leaf of $T$ and $\{1,2\}\in T $. Let $\eta = \Sigma_J^\bullet(\check\eta)\in \mc B_{J,\bm\lambda}^\bullet(T) $ and denote, with $\alpha' = \circ_{m\geq 2}\alpha_m $,
\begin{align*}
\eta_1(\alpha_1,\ldots,\alpha_n) &= \sum_{x\in \mb L}\delta\Big(\tilde md(x,\alpha_1)\leq\tilde md(x,\alpha')\Big) \check\eta(x;\alpha_1,\ldots,\alpha_n) = \Sigma_{J}^\bullet(\check\eta_1)(\alpha_1,\ldots,\alpha_n)\\
\eta_2(\alpha_1,\ldots,\alpha_n) &= \sum_{x\in \mb L}\delta\Big(\tilde md(x,\alpha_1)>\tilde md(x,\alpha')\Big) \check\eta(x;\alpha_1,\ldots,\alpha_n) = \Sigma_{J}^\bullet(\check\eta_2)(\alpha_1,\ldots,\alpha_n)
\end{align*}
Obviously, $\eta=\eta_1+\eta_2 $ and $\eta_1,\eta_2\in\mc B_{J,\bm\lambda}^\bullet(T)  $. Define
\begin{align*}
\check\theta_1(x;\alpha_1) &=  \vert \alpha_1\vert^{-1}\sum_{\alpha_2,\ldots,\alpha_n\in\aleph_J^\bullet}\check\eta_1(x;\alpha_1 ,\ldots,\alpha_n)\prod_{m=2}^n F_m(\alpha_m)\\
\check\vartheta_2(x;\alpha_2,\ldots,\alpha_n) &=  \vert \alpha_2\vert^{-1}\sum_{\alpha_1\in\aleph_J^\bullet}\check\eta_2(x;\alpha_1,\alpha_2,\ldots,\alpha_n)  F_1(\alpha_1)
\end{align*}
Just for this proof, we abbreviated $\vert \alpha\vert = \vert \supp\alpha\vert $. We will show that
\begin{align*}
\Vert \check\theta_1\Vert_{J, \bm\lambda}^\bullet &\leq \frac{c_J^\bullet(\bm\lambda)}{16} \sup_{\vartheta_1\in \mc B_{J,\bm\lambda}^\bullet(T') } \left\vert \sum_{\alpha_2,\ldots,\alpha_n\in\aleph_J^\bullet} \vartheta_1(\alpha_2,\ldots,\alpha_n)\prod_{m=2}^n F_m(\alpha_m) \right\vert\\
\Vert \check\vartheta_2\Vert_{J,T',\bm\lambda}^\bullet &\leq\frac{c_J^\bullet(\bm\lambda)}{16} \sup_{\theta_2\in \mc B_{J, \bm\lambda}^\bullet}\left\vert \sum_{\alpha_1\in\aleph_J^\bullet} \theta_2(\alpha_1)F_1(\alpha_1) \right\vert
\end{align*}
with $T' = T\setminus\{\{1,2\}\} $. It follows from this by working inwards from the leaves of the tree that
\begin{align*}
\left\vert \sum_{\alpha_1,\ldots,\alpha_n} \eta(\alpha_1,\ldots,\alpha_n)\prod_{m=1}^n F_m(\alpha_m)\right\vert& \leq \left(\frac{c_J^\bullet(\bm\lambda)}{16}\right)^n 2^{\vert T\vert}\prod_{m=1}^n \left\Vert \vert\alpha\vert^{d^T(m)}F_m\right\Vert_{\bm\lambda}\\
&\leq \left(\frac{c_J^\bullet(\bm\lambda)}8\right)^n \prod_{m=1}^n d^T(m)!\left\Vert F_m\right\Vert_{\bm\lambda},
\end{align*}
which would be the claim. The last step follows from $\vert\alpha\vert^d \dot{\bm\lambda}_J^\bullet(\alpha)\leq d!  \bm\lambda_J^\bullet (\alpha)$, as is easily verified. \\
For the bound on $\check\theta_1 $, define 
\begin{align*}
\omega_1(y;\alpha_1,\alpha') &= \exp\left(-\tilde md(y;\alpha')-\mu d(y;\alpha_1)   \right)\\
\omega_1(\alpha_1;\alpha' ) &= \sum_{y\in\mb L}\omega_1(y;\alpha_1,\alpha' ) .
\end{align*}
Trivially,
\begin{align*}
\omega_1(\alpha_1,\alpha') \geq \exp\left(-\mu d(\alpha_1;\alpha')   \right).
\end{align*}
Therefore,
\begin{align*}
\Vert \check\theta\Vert_{J, \bm\lambda}^\bullet &= \sum_{x\in\mb L}\sup_{ \alpha_1\in \aleph_J^\bullet   }  \bm\lambda_J^\bullet(\alpha_1 ) e^{- md_{\text{t}}(\alpha_1) }  e^{\tilde m d(x;\alpha_1)} \vert  \check{\theta}(x;\alpha_1)\vert\\
&= \sum_{x\in\mb L}\sup_{ \alpha_1\in \aleph_J^\bullet   }  \vert \alpha_1\vert^{-1} \bm\lambda_J^\bullet(\alpha_1 ) e^{- md_{\text{t}}(\alpha_1) }  e^{\tilde m d(x;\alpha_1)} \left\vert  \sum_{\alpha_2,\ldots,\alpha_n\in\aleph_J^\bullet} \check\eta_1(x;\alpha_1,\ldots,\alpha_n)\prod_{m=2}^n F_m(\alpha_m) \right \vert\\
&= \sum_{x\in \mb L}\sup_{ \alpha_1\in \aleph_J^\bullet  }  \vert \alpha_1\vert^{-1}\bm\lambda_J^\bullet(\alpha_1 ) e^{-md_{\text{t}}(\alpha_1) }  e^{\tilde m d(x;\alpha_1)} \\
&\qquad\qquad\times \left\vert  \sum_{\alpha_2,\ldots,\alpha_n\in\aleph_J^\bullet} \sum_{y\in \mb L} \frac{\omega_1(y;\alpha_1,\alpha')}{\omega_1(\alpha_1,\alpha' )}  \check\eta_1(x;\alpha_1,\ldots,\alpha_n)\prod_{m=2}^n F_m(\alpha_m) \right \vert\\
&\leq \sum_{x\in\mb L} \sup_{\alpha_1\in \aleph_J^\bullet} \sum_{y\in\mb L} \sup_{\substack{\alpha_m\in \aleph_J^\bullet \\ m=2,\ldots,n}}  \vert \alpha_1\vert^{-1}\prod_{m=1}^n \left[\bm\lambda_J^\bullet(\alpha_m) e^{-md_{\text{t}}(\alpha_m) } \right] \\&\qquad\qquad\times e^{\tilde md(x;\alpha_1)-\mu d(y;\alpha_1)  } \exp\left(\sum_{\ell\in T}\mu d(\alpha_\ell)  \right)  \left\vert \check{\eta}_1(x;\alpha_1,\ldots,\alpha_k) \right\vert  \\&\qquad\qquad  \times     \sup_{\vartheta\in \mc B_{J,\bm\lambda}^\bullet(T')}\left\vert  \sum_{\alpha_2,\ldots,\alpha_n\in\aleph_J^\bullet} \vartheta(\alpha_2,\ldots,\alpha_n)\prod_{m=2}^n F_m(\alpha_m) \right \vert\\
&\leq \frac{c_J^\bullet(\bm\lambda)}{16} \sup_{\vartheta\in \mc B_{J,\bm\lambda}^\bullet(T') } \left\vert \sum_{\alpha_2,\ldots,\alpha_n\in\aleph_J^\bullet} \vartheta(\alpha_2,\ldots,\alpha_n)\prod_{m=2}^n F_m(\alpha_m) \right\vert
\end{align*}
where in the last step, we used that $\tilde md(x;\alpha_1) = \tilde md(x;\circ_m\alpha_m) $ in the support of $\check\eta_1 $, and that
\begin{align*}
\vert\alpha_1\vert^{-1}\sum_{y\in \mb L} e^{-\mu d(y;\alpha_1)  } &\leq  \vert\alpha_1\vert^{-1}\sum_{\substack{y\in \mb L \\ z\in \supp\alpha_1 }}  e^{-\mu d(y;z)  } \leq \frac{c_J^\bullet(\bm\lambda)}{16}.
\end{align*}
The bound on $\check\vartheta_2 $ is entirely analogous, using
\begin{align*}
\omega_2(y;\alpha,\alpha') &= \exp\left(-\tilde md(y;\alpha_1)-\mu d(y;\alpha_2)   \right).
\end{align*}

\end{proof}

\section{Applications}
In this section we apply Theorem \ref{thmabound} to the models described in section 1.

\subsection{Proof of Theorem \ref{mainthm}}
Set $\kappa_J=\infty$. Let $m_V, m,\lambda_J>0 $ be given and consider a symmetric, normal covariance $C\in\text{End }\mb C^{\mb L\times \N} $ with  
\begin{align*}
\lambda_{\text{min}}\big(\Re C^{-1}\big) =&\, \upmu>0 & \text{and}& & \Vert C\Vert_{6m,\infty} &= \mf c_{\infty}<\infty.
\end{align*}
Depending on these data, choose $R>r>0$, $\mf v_2\geq 0$, some fixed nondecreasing $\omega:[r,R]\to\mb R $ with $\omega(r)\leq 1 $, and some $c_v>0 $, according to conditions we are about to derive. Let $V_1(\phi;J)$ be analytic and satisfy (\ref{assv1}), and let $V_2 $ be a two body potential with source term that satisfies (\ref{assv2}). W.l.o.g., $V_1(0;J)=0 $.\\
Note first that 
\begin{align*}
\bm\vert V_1\bm\vert_{1,\dot\lambda} &\leq \sup_{r'\in[r,R]}\omega(r')^{-1} \Vert V_1\Vert_{r'+\mf v_1^{-1},\lambda_J,\dot m}\leq \frac{\mf v_1}{\dot\lambda_1}\\
\bm\vert V_2\bm\vert_{2,\dot\lambda} &\leq \tilde c_{2} \Vert v^\half\Vert^2_{2\dot m} \dot\lambda_2^{-4} + \vert a\vert \cdot  \lambda_J^{-1}\dot\lambda_2^{-1}\leq \tilde c_2 \left(\frac{\mf v_2^{\frac14}}{\dot\lambda_2}\right)^4 + \frac{\mf v_2^{\frac14}}{\dot\lambda_2}
\end{align*}
for some $\tilde c_{2} = \tilde c_{2}({\sf N}) $ whenever $\dot\lambda_J\geq \lambda_J $ and $\dot\lambda_1\geq \mf v_1 $. Let $c_2>0 $ be such that $$ \tilde c_2 c_2^4 + c_2 = 2^{-5}  c_{\Text{(\ref{smallness})}} . $$ By (\ref{assv2}), according to Remark \ref{rempositivity}, $V_2$ satisfies the positivity requirement (\ref{condpos}) for any $\dot\lambda_\phi\leq c_2\mf v_2^{\frac 14} $ with $$ c_{pos} = c_2c_v^2 \qqquad c_{pos}'=0. $$ Consider the norms introduced in Definitions \ref{defnorms} and \ref{altdef} for the following sequence of parameters $\bm\lambda $ as in Theorem \ref{thmabound}:
\begin{align*}
\bm\lambda &=  \Big(1,\lambda_J,1,1,1,1,\frac1{12},1,m,m_V,\infty,\infty\Big) \\
\bm\lambda\ob 4 = \Lambda^3_{\Text{iv}}\bm\lambda &= \Big(1,\lambda_J,1,1,1,1,\frac1{12},e^{-\frac\upmu8r^2},2m,m_V,\infty,\infty\Big)\\
\bm\lambda\ob 3 = \Lambda_{J}\bm\lambda\ob 4 &= \Big(1,\lambda_J,1,1,1,1,\frac1{26},e^{-\frac\upmu8r^2},2m,m_V,\infty,\infty\Big)\\
\bm\lambda\ob 2 = \Lambda^3_{\Text{ii}}\bm\lambda\ob 3 &= \Big(1,\lambda_J,1,1,1,1,\frac1{2^43^2},e^{-\frac\upmu8r^2},2m,m_V,\infty,\infty\Big)\\
\bm\lambda\ob 1 = \Lambda^2\bm\lambda\ob 2 &= \Big(\mf c_\infty^{-\half}\lambda\ob1_s,\lambda_J,\lambda\ob1_s,1,1,c_{\sf g}(2m)^{-1} ,2^{-16}c_{\Text{P.\ref{boundsintint}}}^{-1}   ,e^{-\frac\upmu8r^2},2m,m_V,\infty,\infty\Big)\\
\dot{\bm\lambda} = \Lambda_{\Text{P.\ref{boundsintint}}}\bm\lambda\ob 1 &= \Big(\dot\lambda_\phi,\lambda_J,\dot\lambda_s,\dot\lambda_1,\dot\lambda_2,1 ,1,1,2m+3m_V,m_V,\infty,\infty\Big)
\end{align*}
with
\begin{align*}
\lambda\ob1_s&=\frac{ \big[1\wedge c_{\Text{L.\ref{lemhadamard}}}^{-1}\big] }{2^{25}3^3 {\sf N}^2 c_{\sf g}(2m)^\half   c_{\Text{P.\ref{boundsintint}}}} \\
\dot\lambda_s &= \left[1+6\upmu_{\leq1}^{-1}{\lambda\ob1_s}^{-2}\frac{c_{\sf g}(m_V) }{c_{\sf g}(2m) }\right]^{-1}\\
\dot\lambda_1 &= \frac {c_{\Text{(\ref{smallness})}} \mf v_1}{32} \wedge \frac{\mf c_\infty^{-\half}\lambda\ob1_s}{2{\sf N}} \\
\dot\lambda_\phi = \dot\lambda_2 &=  c_2\mf v_2^{\frac14}\wedge \left[\frac{2{\sf N}}{\mf c_\infty^{-\half}\lambda\ob1_s}+6\upmu_{\leq1}^{-1}{\lambda\ob1_s}^{-2}\frac{c_{\sf g}(m_V) }{c_{\sf g}(2m) }\right]^{-1}
\end{align*}
According to Proposition \ref{propprop2}, the above choice of $\bm\lambda\ob1 $ is valid as long as 
\begin{align}\label{addconstr}
e^{\frac\upmu{16}r^2}    \geq  2^{25}3^3{\sf N}^2 c_{\sf g}(2m)^{\half} c_{\Text{P.\ref{boundsintint}}} c_{\Text{L.\ref{lemhadamard}}}  .
\end{align}
Set also $\tilde{\bm\lambda}_\bullet^\bullet = \bm\lambda_\bullet^\bullet $, except for the $\lambda_X $ component, where $\tilde\lambda\ob\bullet_X = 12\lambda\ob\bullet_X $ for $\bullet =  $ void, $4,3,2 $. It is now easy to check that all conditions of Theorem \ref{thmabound} (i) and (ii) are satisfied if:
\begin{enumerate}
\item In addition to (\ref{addconstr}), $r $ is big enough so that (\ref{poslocact}) and (\ref{stabassumr}) are satisfied.
\item $R $ is so that the stability condition (\ref{choicer}) is satisfied.
\item $\mf v_1 $ satisfies the smallness condition
\begin{align*}
\frac {c_{\Text{(\ref{smallness})}} \mf v_1}{32} \leq  \frac{\mf c_\infty^{-\half}\lambda\ob1_s}{2{\sf N}}
\end{align*}
\item $\mf v_2$ satisfies the smallness condition
\begin{align*}
c_2\mf v_2^{\frac14}\leq r\wedge  \left[\frac{2{\sf N}}{\mf c_\infty^{-\half}\lambda\ob1_s}+6\upmu_{\leq1}^{-1}{\lambda\ob1_s}^{-2}\frac{c_{\sf g}(m_V) }{c_{\sf g}(2m) }\right]^{-1}.
\end{align*}
\end{enumerate}
We conclude that $\log \mc Z(J) $ exists, is analytic in $J$, and satisfies
\begin{align*}
\bm\vert \log \mc Z\bm\vert_{\bm\lambda} \leq   1 + 12 \sup_{x\in\mb L}\Big\vert \log  \mu_{C(x,x)}(B_r)\Big\vert \leq 5.
\end{align*}
(We also conclude that $\log \mc Z(J) $ is nonperturbatively close to its pure small field approximation $\log \mc Z_s(J) $). Since $\bm\vert \log \mc Z\bm\vert_{\bm\lambda} = \Vert\log \mc Z \Vert_{\lambda_J,m} $, the Theorem follows.\qed

\subsection{Many Boson systems}

We now apply Theorem \ref{mainthm} to the situation of section 1.1. There, $\mb L = \theta\mb Z/\beta\mb Z\times \mb Z^D/L\mb Z^D $, ${\sf N}=2 $ (labeling real and imaginary parts of the complex field), and the propagator is given by $$ C =\left(\begin{array}{cc} {\sf C}+{\sf C}^T & i({\sf C}-{\sf C}^T) \\ -i({\sf C}-{\sf C}^T) & {\sf C}+{\sf C}^T  \end{array}\right) $$ with ${\sf C}(x,x') = \sum_{{\sf p}\in \beta\mb Z\times L\mb Z^D} {\sf C}(x-x'+{\sf p}) $ (use the Poisson summation formula) and 
\begin{align*}
{\sf C}( \tau,\bd x ) = (2\pi)^{-D-1}\int\limits_{\vert k_0\vert\leq \frac\pi\theta} \int\limits_{\vert \bd k\vert_\infty\leq \pi} \frac{e^{ik_0\tau+i\bd k\cdot\bd x}\ud k_0\ud\bd k}{2 e^{-\hat h(\bd k)+\mu} \sin^2\frac\theta 2 k_0 +  \left[1-e^{-\hat h(\bd k)+\mu}\right] + i e^{-\hat h(\bd k)+\mu}\sin\theta k_0}
\end{align*}
is real, but ${\sf C}(x)\neq {\sf C}(-x) $. $\mu<0 $, $\hat {\text{h}}(\bd k) $ is the Fourier transform of $\text{h} $, and, by assumption, $\hat {\text{h}}(\bd k) = \tilde h(\bd k)\sum_{i=1}^D (1-\cos k_i) $ with a uniformly positive and bounded $\tilde h $. Define the (anisotropic) metric $$d_{\tilde\upmu}\big((\tau,\bd x),(\tau',\bd x')\big) = \tilde\upmu \vert \tau-\tau'\vert + \tilde\upmu^\half \vert \bd x-\bd x'\vert_2.  $$ with $\tilde\upmu=e^{-\mu}-1 $. We easily derive the estimates\footnote{For the $\upmu $ - independent bound on $\lambda_{\text{min}}(\Re C) $, the $\sin^2\frac\theta 2 k_0 $ term is essential. Otherwise, we often use that the integral for ${\sf C}( \tau,\bd x ) $ is convergent even at $\mu=0 $.} (see below (\ref{cinfty}) for $c_\infty $ (\ref{cd}) for $c_d $)
\begin{align*}
\upmu &= 1-e^\mu & c_{\sf g}(m) &\leq c_1\cdot m^{-D-1}\cdot \tilde\upmu^{-1-\frac D2}  & c_{\sf g}'(a) &\leq c_2 \cdot a^{D+1}\tilde\upmu^{-1-\frac D2} \\
\Vert C\Vert_{6m_0,\infty} &\leq c_3  & \vert \mu_{C(x,x)}\vert &\leq c_4  & \lambda_{\text{min}}(\Re C)&\geq c_5 \\
c_\infty &\leq c_6  & c_d&\leq c_7
\end{align*}
for some constants $m_0,c_1,c_2,\ldots $ depending only on $\theta$ and $\text{h} $.\\
The interaction is given by
\begin{align*}
V_1(\phi;J) &= V(\phi) + \mc D(\phi) + \langle J,\phi\rangle,
\end{align*}
with an explicit quartic part $V(\phi) $ and a power series part $\mc D(\phi) $, see section 1.1. $V(\phi) $ is not stable in the sense required by Theorem \ref{mainthm}, but we are allowed to include it into $V_1 $ because we have the bounds (norms as in section 1.1, in particular with the standard torus distance)
\begin{align*}
\Vert V\Vert_{2R,2{\sf m}} &\leq c_8\cdot \tilde{\mf v}^{-4\epsilon} & \Vert \mc D\Vert_{2R,2{\sf m}} & \leq c_9 \tilde{\mf v}^{\half -8\epsilon},
\end{align*}
with $R=\tilde{\mf v}^{-\frac 14-\epsilon} $, for some small $\epsilon>0 $ (not explicitly chosen in \cite{BFKT4}) and  $\tilde{\mf v} = \theta\mf v $, $\mf v $ small enough. In particular, if $\lambda_J\leq \tilde{\mf v}^{2} $ (for convenience), we have (\ref{assv1}) with $d=d_{\tilde \upmu} $, $\dot m =  2{\sf m} \tilde\upmu^{-\half}$\footnote{For simplicity, we formulated Theorem \ref{mainthm} for isotropic metrics (same mass for all directions), so this choice for $\dot m $ suboptimal. }, $r = \tilde{\mf v}^\delta $, $\mf v_1 = \tilde{\mf v}^{\frac 14-\epsilon} $ and
\begin{align*}
\omega(r') &= c_{10}\cdot \tilde{\mf v}^{-4\epsilon}\left(\frac{\tilde{\mf v}^{-\frac 14+\epsilon}+r'}{\tilde{\mf v}^{-\frac 14+\epsilon}+\tilde{\mf v}^{-\frac 14-\epsilon}}\right)^4.
\end{align*}
In order to apply Theorem \ref{mainthm}, it has to be shown that the exponents $0<\delta<\frac14 $  can be chosen so that the stability condition (\ref{choicer}) is satisfied, $r $ is big enough so that (\ref{poslocact}), (\ref{stabassumr}) and (\ref{addconstr}) are satisfied for $\tilde{\mf v} $ small enough (clearly, $\omega(r)\leq 1 $ is already satisfied for $\delta<\frac14  $ and $\tilde{\mf v} $ small enough). According to the above, (\ref{choicer}) reads
\begin{align*}
\upmu \tilde{\mf v}^{-2\delta} \geq c_{11} \tilde{\mf v}^{-4\epsilon} \big(m_V^{-1}\log \tilde{\mf v} ^{-4\epsilon}\big)^{D+1} \tilde\upmu^{-1-\frac D2}.
\end{align*}
Set $m = m_0\wedge \frac13 {\sf m}\tilde\upmu^{-\half} $ and $m_V = \frac23{\sf m}\tilde\upmu^{-\half} -m.  $ Then the above is satisfied for $\delta>2\epsilon $ and $ \tilde{\mf v} $ small enough, depending on $\mu $. The conditions (\ref{poslocact}), (\ref{stabassumr}) and (\ref{addconstr}) read
\begin{align*}
e^{-\frac{\upmu}{4}\tilde{\mf v}^{-2\delta}} &\leq c_{12} \\
\tilde{\mf v}^{\delta}&\leq c_{13} \cdot \upmu^\half \cdot  \big[c_{\sf g}(2m)\Vert C\Vert_{6m,\infty}\big]^{-\frac12}\\
&\leq c_{14} \cdot \upmu^\half \cdot \big[m^{-D-1}\tilde\upmu^{-1-\frac D2} \big]^{-\frac12}\\
e^{-\frac{\upmu}{16}\tilde{\mf v}^{-2\delta}}  &\leq  c_{15}\big[m^{-D-1}\tilde\upmu^{-1-\frac D2} \big]^{-\frac12}
\end{align*}
These conditions can be satisfied for any $\delta>0 $, but the maximal choice $\delta = \frac14-\epsilon $ is optimal. The smallness condition on $\mf v_1 =  \tilde{\mf v}^{\frac 14-\epsilon}$ that is required in the assumptions of Theorem \ref{mainthm} is
\begin{align*}
\mf v_1 \leq   c_{16} \cdot  \big[m^{-D-1}\tilde\upmu^{-1-\frac D2} \big]^{-\frac12}.
\end{align*}
Applying Theorem \ref{mainthm}, we have proven

\begin{thm}
Let $\log \mc Z(J^*,J) $ of (\ref{genfctmb}) be the small field approximation to the coherent state generating function of the truncated correlations of the Bose gas, with kinetic energy $\Text{h} $, repulsive two body potential $v$ of strength $\mf v $ and decay rate ${\sf m} $, and at chemical potential $\mu<0 $ and inverse temperature $\beta $. (The approximation depends on two parameters $\theta,\epsilon>0 $ that are small, and $J^*,J\in\mb C^{\mb L} $ with $\mb L = \theta\mb Z/\beta\mb Z\times\mb Z^D/L\mb Z^D $). Suppose the following smallness conditions on $\mf v $ are satisfied:
\begin{align*}
\mf v &\leq \const \upmu^{\frac{2}{1-4\epsilon}}\min\left\{1,\left[m^{-D-1}\tilde\upmu^{-1-\frac D2} \right]^{-\frac2{1-4\epsilon}},\log \left[m^{-D-1}\tilde\upmu^{-1-\frac D2} \right]\right\} \\
\mf v&\leq \const \left[m^{-D-1}\tilde\upmu^{-1-\frac D2} \right]^{-\frac2{1-4\epsilon}}
\\
\mf v^{ \frac{1-12\epsilon}2   } \log \mf v ^{-4\epsilon}&\leq \const m_V^{-D-1}  \cdot \upmu \cdot \tilde\upmu^{1+\frac D2}
\end{align*}
with $\const $ small enough (depending on $\Text{h},\theta $ and $\sf m $),
\begin{align*}
\tilde\upmu &= e^{-\mu}-1 &\upmu &= 1-e^\mu&  m &= m_0\wedge \frac13 {\sf m}\tilde\upmu^{-\half} \text{  ($m_0 $ small enough)} & m_V &= \frac23{\sf m}\tilde\upmu^{-\half} -m;
\end{align*}
Then $\log \mc Z(J^*,J) $ exists, is analytic in $J^*,J$ for $\Vert J\Vert_\infty,\Vert J^*\Vert_\infty\leq \const \mf v^2  $, and its power series coefficients at $J=J^*=0 $ (i.e. the truncated correlations) satisfy
\begin{align*}
\sup_{x\in \mb L}\sum_{\substack{x_1,\ldots,x_n\\ x_1^*,\ldots,x_n^* \\x\in \{x_1,\ldots,x_n^*\}}} e^{md_{\tilde\upmu,t}(x_1,\ldots,x_n^*)} \vert C_n(x_1,\ldots,x_n;x_1^*,\ldots,x_n^*)\vert \leq \const
\end{align*}
for $d_{\tilde\upmu,t}(x_1,\ldots,x_n^*) $ the size, in the metric $\tilde\upmu \vert \tau-\tau'\vert + \tilde\upmu^\half \vert \bd x-\bd x'\vert_2 $, of a minimal tree with vertices $x_1,\ldots,x_n^*. $

\end{thm}

\begin{rem}
It is interesting to investigate the range of chemical potentials for which this Theorem is useful at a given interaction strength $\mf v $. For small $\mu $, the above conditions read
\begin{align*}
\mf v &\leq \const \vert\mu\vert^{\frac{2}{1-4\epsilon}}\cdot \vert\mu\vert^{\frac{2+D}{1-4\epsilon}}  \\
\mf v&\leq \const \vert\mu\vert^{\frac{2+D}{1-4\epsilon}}
\\
\mf v^{ \frac{1-12\epsilon}2 } \log \mf v ^{-4\epsilon}&\leq \const   \vert \mu\vert\cdot \vert \mu\vert^{1+\frac D2} m_V^{-D-1} =\const   \vert \mu\vert\cdot \vert \mu\vert^{\half} 
\end{align*}
The first condition is the dominant one, and gives the restriction $$\mu\leq -\const\mf v^{\frac{1-4\epsilon}{4+D}} . $$ The factors $ \vert\mu\vert^{\frac{2+D}{1-4\epsilon}}$ are really $[c_{\sf g}(m_0)\Vert C\Vert_{6m_0,\infty}\Vert]^{-\frac{2}{1-4\epsilon}} $. With some effort (further reducing the clarity of the proof), this could be replaced by the $1,\infty $ norm $\Vert C\Vert_{6m_0}^{-\frac{2}{1-4\epsilon}} $. In an ideal situation, one has $\Vert C\Vert_{6m_0}\sim \hat {\sf C}^{-1}(0)\sim \upmu^{-1} , $ with $\hat {\sf C}^{-1}$ the integrand of the Fourier integral defining $\sf C $. Similarly, incorporating anisotropic norms would improve $\vert \mu\vert^{1+\frac D2} m_V^{D+1} \to 1 $. The last two conditions then reduce to the intuitively optimal $\mu\leq -\const \mf v^{\frac{1-12\epsilon}{2}} $. The first condition (which is essentially due to (\ref{stabassumr}) and therefore related to our treatment of the oscillations of coherent states and the boundary terms due to characteristic functions) weakens this inequality to $\mu\leq -\const \mf v^{\frac{1-4\epsilon}{3}}. $

\erem
\end{rem}

\subsection{Unbounded spin systems}

Consider now the situation of section 1.2. With a similar and easier analysis as in the last section, we get from Theorem \ref{mainthm}:

\begin{thm}
Let $\log \mc Z(J) $ of (\ref{genfctuss}) be the generating functional of truncated correlations of an unbounded spin system on a finite lattice $\mb L$, with kinetic energy determined by a (possibly complex) covariance $C$, and with a repulsive two body interaction $v $. Let $m,\upmu,\mf v>0 $, and assume that
\begin{align*}
\lambda_{\Text{min}}(\Re C^{-1})&\geq \upmu &  \Vert C\Vert_{6m,\infty}&\leq \mf c_\infty\\
\lambda_{\Text{min}}(v)&\geq c_v\mf v &  \Vert v^\half \Vert_{10m}&\leq \mf v^\half.
\end{align*}
Then, if $c_J $ and $\mf v $ are small enough (depending on $m,\upmu,c_v $ and $c_\infty$), $\log \mc Z(J) $ exists, is analytic for $\Vert J\Vert_\infty c_J\mf v^{-\frac 14}$, and its power series coefficients (i.e. the truncated correlations) satisfy
\begin{align*}
\sup_{x\in\mb L}\sum_{\substack{ x_1,\ldots,x_n \\ x\in\{x_1,\ldots,x_n\} }} e^{md_t(x_1,\ldots,x_n)} \vert C_n(x_1,\ldots,x_n)\vert \leq\const
\end{align*}

\end{thm}

\begin{rem}
As in \cite{APS}, one might be interested in the asymptotics of the constant in the bound for $C_n $. For this, notice that
\begin{align*}
\mc Z(J) &= e^{-\half \langle J,C J\rangle} \int_{\Xi^{\mb L}} \prod_{x\in X}\ud \phi(x) e^{-\half \langle \phi,C^{-1}\phi\rangle + V(\phi + C^\half J) }
\end{align*}
This is true for real $C$ and follows for complex $C$ by analytic continuation, or by a Stokes argument similar to the ones in \cite{BFKT3}. It is easy to see that $V_2(\phi;J) = V(\phi+C^\half J ) $ still satisfies the assumptions (in particular the stability assumptions) needed for Theorem \ref{thmabound}. Going through these assumptions as before shows $\Vert \log e^{\half \langle J,C J\rangle}\mc Z(J)\Vert_{\lambda_J,m}<\infty $ for any $\lambda_J\geq c_J^{-1} \mf v^{\frac 14} $. This implies
\begin{align*}
\sup_{x\in\mb L}\sum_{\substack{ x_1,\ldots,x_n \\ x\in\{x_1,\ldots,x_n\} }} e^{md_t(x_1,\ldots,x_n)} \vert C_n(x_1,\ldots,x_n)\vert \leq c^n n!\mf v^{\frac n4}.
\end{align*}
for a constant $c$ independent of $n>2 $ and $\mf v$. In the case of a $2$ component $\phi^4 $ model on the torus with the discrete Laplacian as kinetic energy, the asymptotics $c^n n! \mf v^{\frac n2-1} $ have been conjectured in \cite{APS}, based on the tree level perturbation theory of the model. As far as optimal $\mf v $ behavior is concerned, the authors of that work could prove a bound of the type $c^n (n!)^2\mf v^{\frac n2-1} $. This bound could also be obtained from our results, by first using Theorem \ref{thmabound} (ii) to reduce $\mc Z $ to its small field approximation $\mc Z_s $ (with an optimal choice of $r$, the corresponding large field error has the asymptotics $c^nn!\mf v^{\frac n4} e^{-c'\mf v^{-\half}}\leq {c''}^n (n!)^2 \mf v^{\frac n2} $), and then analyzing the convergent perturbation theory for $\log \mc Z_s  $. We leave this to the interested reader.

\erem
\end{rem}

\begin{appendix}
\section{Notation}

Due to the lack in this field of standard notation (or even a standard reference on which this thesis could be based), we provide an extensive index of notation for the convenience of the reader.\\[15pt]
\textbf{General Mathematics.}\tabulinesep=3pt
\begin{longtabu}{|l|l|X|}\hline
%
$\chi(\cdots) $ &  & for some condition $\cdots $ is $1 $ if the condition is satisfied and zero otherwise.  \\  \hline  
%
%
$\dot\cup $ &   & Disjoint union, (almost) always of nonempty sets. \\ \hline
$M\dot\cap M'\neq\emptyset $ & Page 12 & $M\cap M'\neq\emptyset $ unless $M =\emptyset. $  \\ \hline
$\ul n $ &  & $=\{1,\ldots,n\} $ \\ \hline
$\sigma(C) $ &  & The spectrum of an endomorphism $C$ \\ \hline
$\lambda_{\Text{min}}(C) $ &  & Smallest eigenvalue of a self adjoint endomorphism. \\ \hline
${\sf P}(M) $ &   & set of all pairs of elements of $M$ (full graph on $M$).\\ \hline
$\{N_m\}_1^m,(N_m)_1^n,\{(x_m)_1^n\} $ &   & $= \{N_m,m=1,\ldots,n\} $ or $(N_1,\ldots, N_m) $. $\{(x_m)_1^n\},x_m\in \mb L ,$ is the equivalence class of $(x_m)_1^n $ in $\mb L^n/\mc S_n $ (i.e. an unordered sequence / multiset)\\ \hline
$\mc G(\{X_m\}_1^n),\mc G((X_m)_1^n) $ &   & The connectedness graph of $\{X_m\}_1^n $ ($(X_m)_1^n $). \\ \hline 
$\mf T(M),\mf F(M)  $ &   & Set of trees (forests) on $M$.   \\ \hline
$d^F(x) $ &   & Degree of $x$ in the forest $F$.  \\ \hline
$\mc T(X),\mc T^c(F;X) $ & Page 15/40   & $\mc T(X) $ is any minimal spanning tree of $X$. $ \mc T^c(F;X)$ is any minimal forest $F' $ on $X$ such that $F\dot\cup F'\in\mf T(X). $  \\ \hline
$\mc L(F),\mc L(F)\vert_2 $ &  Page 32 & $= \{(\ell,x),x\in\ell,\ell\in F\} $, the set of legs of a forest $F$. $\mc L(F)\vert_2=  \{ (x,(\ell,x)\in \mc L   ) \}  $ (an unordered sequence).  \\ \hline 
%
%
%
%
\end{longtabu}

\noindent
\textbf{Partitions, Interpolation.}
\begin{longtabu}{|l|l|X|}\hline
%
%
$\mc P(M),\mc P(x) $ & Page 10 & the set of partitions of $M$. Elements of $\mc P(M) $ are denoted $\{N_m\}_1^n=\mc P  $. For $x\in M $, $ \mc P(x)$ is the unique $N_m $ with $x\in N_m$. \\ \hline
$\mc P(s),\mc P(F) $ & Page 14 & The partition of $\mb L $ defined by an interpolation parameter $s\in[0,1]^{{\sf P}(\mb L)} $ or a forest $F\in\mf F(\mb L) $. \\ \hline
$ \mc M(M)$ & Page 11 & $=\{\{N_m\}_1^n,\vert N_m\vert\geq 2,\cup N_m=M \}. $ \\ \hline 
${2^{\mb L}}',\mc P'(\mb L) $ & Page 12 &  ${2^{\mb L}}' = \{(X,Q),\emptyset\neq Q\subset X\subset \mb L\} $, and $\{(X_m,Q_m)\}_1^n\in\mc P'(\mb L)  $ if $ (X_m,Q_m)\in {2^{\mb L}}'$ and $\{X_m\}_1^n\in\mc P(\mb L) $ \\ \hline
$\tilde2^{\mb L},\mc M ,\mc C $ & Page 12/13 &  ${\tilde 2}^{\mb L} = \{(Z,X,Q),\, Z\subset 2^{\mb L}, (X,Q)\in {2^{\mb L}}',Z\dot\cap X\neq \emptyset\} $. See page 12 for the definitions of $\mc M,\mc C $  \\ \hline
%
%
%
%
$C_{s} $ & Page 14 & Subscript $s $ denotes interpolation by $s $ (Hadamard product).  \\ \hline
\end{longtabu}

\noindent
\textbf{Fields, Coefficient systems.} 
\begin{longtabu}{|l|l|X|}\hline
${\sf N},\Xi,\Xi_{\mb C}  $ & Page 7 & Number of field components (real dimension of the target space $\Xi =\mb R^{\sf N} $ of the fields); $\Xi_{\mb C} = \mb C^{\sf N}. $ \\ \hline
$ \langle\phi,\psi\rangle  $ & Page 7 &  real Euclidean scalar product on $\Xi^{\mb L}. $ \\ \hline
$\xi,x,\zeta,z $ & Page 7 &  For $ \xi,\zeta\in \mb L\times\N$, $x,z $ denote their first component. See page 7 for this somewhat ambiguous convention. \\ \hline 
$\bd L,\bd L\vert_X $ & Page 7 & Space of unordered sequences $\bm \xi $ (of any finite length) of elements of $\mb L\times\N $ (resp. $X\times\N,X\subset \mb L $). \\ \hline 
$\supp\bm\xi $ & Page 7 &  $  = \{x_m\}_1^n$ for $\bm\xi = \{(\xi_m)_1^n\}; $ confer the above convention. \\ \hline 
$n(\bm \xi),n(\bm \xi,\xi)$ & Page 21 & $n(\bm \xi) $ is the length of the unordered sequence $\bm\xi $. $n(\bm \xi,\xi)  =  \vert \{m,\xi_m=\xi\}\vert$ for $\bm\xi=\{(\xi_m)_1^n\} $. \\ \hline 
$\phi(\bm\xi)  $ & Page 7 &  $= \phi(\xi_1)\cdots\phi(\xi_n) $ for $\bm\xi = \{(\xi_1,\ldots,\xi_n)\} $ \\ \hline
%
%
$\nabla_{\phi;\bm\xi},\nabla_{s,F},\nabla_{J,\bm\zeta} $ & Page 21 & Derivative with respect to the fields / the source. \\ \hline 
$s^F,\ud\bd s^F $ & Page 14 &  $s^F(\{x,y\}) =  \min\{s(\ell),\ell \Text{ on the }F\Text{ path linking }x,y\} $ is interpolation parameter defined by a forest $F$ with weights $s(\ell),\ell\in F $. $\ud\bd s^F $ is integration over these weights.   \\ \hline
$\aleph_\bullet^\bullet,\mc F_\bullet^\bullet,\mc S_\bullet^\bullet $ & Page 20 & Spaces of power series coefficients, of field variables, and of sets of dependence. \\ \hline 
$\alpha,\Psi,\varsigma $ & Page 21 & Generic name for elements of $\aleph_\bullet^\bullet,\mc F_\bullet^\bullet,\mc S_\bullet^\bullet . $   \\ \hline 
$\supp\alpha,\alpha\cap\alpha',\alpha\circ\alpha' $ & Page 23 & See page 23. \\ \hline 
$\check\aleph_\bullet^\bullet  $ & Page 40 & $ = \mb L\times \aleph_\bullet^\bullet $ . \\ \hline 
\end{longtabu}

\noindent
\textbf{Activities, Algebra.}
\begin{longtabu}{|l|l|X|}\hline
$A,\dot A,A' $ & Page 17/10 & $A'(X;\phi;s;J) $ is the intermediate activity before integrating out $\phi$. $A(X;J) $ is $A'$ after integration without large field decomposition. $\dot A $ is the normalized version of $A$ (see (\ref{mealg})). \\ \hline 
$A_s,B $ & Page 18 & The pure small field and large field activities.   \\ \hline 
$\mc Z(J),\mc Z_s(J) $ & Page 9/38 & The partition function and its pure small field approximation. \\ \hline 
$\mc L,\dot B $ & Page 13 & Intermediate activities in the Mayer resummation of the large field small field cluster expansion. \\ \hline 
$\bm\int\limits_T\bm\int\limits_{T,Q} $ & Page 17/18 & Integro Differential operators that integrate out scale $\kappa $, with or without large field decomposition. \\ \hline 
$\chi_Q,\chi_Q^c $ & Page 18 & Small/large field characteristic functions.  \\ \hline 
$r,R $ &  & Cutoff parameters for small and large fields \\ \hline 
\end{longtabu}

\noindent
\textbf{Norms.}
\begin{longtabu}{|l|l|X|}\hline
$ \sum_{\alpha\in\aleph_\bullet^\bullet}   $ & Page 21/22 & See page 21. \\ \hline 
$\Sigma_\bullet^\bullet $ & Page 40 & $:\mb C^{\check\aleph_\bullet^\bullet}\to \mb C^{ \aleph_\bullet^\bullet}, $ the sum over the first component. \\ \hline 
$\bm \lambda,\bm\Lambda $ & Page 40 & Array (space of arrays) of parameters which tune our norms. \\ \hline 
$\mc B_{\bullet,\bm\lambda}^\bullet,\check{\mc B}_{\bullet,\bm\lambda}^\bullet $ & Page 40 & $\mc B_{\bullet,\bm\lambda}^\bullet\subset \mb C^{\aleph_\bullet^\bullet},\check{\mc B}_{\bullet,\bm\lambda}^\bullet\subset \mb C^{\check\aleph_\bullet^\bullet}, $ sets of test functions used in the definition of the fundamental analyticity norms. See page 41 for an alternative definition of $ {\mc B}_{\phi,\bm\lambda}  $ \\ \hline 
$\mc B_{J,\bm\lambda}^\bullet(T),\check{\mc B}_{J,\bm\lambda}^\bullet(T) $ & Page 41 & $\mc B_{\bullet,\bm\lambda}^\bullet(T)\subset \mb C^{(\aleph_\bullet^\bullet)^k} $ (for $T\in\mf T(\ul k) $) a set of test functions as needed for Property 3 (i).  \\ \hline 
$\mc D_{\bullet,\bm\lambda}^\bullet  $ & Page 25 & $\subset \mc F_\bullet^\bullet $, sets of field / interpolation parameter configurations used in the definition of the fundamental analyticity norms.   \\ \hline 
$\mf G_{Q,\bm\lambda}  $ & Page 25 & ``large field regulators''. \\ \hline 
$\bm\vert \cdot\bm\vert_{\bullet,\bm\lambda}^\bullet,\bm\Vert\cdot\bm\Vert_{\phi,\bm\lambda} $ & Page 22 & The fundamental analyticity norms on our activities. \\ \hline 
$\Vert \cdot\Vert_{\bm\lambda} $ & Page 22 & Fundamental norms for coefficient functions. \\ \hline 
$\Vert \cdot\Vert_{m},\Vert \cdot\Vert_{m,\infty} $ & Page 8 & $1,\infty $ norm and $\infty $ norm with exponential weight of mass $m$ for kernels on $\mb L $ or $\mb L\times \N $ \\ \hline 
$\Vert \cdot\Vert_{R,\lambda_J\dot m},\Vert \cdot\Vert_{m,\infty} $ & Page 8 & $1,\infty $ + exponential tree decay norm for coefficients of power series. \\ \hline 
$d_t(\alpha),d(x;\alpha) $ & Page 40 & Tree size (resp. distance to $x$) of $\supp\alpha. $ Mind the exceptions for $\alpha\in\aleph_2. $\\ \hline 
$m,m_V,\tilde m,\mu $ & Page 40/41 & Components of $\bm\lambda $. The masses needed in the definition of the fundamental norms on test functions. \\ \hline 
$\lambda_\bullet $ & Page 40 & Weights needed in the definition of the fundamental norms on test functions.  \\ \hline 
$\bm\lambda_\bullet^\bullet(\alpha) $ & Page 40 & Appropriate combination of the weights $\lambda_\bullet  $ as needed in the definition of the fundamental norms on test functions.  \\ \hline 
\end{longtabu}

\noindent
\textbf{Bounds, Construction.}
\begin{longtabu}{|l|l|X|}\hline
$c_{\sf g}(m),c_{\sf g}'(a) $ & Page 6/7 & Geometric constants. For a standard $D$-dimensional lattice: $c_{\sf g}(m)\sim m^{-D} $ and $c_{\sf g}'(a)\sim a^D $ if $m<1,a>1 $.  \\ \hline 
$\omega(r') $ & Page 8 & Function controlling the size of the nonpolynomial part of the interaction on the small and large field regions. See in particular Remark \ref{remchoicer}.  \\ \hline
$c_{pos},c_{pos}' $ & Page 15 & Constants in the positivity assumption on the polynomial part of the interaction. See Remark \ref{rempositivity}.  \\ \hline 
$\Lambda_\bullet^\bullet $ & Page 23-25 & Transformations on the space of parameters as needed to formalize our way of bounding in Properties 1-3, or appearing in certain Propositions and Lemmas. See section 5.2 for concretizations.  \\ \hline 
$c_J^\bullet(\bm\lambda),c_\bullet^\bullet $ & Page 24/25 & $c_J^\bullet(\bm\lambda) $ is the constant from Property 3 (i). See Proposition \ref{propprop3} for its value. The names of other constants in the construction indicates their origin.  \\ \hline 
$\delta_T\eta $ & Page 24 & Test function needed to control the Mayer resummation.  \\ \hline 
$\gamma,\gamma' $ & Page 24 & Coefficient functions arising from the $\bm\int_T $ integro differential operators. Feature in Property 2. See page 33 for their definition. \\ \hline 
$\upmu_{\leq 1} $ & Page 25 & $=\upmu\wedge 1 $ \\ \hline 
$\mu_{C(x,x)}(B_r), \vert \mu_{C(x,x)}\vert $ & Page 33 & Small field normalization and total mass of the single site measure $ \mu_{C(x,x)}$. \\ \hline 
%
%
%
\end{longtabu}

\end{appendix}

\end{document}